\documentclass[12pt,oneside,a4paper]{book}

\usepackage[utf8]{inputenc}
\usepackage{graphicx}
\usepackage{color}
\graphicspath{ {images/} }
\usepackage{hyperref}
\hypersetup{
  bookmarksopen= true,
  colorlinks   = true, 
  urlcolor     = blue, 
  linkcolor    = blue, 
  citecolor   = red 
}

\usepackage{fancyhdr}

\usepackage[T1]{fontenc}
\usepackage{url}
\usepackage{ifthen}
\usepackage[cmex10]{amsmath} 
\usepackage{amssymb,mathtools}
\usepackage{multirow}
\usepackage{amsthm}
\usepackage{cuted}
\usepackage{bm}
\usepackage{braket}
\usepackage{comment}
\usepackage{paralist}

\usepackage[toc]{glossaries}

\usepackage{array}
\newcolumntype{M}[1]{>{\centering\arraybackslash}m{#1}}
\newcolumntype{N}{@{}m{0pt}@{}}

\usepackage{caption}
\DeclareCaptionLabelFormat{cont}{#1~#2\alph{ContinuedFloat}}
\usepackage[caption = false]{subfig}

\usepackage{framed}
\usepackage[framemethod=TikZ]{mdframed}

\usepackage{enumerate}

\usepackage{tikz}
\usetikzlibrary{shapes,snakes}
\usetikzlibrary{tikzmark}
\usetikzlibrary{fit,backgrounds}
\usetikzlibrary{matrix,decorations.pathreplacing,angles,quotes}
\usetikzlibrary{patterns}
\usetikzlibrary{calc}
\usetikzlibrary{positioning}
\tikzstyle{block} = [rectangle, draw, 
    minimum width=4em, text centered, rounded corners, minimum height=1em]
\tikzstyle{rblock} = [rectangle, draw, 
    minimum width=5em, text centered, minimum height=4em, rounded corners]   
\tikzstyle{rec} = [rectangle, draw]
\tikzstyle{line} = [draw, -latex]
\tikzstyle{circ} = [circle, draw, inner sep=0em, minimum width=2em,
     text centered]
\tikzstyle{circ2} = [circle, draw, inner sep=0em, minimum width=2.2em,
     text centered]

\tikzset{snake arrow/.style=
{
decorate,
decoration={snake,amplitude=.4mm,segment length=2mm,pre length=1mm, post length=1mm}}
}
\tikzset{snake arrow0/.style=
{
decorate,
decoration={snake,amplitude=.4mm,segment length=2mm,post length=1mm}}
}

\usetikzlibrary{intersections}

\usepackage{color}

\theoremstyle{plain}
\newtheorem{theorem}{Theorem}[section]
\newtheorem{definition}[theorem]{Definition}
\newtheorem{lemma}[theorem]{Lemma}
\newtheorem{proposition}[theorem]{Proposition}

\newtheorem{conjecture}[theorem]{Conjecture}

\newtheorem{postulate}[theorem]{Postulate}

\theoremstyle{definition}

\newtheorem{example}[theorem]{Example}

\newcommand{\cB}{\mathcal{B}}
\newcommand{\cC}{\mathcal{C}}
\newcommand{\cD}{\mathcal{D}}
\newcommand{\cE}{\mathcal{E}}
\newcommand{\cF}{\mathcal{F}}

\newcommand{\cH}{\mathcal{H}}

\newcommand{\cL}{\mathcal{L}}
\newcommand{\cM}{\mathcal{M}}
\newcommand{\cN}{\mathcal{N}}

\newcommand{\cP}{\mathcal{P}}

\newcommand{\cS}{\mathcal{S}}
\newcommand{\cT}{\mathcal{T}}

\newcommand{\cV}{\mathcal{V}}

\newcommand{\cX}{\mathcal{X}}

\newcommand{\Her}[1]{\cL_{\mathrm{H}}(#1)}
\newcommand{\Psd}[1]{\cL_{\mathrm{H}}^+(#1)}

\DeclareMathOperator{\Tr}{Tr}
\DeclareMathOperator{\id}{id}

\newcommand{\ketbra}[2]{\ket{#1}\!\bra{#2}}

\DeclareMathOperator{\conv}{conv}

\DeclareMathOperator{\sco}{sc} 
\DeclareMathOperator{\nege}{neg}

\title{
\textbf{The Diversity of Entanglement Structures with Self-Duality and Non-Orthogonal State Discrimination in General Probabilistic Theories}\\
\vspace{2.5em}
{\large A dissertation for the degree of\\
Doctor of Philosophy (Mathematical Science)
}
}

\author{
{\large Hayato Arai}\\[3.2em]
{\large  Graduate School of Mathematics}\\
{\large  Nagoya University}\\
\includegraphics[width=3cm]{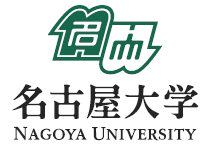}
}
\date{2022}

\begin{document}

\maketitle

\frontmatter

\chapter*{Abstract}

It is an important problem for mathematical physics to derive the mathematical model of quantum theory.
A modern approach, called General Probabilistic Theories (GPTs), starts from informational and operational aspects about states and measurements.
From a mathematical perspective, studies of GPTs aim to characterize proper positive cones by physically meaningful conditions.
Studies of GPTs have been widespread recently,
but they remain incomplete.

One of the essential problems of GPTs is to characterize the Standard Entanglement Structure (SES) from Entanglement Structures (ESs).
ES is a possible structure of a quantum composite system in GPTs.
It is strongly believed in standard quantum theory (or traditional physics)
that
a model of a composite system is uniquely determined as the SES.
However, a model of composite system in GPTs is not uniquely determined.
Moreover, because the definition of ESs is derived only from physically reasonable postulates,
ESs other than the SES are not denied as a model of some physical composite systems.
In other words, there is a theoretical possibility that some physical composite systems obey other ESs instead of the SES.
Therefore, it is an important problem to find (additional) reasonable postulates that determine various ESs as the standard one.
For this problem,
we need to investigate the diversity of ESs and characterize them by physically reasonable postulates.

In order to investigate the diversity of ESs,
this thesis considers a fundamental informational task called state discrimination.
State discrimination is a task whose success probability depends on the performance of measurements in models,
and the equivalent condition of perfect state discrimination is orthogonality of states in standard quantum theory.
On the other hand, the prior works \cite{Arai2019,YAH2020} have revealed that some ESs of GPTs have extraordinary performances for discrimination tasks,
i.e., some models enable to discriminate non-orthogonal states.
The preceding studies \cite{Arai2019,YAH2020} entirely depend on a certain type of concrete measurements beyond standard quantum theory.
This thesis then explores general measurements
and gives equivalent conditions for a given measurement to have a performance superior to standard quantum theory.
These equivalent conditions moreover give two important applications.
One is the derivation of the SES by the bound of the performance for discrimination tasks when we impose an additional condition.
Another one is non-simulability of beyond-quantum measurement.
These two applications mean that models in beyond-quantum measurement is distinguished from the SES by state discrimination.



Next, this thesis focuses on self-duality.
In GPTs, it is known that self-duality has an important role in characterizing Euclidean Jordan Algebras when considering the combination of self-duality and a kind of symmetric condition called homogeneity.
However,
it is an open problem how drastically just one of the above two conditions restricts models
even when considering ESs instead of all models.
Therefore, we first explore a group-symmetric condition in ESs.
As a result, we reveal that a group-symmetric condition weaker than homogeneity derives the SES uniquely from ESs.

On the other hand,
it is so challenging to explore self-dual ESs
that no example of self-dual ESs is known except for the SES.
Moreover, considering general models in GPTs,
only a few examples of self-dual and non-homogeneous models are known.
Then, in order to clarify the diversity of self-dual models,
this thesis develops a general theory of self-duality.
Applying the general theory,
we show the existence of self-dual ESs except for the SES.
Moreover, the general theory also gives an important classes of self-dual ESs called Pseudo Standard Entanglement Structures (PSESs).
A PSES is a self-dual ES that cannot be distinguished from the SES by physical experiments with small errors.
We show the existence of an infinite number of PSESs.
Furthermore, we also show that some of PSESs discriminate non-orthogonal states,
i.e.,
we show that some models enable non-orthogonal discrimination even though the ES is self-dual and near the SES.

\def\bibname{List of Publications\\[2em]

\textnormal{\normalsize 
	The thesis is based on the following publications.\\[-2em]}
}

This thesis contains all topics in \cite{AH2021} and some examples given in \cite{MAB2022}.

\chapter*{Acknowledgements}

First and foremost,
I am deeply grateful to my advisor Masahito Hayashi.
He is always willing to discuss my topics and provided his helpful knowledge and insights throughout my doctor course.
Also, he gave me many instructive advice and valuable experiences for my academic career.
Moreover, he pointed out some errors in my early proofs.
I also appreciate to another advisor Fran\c cois Le Gall for offering many advice about my plan of studies.
Next, I am grateful to my collaborators Shintaro Minagawa and Francesco Buscemi in Graduate school of Informatics, Nagoya University.
They actively discussed our topic and continuously worked with me for a long period.
I would like to thank Shintaro Minagawa for pointing out a small error in Section~\ref{appe:2-1} of this thesis.
I would also like to thank Francesco Buscemi for a good comment on this thesis especially about the definition of the example in Appendix~\ref{appe:2-3}.
Similarly, I would also like to thank Taro Nagao for many comments on this thesis.
Next, I thanks my colleagues in Graduate school of mathematics, Nagoya University.
I would especially thank Yuuya Yoshida and Seunghoan Song for answering many questions about doctor course in early period of doctor course.
Also, I thank Daiki Suruga, Ziyu Liu, and members of Le Gall's group for many active conversations about academic topics.
Finally,
my graduate career is supported by a JSPS Grant-in-Aids for JSPS Research Fellows No. JP22J14947, a Grant-in-Aid for JST SPRING No. JPMJSP2125, research assistant of a JSPS Grant-in-Aids for Scientific Research (B) Grant No. JP20H04139, and research assistant of Graduate school of Informatics, Nagoya University.

\mainmatter
\pagenumbering{arabic}

\cleardoublepage
\phantomsection
\addcontentsline{toc}{chapter}{Contents}
\tableofcontents
\cleardoublepage

\clearpage

\listoffigures

\clearpage

\listoftables

\begin{table*}[htb]
	\caption{Notations 1}
	\centering
	\begin{tabular}{clc}
	\hline
	notation & meaning & equation \\ \hline \hline
	$\top$ & transposition map from matrix to matrix & -\\
	$\Gamma$ & partial transposition map with $\Gamma=\id\otimes \top$ & -\\
	$\mathrm{ER}(\cC)$ & the set of all extremal rays of a proper cone $\cC$ & Def.~\ref{def:ray}\\
	$\mathrm{EP}(X)$ & the set of all extremal points of a convex set $X$ & -\\
	$\cC^\ast$ & the dual cone of a proper cone $\cC$&\eqref{def:dual-cone}\\
	$\cS(\cC,u)$ & the state space of the model $\cC$ with the unit $u$ &\eqref{def:state}\\
	$\cE(\cC,u)$ & the effect space of the model $\cC$ with the unit $u$  &\eqref{def:eff} \\
	\multirow{2}{*}{$\cM(\cC,u)$} & the measurement space of the model $\cC$  &\multirow{2}{*}{Def.~\ref{def:measurement}} \\
	&\multicolumn{1}{r}{ with the unit $u$ } & \\
	\multirow{2}{*}{$\cM_n(\cC,u)$} & the measurement space of the model $\cC$&\multirow{2}{*}{Def.~\ref{def:measurement}} \\
	&\multicolumn{1}{r}{ with the unit $u$ and $n$-outcome } & \\
	\multirow{2}{*}{$\Her{\cH}$} & the set of all Hermitian matrices &\multirow{2}{*}{-}\\
	&\multicolumn{1}{r}{on a Hilbert space $\cH$ } & \\
	\multirow{2}{*}{$\Psd{\cH}$} & the set of all Positive semi-definite matrices&\multirow{2}{*}{-}\\
	&\multicolumn{1}{r}{ on a Hilbert space $\cH$ }&\\
	$\cC_1\otimes\cC_2$ \quad& the tensor product of positive cones & \eqref{eq:tensor} \\ 
	$P_x$ \quad& the projetion map defined by $x$ & \eqref{def:projection} \\ 
	$\mathrm{SEP}(A;B)$\quad & the positive cone that has only separable states &\eqref{eq:sep}\\
	$\mathrm{SES}(A;B)$\quad & the standard entanglement structure &-\\
	\multirow{2}{*}{$\mathrm{Err}(\rho_1;\rho_2;\bm{M})$}\quad & the sum of error probability of discrimination  &\multirow{2}{*}{\eqref{def:error-1}}\\
	&\multicolumn{1}{r}{of $\rho_1$ and $\rho_2$ by a measurement $M$}  &\\
	\multirow{2}{*}{$\mathrm{Err}_{\cC}(\rho_1;\rho_2)$}\quad & the minimization of the sum of error probability& \multirow{2}{*}{\eqref{def:error-2}}\\
	&\multicolumn{1}{r}{of discrimination of $\rho_1$ and $\rho_2$} & \\
	$\rm{DOVM}(A;B)$\quad & the class of dual-operator-valued measures & \eqref{def:dovm}\\
	\multirow{2}{*}{$\lambda_k(X)$} \quad& the $k$-th eigenvalue of a Hermitian matrix $X$ & \multirow{2}{*}{-}\\
	&\multicolumn{1}{r}{in ascending order} & \\
	$\rm{BQ}(A;B)$\quad & the class of BQ measures on $\cH_A\otimes\cH_B$& Def.~\ref{def:bq}\\
	$\rm{AQ}(A;B)$\quad & the class of AQ measures on $\cH_A\otimes\cH_B$& Def.~\ref{def:aq}\\
	$\rm{BQ}(A;B)$\quad & the class of NAQ measures on $\cH_A\otimes\cH_B$& Def.~\ref{def:naq}\\
	$\rm{POVM}(A;B)$\quad & the class of POVM on $\cH_A\otimes\cH_B$ & Def.~\ref{def:povm}\\
	$\cD(X)$\quad & the domain of an operator $X$ & \eqref{def:domain-1}\\
	$\cD(\bm{M})$\quad & the domain of a measurement $\bm{M}$ & \eqref{def:domain-2}\\
	\hline
	\end{tabular}
\end{table*}

\clearpage

\begin{table*}[htb]
	\caption{Notations 2}
	\centering
	\begin{tabular}{clc}
	\hline
	notation & meaning & equation \\ \hline \hline
	$\mathrm{ME}(A;B)$\quad & the set of all maximally entangled states &-\\
	\multirow{2}{*}{$D(\cC\|\sigma)$\quad }& the distance between an ES $\cC$ &\multirow{2}{*}{\eqref{def:distance1}}\\
	&\multicolumn{1}{r}{  and a state $\sigma$}&\\
	$D(\cC_1\|\cC_2)$\quad & the distance between ESs $\cC_1$ and $\cC_2$ &\eqref{def:distance2}\\
	\multirow{2}{*}{$D(\cC)$\quad} & the distance between an ES $\cC$ &\multirow{2}{*}{\eqref{def:distance}}\\
	&\multicolumn{1}{r}{ and the SES} &\\
	$\mathrm{Aut}(\cC)$ & the set of automorphism on proper cone $\cC$ & \eqref{def:auto}\\
	$\bigoplus_{i=1}^k \cC_i$ & a direct sum over more than 1 positive cones & \eqref{def:direct-sum}\\
	$\mathrm{GU}(A;B)$\quad & the group of global unitary maps&\eqref{eq:gu}\\
	$\tilde{\cC}$\quad & a self-dual modification of pre-dual cone $\cC$ \quad &-\\
	\multirow{2}{*}{$\mathrm{MEOP}(A;B)$}\quad & the set of maximally entangled&\multirow{2}{*}{\eqref{eq:proj}}\\
	&\multicolumn{1}{r}{ orthogonal projections }&\\
	$\mathrm{NPM}_r(A;B)$\quad & a set of non-positive matrices &\eqref{def:NPM}\\
	$\cC_r(A;B)$\quad & a set of non-positive matrices with parameter $r$ &\eqref{def:Kr}\\
	$r_0(A;B)$\quad & the parameter given in Proposition~\ref{prop:construction1} &\eqref{def:r0}\\
	\multirow{2}{*}{$\cP_0(\vec{P})$\quad} & a family belonging to $\mathrm{MEOP}(A;B)$ &\multirow{2}{*}{\eqref{def:PE}}\\
	&\multicolumn{1}{r}{ defined by a vector $\vec{P}\in\mathrm{MEOP}(A;B)$ }&\\
	$\mathrm{LU}(A;B)$\quad & the group of local unitary maps&\eqref{eq:lu}\\
	\multirow{2}{*}{$N(r;\{E_k\})$\quad} & a non-positive matrix with a parameter $r\ge0$&\multirow{2}{*}{\eqref{def:Nr}}\\
	&\multicolumn{1}{r}{ and a family $\{E_k\}\in\mathrm{MEOP}(A;B)$ }&\\
	\hline
	\end{tabular}
\end{table*}

\chapter{Introduction}\label{chap:1}
 
This thesis addresses General Probabilistic Theories (GPTs) \cite{AH2021,MAB2022,PR1994,BBLW2007,Pawlowski2009,Short2010,Barnum2012,Plavala2017,Matsumoto2018,Takagi2019,Yoshida2020,CDP2010,Spekkens2007,KBBM2017,Stevens:2013,Barnum.Steering:2013,Janotta2014,Lami2017,Aubrun2020,Plavala2021,Arai2019,YAH2020,ALP2019,ALP2021,Kimura2010,Bae2016,Yoshida2021,Muller2013,Barnum2019,Janotta2013}.
The studies of GPTs aim to characterize models of physical theory, including quantum theory.

The mathematical model of quantum theory describes physical systems very precisely.
However, such consistency to physical systems is the almost only reason why the model of quantum theory is given as the present form.
Many researchers therefore have studied foundations of quantum theory \cite{PR1994,Neumann1932,Bell1964,Spekkens2005,DCPbook2017,Davies1970,Ozawa1984,Luders1951}.
Recently, as quantum information theory has been developed,
operational and informational aspects of quantum theory have been investigated.

The approach of GPTs is one of such frameworks for a foundation of quantum theory to discuss operational and informational aspects in mathematical models.
Because GPT only imposes fundamental postulates on models,
there exist continuously many models in GPTs.
The main aim of the studies of GPTs is to find reasonable postulates that derive the model of quantum theory from such models.
As seen in the next section,
many preceding studies have provided not only deep knowledge about a foundation of quantum theory \cite{PR1994,BBLW2007,Short2010,ALP2019,Barnum2019} but also many contributions to quantum information theory \cite{Pawlowski2009,Takagi2019,Stevens:2013}.
However, many problems are not solved completely, and this thesis tackles such problems.

In terms of derivation of the model of quantum theory,
it is an essential problem to distinguish quantum theory from models with similar structures to quantum theory.
An entanglement structure \cite{Janotta2014,Lami2017,Aubrun2020,Plavala2021,Arai2019,YAH2020,ALP2019,ALP2021} is a typical example of such similar models.
This thesis investigates how drastically some properties characterize entanglement structures
and reveal the diversity of entanglement structures.

\section{Concept and Motivation of General Probabilistic Theories}\label{sect:1-1}

A model of GPTs is given as a generalized model of quantum theory
with states and measurements.
In quantum theory,
a state is assigned information about the system,
and a measurement is a process to extract partially the information with a certain probability distribution.
Such probability distributions are applied to topics of quantum information theory.
As a generalization of quantum theory and quantum information theory,
GPT also starts with the following concept of measurement process on a state and information with probability distribution.

Apply a measurement $\bm{M}=\{M_i\}_{i\in I}$ to a given state $\rho$.
Then, only one outcome $i$ is obtained with a certain probability $p_i$ (Figure~\ref{figure_measurement}).
Here, we assume the probability $p_i$ is exactly given as a function of $\rho$ and $M_i$,
which corresponds to the empirical law that the same state and the same measurement give the same probability distribution.
As a prerequisite for dealing with probabilistic measurement processes,
a model of GPTs contains such a structure mathematically.
Also, as an operationally inevitable requirement,
we assume that the state space and measurement space are convex
because we can consider mixture objects with a certain probability.

\begin{figure}[t]
	\centering
	\includegraphics[width=6cm]{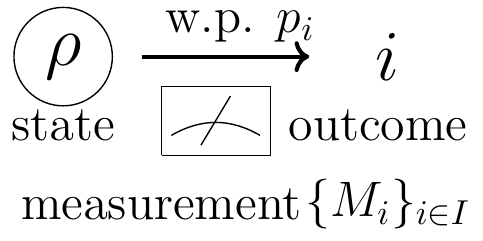}
	\caption[Measurement Processing]{
	When a state $\rho$ is measured by a measurement $\bm{M}=\{M_i\}_{i\in I}$,
	an outcome $i\in I$ is obtained with probability $p_i$.
	}
	\label{figure_measurement}
\end{figure}

As a minimum structure to deal with the above probabilistic measurement process,
a model of GPTs is defined by a proper positive cone in a real-vector space with an inner product (Definition~\ref{def:model}).
The rigorous definition is given in Section~\ref{sect:2-2-1},
but roughly speaking,
a model of GPTs is defined by each different proper positive cone.
This definition is mathematically so general that there exist continuously many different models in GPTs
even if the dimension of the vector space is finite.
The main motivation of the studies of GPTs is to derive the model of quantum theory from such many different models
by reasonable conditions about probability distributions of states and measurements.

\section{Approaches to Characterization of Physical Models}\label{sect:1-2}

Roughly speaking,
there are three approaches to characterize the traditional models of physical systems (especially quantum theory).
The first approach is to impose traditional postulates in theoretical physics or quantum information theory,
which is well-considered in early studies of GPTs.
A typical example is \textit{no-signaling principle}, a property in both classical and quantum theory.
It was believed that quantum theory is a model that has the most correlated causality with no-signaling principle.
However, the reference \cite{PR1994} found a model with no-signaling principle with more correlated causality than quantum theory.
Another typical example is \textit{no-cloning principle}, which is a property of quantum theory but not classical theory.
In early studies, no-cloning principle is considered a peculiar property in quantum theory.
However, the reference \cite{BBLW2007} showed that the availability of cloning processes is a peculiar property in classical theory.
In this way, traditional postulates often do not characterize quantum theory
(they sometimes characterize classical theory).

The second approach is to impose a limit of a performance for an information task.
Traditional postulates in the first approach are often qualitative.
In order to characterize quantum theory,
we need indicators more quantitative than traditional postulates.
A typical example of quantitative postulate is given as a limit of a performance for certain information tasks.
The reference \cite{Pawlowski2009} introduced a performance for a communication task,
and the reference \cite{Pawlowski2009} showed that the upper bound of the performance characterizes the amount of correlation in models.
Such an informational approach is not only accurate but also well-motivating not only for foundation of physics but also for practical physics.
An upper bound of a performance for information tasks regulates our operational limit to physical systems or information processing.
In other words,
such characterizations provide a rigorous statement why an unknown model cannot be implemented on any physical system.
Therefore,
many information tasks have been considered recently \cite{Pawlowski2009,Takagi2019,Yoshida2020,Lami2017,Aubrun2020,Arai2019,YAH2020,Kimura2010,Bae2016,Yoshida2021}.
However, a condition of performance for information tasks to characterize quantum theory has not been found yet.

The third approach focuses on mathematical properties.
Recently, the reference \cite{Barnum2019} showed that a strongly symmetric condition derives  \textit{essentially finite types of models},
including quantum theory.
Also, the reference \cite{KBBM2017} imposes the same condition and a condition about spectrality,
and it derives the thermodynamical behavior of quantum theory.
In this way, mathematically strong conditions determine models drastically.
However, such strong conditions do not often have reasonable physical meanings.
In order to give a physical meaning of such conditions,
we need to investigate the relationship between such mathematical conditions and performances for some informational tasks or physical operations.

This thesis deals with the second and the third approach for the problem in the next section.

\section{Non-uniqueness of Entanglement Structures}\label{sect:1-3}

For a foundation of quantum theory,
it is essential to distinguish the model of quantum theory from similar models.
In GPTs,
all models are roughly divided into two types;
the first type is a model with a finite number of extremal rays (Definition~\ref{def:ray}) and the second type is a model with an infinite number of extremal rays.
The model of quantum theory has an infinite number of extremal rays,
and therefore, it is more important to investigate such models.
One important class of such models is Entanglement Structures (ESs).
ES is a possible structure of a quantum composite system in GPTs \cite{Janotta2014,Lami2017,Aubrun2020,Plavala2021,Arai2019,YAH2020,ALP2019,ALP2021}.
Even though a model of a composite system is uniquely determined in standard quantum theory,
it is not uniquely determined in GPTs.
Because the definition of ES is physically reasonable,
there is a theoretical possibility of another ES in physics.
Therefore,
an investigation of ESs contributes not only to a foundation of quantum theory
but also to an evaluation of other possibilities of quantum-like physical models.

In order to investigate the diversity of ESs,
this thesis studies the following five theme.
\begin{enumerate}
	\item[A.] Characterization of Dual-Operator-Valued Measurement (Section~\ref{sect:3-2})
	\item[B.] Non-Simulability of Beyond Quantum Measurement (Section~\ref{sect:3-2-3})
	\item[C.] Entanglement Structures with Group Symmetry (Section~\ref{sect:4-1-2})
	\item[D.] Self-Dual Modification (Section~\ref{sect:4-2})
	\item[E.] Existence of PSES and Difference from the SES (Section~\ref{sect:4-3})
\end{enumerate}
These themes are related to both of the second and the third approaches in Section~\ref{sect:1-2}.
Here, we roughly divide them into two parts as follows in terms of the approaches. 

First,
corresponding to the second approach in Section~\ref{sect:1-2},
this thesis focuses on a fundamental informational task called state discrimination.
State discrimination is a task whose success probability depends on the performance of measurements in a given model,
and it is equivalent to orthogonality of states in standard quantum theory.
On the other hand, preceding studies \cite{Arai2019,YAH2020} revealed that some models of GPTs discriminate non-orthogonal states.
The preceding studies \cite{Arai2019,YAH2020} entirely depend on some usable concrete measurements beyond standard quantum theory.
As Theme~A,
this thesis explores in detail such measurements beyond standard quantum theory.
As a result,
we give an equivalent condition when a measurement has a performance superior to standard quantum theory.
By applying the equivalent condition to the characterization of ESs,
we additionally show that a bound of the performance for a discrimination task determines ESs with certain condition as the SES uniquely.
Also,
as Theme~B,
by applying the results about beyond-quantum measurements,
we discuss the simulability of beyond-quantum measurement in standard quantum theory and reveal its impossibility.

Second,
corresponding to the third approach in Section~\ref{sect:1-2},
this thesis focuses on two mathematical properties, \textit{self-duality} and \textit{homogeneity}.
Self-duality means the equivalence between the proper cone and its dual cone.
Homogeneity is a strongly symmetric condition about a group acting on the vector space.
If a proper cone satisfies self-duality and homogeneity,
the cone induces the structure of Euclidean Jordan Algebra.
Cones corresponding to Euclidean Jordan Algebra are classified into essentially finite types of cones, including the model of quantum theory.
This result is given by the references \cite{Jordan1934,Koecher1957} in pure mathematics,
but it has been open for a long time how drastically only one of the conditions determines a structure of proper positive cones.
Moreover, only a few examples of self-dual cones are known \cite{Levent2001}.
Other known examples of self-dual cones also satisfy homogeneity, and therefore, such examples are classified essentially finitely.
This thesis investigates the diversity of ESs with one of the above two conditions.
As Theme~C,
we show that a weaker condition about group symmetry is sufficient to derive the standard entanglement structure than homogeneity.
On the other hand,
as Theme~E,
we show that self-duality does not derive the standard entanglement structure
even if we impose an additional condition that the model cannot be distinguished from the standard entanglement structure by physical experiments.
In order to show the result about self-duality,
we give a general theory about self-dual cones as Theme~D.

\section{Outlines and Contributions}\label{sect:1-4}

First, we explain the outline of this thesis.
In Chapter~\ref{chap:2}, we give mathematical definitions and fundamental propositions about GPTs.
In Section~\ref{sect:2-1}, we give a definition and fundamental properties of positive cones.
In Section~\ref{sect:2-2} and Section~\ref{sect:2-3}, we introduce a model of GPTs and entanglement structures, respectively.
In Section~\ref{sect:2-3}, we give some important examples of entanglement structures.
Properties of the examples are written in Appendix~\ref{appe:2}.
In Chapter~\ref{chap:3}, we investigate state discrimination in ESs.
In Section~\ref{sect:3-1}, we introduce state discrimination and preceding studies about it.
In Section~\ref{sect:3-2}, as our Main result 1,
we characterize Dual-Operator-Valued Measurements (DOVMs) by the performance for state discrimination.
Also, as our Main result 2,
we give the setting of simulability of DOVMs and show its impossibility.
In Section~\ref{sect:3-4},
we give proofs about the results in Chapter~\ref{chap:3}.
In Chapter~\ref{chap:4},
we investigate the diversity of ESs with self-duality and group symmetry.
In Section~\ref{sect:4-1},
we introduce preceding studies about self-duality and homogeneity,
and we show that a weak symmetric condition determines the standard entanglement structures uniquely as the main result 3.
In Section~\ref{sect:4-2}, as our Main result 4,
we give a general theory about self-dual cones.
In Section~\ref{sect:4-3}, as our Main result 5,
we introduce Pseudo Standard Entanglement Structures (PSESs),
and we show the existence of an infinite number of PSESs and their performance for state discrimination.
In Section~\ref{sect:4-4},
we give proofs of the statements in Chapter~\ref{chap:4}.
In Chapter~\ref{chap:5},
we conclude this thesis.

Next, we mention the contributions of the results in this thesis.
About Theme~A, the Author has fully contributed to the all contents.
The Author's contributions to Theme~B are the first setting, results, and its proof.
The present setting of Theme~B (Definition~\ref{def:simulable}) is obtained by discussing with Prof. Masahito Hayashi.
Theme~C, Theme~D, and Theme~E are the main results of \cite{AH2021}.
The Author's contributions to Theme~C are the setting, results, and its proof.
The settings of Theme~C, Theme~D, and Theme~E especially Definition~\ref{def:pre-dual} and $\epsilon$-distinguishability (in Section~\ref{sect:4-3-1}) were obtained by discussing with Prof. Masahito Hayashi.
Also, in Section~\ref{sect:2-3} and Appendix~\ref{appe:2}, this thesis gives some examples of entanglement structures that play an important role in the reference \cite{MAB2022}.
The Author's main contributions in the reference \cite{MAB2022} are the mathematical parts and the construction of the examples in Appendix~\ref{appe:2}.

\chapter{Mathematical description of GPTs}\label{chap:2}

In this chapter,
we introduce GPTs and their mathematical description.

First, as a preliminary,
we introduce positive cones and their fundamental properties in Section~\ref{sect:2-1}.
Positive cones play an essential role to define a model of GPTs.
Roughly speaking,
a model of GPTs, especially ESs,
is determined by a proper positive cone.
Therefore, a study of GPTs is mathematically regarded as a characterization of certain proper positive cones.

Next,
we introduce a model of GPTs and give some important examples of models in Section~\ref{sect:2-2}.
A model of GPTs is defined by a proper positive cone in a real vector space.
In a model of GPTs, states, measurements, and other objects are defined by the structure of the positive cone and its dual cone.

Finally, we introduce a bipartite composite model of two submodels in GPTs and entanglement structures in Section~\ref{sect:2-3}.
In GPTs,
a model of the composite system is not uniquely determined even if the subsystems are equivalent.
An entanglement structure is defined as a possible composite model whose submodels are equivalent to quantum theory.

\section{Preliminaries}\label{sect:2-1}

In this section,
we enumerate properties about positive cones.
Section~\ref{sect:2-1-1} gives definitions and fundamental properties.
Section~\ref{sect:2-1-2} gives properties about set-operations over continuous indices.
Section~\ref{sect:2-1-1} gives properties about group symmetry.
Many properties are applied in whole of this thesis,
but some properties are not applied here.
We give such properties in order to explain the mathematical aspects about positive cones.

\subsection{Definition and Fundamental Properties of Positive Cones}\label{sect:2-1-1}

In this thesis,
we assume that any vector space $\cV$ is finite-dimensional.

First,
we define positive cones and proper cones.

\begin{definition}[Positive Cone and Proper Cone]
	Let $\cC$ be a subset of a finite-dimensional real vector space $\cV$.
	We say that $\cC$ is a positive cone
	if $\cC$ satisfies the following two conditions:
	\begin{enumerate}
		\item $rx\in\cC$ for any $r\ge0$ and any $x\in \cC$.
		\item $\cC$ is closed convex set with non-empty interior.
	\end{enumerate}
	Also, we say that a positive cone $\cC$ is proper
	if $\cC$ satisfies the equation $\cC\cap(-\cC)=\{0\}$\footnote{In some references, a proper cone is also called a pointed cone.}.
\end{definition}

Next, we define an extremal ray of a positive cone,
which is convenient for discussion of positive cones. 

\begin{definition}[Extremal Ray]\label{def:ray}
	We say that a subset $R\subset\cV$ is a ray
	if $R$ is written as
	\begin{align}
		R=\left\{rx\middle|r\ge0\right\}
	\end{align}
	for an element $x\in\cV$.
	Also, given a positive cone $\cC$,
	we say that a ray $R\subset\cC$ is an extremal ray of $\cC$
	if any convex decomposition of an arbitrary element $x\in R$ over any element $x_i\in\cC$ with
	\begin{align}
		x=\sum_i a_i x_i
	\end{align}
	implies $x_i\in R$.
\end{definition}

Hereinafter, we denote a set of all extremal rays of a proper cone $\cC$ as $\mathrm{ER}(\cC)$.
Similarly, given a convex set $X$,
we denote a set of all extremal points of the set $X$ as $\mathrm{EP}(X)$.

Here, we give another perspective of the structure of proper cones.
Given a proper cone $\cC$, we define a relation $\le_{\cC}$ as
\begin{align}
	x\le_{\cC} y \quad \Leftrightarrow \quad y-x\in\cC.
\end{align}
Then, the relation $\le_{\cC}$ is a partial order.
Actually, the following proposition holds.
\begin{proposition}[{\cite[Section 2.4.1]{BoydBook2004}}]
	Let $\cC$ be a proper positive cone.
	Then, the relation $\le_{\cC}$ is a partial order.
\end{proposition}
In this way,
the definition of proper cone is characterized by the structure of partial order.
In GPTs, a model is defined by a partial order which is derived from a proper cone.

Next,
we define an order unit of proper cone,
which is regarded as a normalized factor of a model of GPTs.
\begin{definition}[order unit]
	Given a proper cone $\cC$,
	we say that an element $u\in\cC^{\circ}$ is an order unit of $\cC$,
	where the set $X^\circ$ is denoted as interior of $X$.
\end{definition}
An order unit is characterized by the order relation defined by the proper cone.
\begin{proposition}\label{prop-unit1}
	Given a proper cone $\cC$ and an element $u\in\cC$,
	the following two conditions are equivalent:
	\begin{enumerate}
		\item $u$ is an order unit of $\cC$.
		\item For any element $x\in\cV$, there exists a natural number $n$
		such that $x\le_{\cC}nu$.
	\end{enumerate}
\end{proposition}

\begin{proof}[Proof of Proposition~\ref{prop-unit1}]
	\textbf{[STEP1]}
	(i)$\Rightarrow$(ii)
	
	Because $u$ is an order unit of $\cC$,
	there exists a parameter $\epsilon>0$ such that the epsilon ball $N_\epsilon(u)$ of $u$ satisfies $N_\epsilon(u)\subset \cC$.
	Therefore, any element $x\in\cV$ satisfies the relation $u-\frac{\epsilon}{||x||} x\in\cC$.
	Because of the definition of $\le_{\cC}$, we obtain the following inequalities:
	\begin{gather}
		u-\frac{\epsilon}{||x||} x\ge_{\cC}0\nonumber\\
		x\le_{\cC} \frac{||x||}{\epsilon}u.
	\end{gather}
	Hence, there exists a natural number $n$ such that $x\le_{\cC}nu$.
	\\
	
	\textbf{[STEP1]}
	(ii)$\Rightarrow$(i)
	
	Let $\{x_i\}_{i=1}^d$ be an orthonormal basis in $\cV$.
	For any $i$, there exists a natural number $n_i$ such that $x_i\le_{\cC}n_iu$.
	Then, consider the set $X:=\{x\in \cV\mid ||x||=1\}$.
	Any element $x\in X$ has a decomposition over $\{x_i\}$ as $x=\sum_i a_i x_i$,
	where $a_i$ is a real number satisfying $\sum_i |a_i|^2=1$.
	Therefore, we obtain the inequality $x\le_{\cC} \sum_i |a_i|^2 n_i u$.
	Here, the number $\sum_i |a_i|^2 n_i$ is bounded for any $x\in X$.
	Hence,
	there exists a natural number $n_0$
	such that $x\le_{\cC}n_0u$ for any $x\in X$.
	Therefore, the relation $u-\frac{1}{n_0} x\in\cC$ holds for any $||x||=1$,
	which implies that the $\frac{1}{n_0}$ ball of $u$ is contained by $\cC$.
	Thus, the element $u$ belongs to interior of $\cC$.
\end{proof}

Next,
we define the dual cone of a positive cone.
There are some parallel ways to define a dual cone.
In this thesis, we define the dual cone embedded onto the original vector space $\cV$ by the inner product.
\begin{definition}[Dual Cone]
	Given a positive $\cC$,
	we define its dual cone $\cC^\ast$ as
	\begin{align}\label{def:dual-cone}
		\cC^\ast:=\left\{x\in\cV\middle|\langle x,y\rangle\ge0 \ \forall y\in\cC\right\}.
	\end{align}
\end{definition}
Then, the following proposition ensures that a dual cone is a proper cone when the original cone is proper.
\begin{proposition}[{\cite[Section 2.6.1]{BoydBook2004}}]
	Given a proper positive $\cC$,
	the dual cone $\cC^\ast$ is proper positive cone.
	Also, given a proper positive $\cC$,
	the dual of dual cone is equal to the original cone,
	i.e.,
	the equation $\left(\cC^\ast\right)^\ast=\cC$ holds. 
\end{proposition}

About dual cones,
the following two propositions are very important for the proof of whole of this thesis.
\begin{proposition}[{\cite[Section 2.6.1]{BoydBook2004}}]\label{prop:dual-inc}
	Given two proper cones $\cC_1$ and $\cC_2$,
	the following two conditions are equivalent:
	\begin{enumerate}
		\item $\cC_1\subset\cC_2$.
		\item $\cC_2^\ast\subset\cC_1^\ast$.
	\end{enumerate}
\end{proposition}
\begin{proposition}\label{prop-dual}
	Given a proper cone $\cC_1$ and $\cC_2$,
	the following two equations hold:
	\begin{align}
		\left(\cC_1+\cC_2\right)^\ast&=\cC_1^\ast\cap\cC_2^\ast.\label{prop-dual1}\\
		\left(\cC_1\cap\cC_2\right)^\ast&=\cC_1^\ast+\cC_2^\ast.\label{prop-dual2}
	\end{align}
\end{proposition}

\begin{proof}[Proof of Proposition~\ref{prop-dual}]
	\textbf{[OUTLINE]}
	The equation \eqref{prop-dual2} is shown by the same way as \eqref{prop-dual1}.
	Then, we show only \eqref{prop-dual1} here.\\
	
	\textbf{[STEP1]}
	Proof of the inclusion relation $\left(\cC_1+\cC_2\right)^\ast\subset\cC_1^\ast\cap\cC_2^\ast$ of \eqref{prop-dual1}.
	
	Let $x$ be an element in $\left(\cC_1+\cC_2\right)^\ast$.
	Because two cones $\cC_1$ and $\cC_2$ contain the element $0$,
	arbitrary elements $y_1\in\cC_1$ and $y_2\in\cC_2$ satisfy
	$y_i\in\cC_1+\cC_2$ for $i=1,2$.
	Therefore, the elements satisfy $\langle x,y_i\rangle\ge0$ for $i=1,2$.
	Because $y_1,y_2$ are arbitrary,
	the element $x$ belongs to $\cC_1^\ast$ and $\cC_2^\ast$,
	which implies $x\in\cC_1^\ast\cap\cC_2^\ast$.\\

	\textbf{[STEP2]}
	Proof of the inclusion relation $\supset$ of \eqref{prop-dual1}.
	
	Let $x$ be an element in $\cC_1^\ast\cap\cC_2^\ast$.
	Therefore, arbitrary elements $y_1\in\cC_1$ and $y_2\in\cC_2$ satisfy
	$\langle x,y_i\rangle\ge0$ for $i=1,2$,
	which implies the inequality $\langle x,y_1+y_2\rangle\ge0$.
	Because $y_1,y_2$ are arbitrary,
	the element $x$ belongs to $\left(\cC_1+\cC_2\right)^\ast$.	
\end{proof}

\subsection{Properties about Uncountable Operation of Positive Cones}\label{sect:2-1-2}

First, the following proposition guarantees that the intersection of uncountable positive cones is also a positive cone.
\begin{proposition}\label{prop:def-cone1}
	Let $\{\cC_\lambda\}_{\lambda\in\Lambda}$ be a family of positive cones with an uncountably infinite set $\Lambda$.
	There exists a positive cone $\cC$ such that the relation $\cC_\lambda\supset\cC$ holds for any $\lambda\in\Lambda$.
	Then, the set $\bigcap_{\lambda\in\Lambda} \cC_\lambda$ is a positive cone,
	i.e.,
	$\bigcap_{\lambda\in\Lambda} \cC_\lambda$ satisfies the following three conditions:
	\begin{enumerate}[(i)]
		\item $\bigcap_{\lambda\in\Lambda} \cC_\lambda$ is closed and convex.
		\item $\bigcap_{\lambda\in\Lambda} \cC_\lambda$ has an inner point.
		\item $\bigcap_{\lambda\in\Lambda} \cC_\lambda\cap\left(-\bigcap_{\lambda\in\Lambda} \cC_\lambda\right)=\{0\}$.
	\end{enumerate}
\end{proposition}

\begin{proof}
	\textbf{[STEP1]}
	Proof of (i).
	
	Because $\cC_\lambda$ is a positive cone for any $\lambda\in\Lambda$,
	$\cC_\lambda$ is closed and convex.
	Because $\cC_\lambda$ is closed, $\bigcap_{\lambda\in\Lambda} \cC_\lambda$ is also closed.
	Take any two elements $x,y\in\left(\bigcap_{\lambda\in\Lambda} \cC_\lambda\right)$.
	Because $x$ and $y$ satisfy $x,y\in\cC_\lambda$ for any $\lambda$,
	$px+(1-p)y\in\cC_\lambda$ for any $p\in[0,1]$,
	which implies that $px+(1-p)y\in\left(\bigcap_{\lambda\in\Lambda} \cC_\lambda\right)$i.e.,
	$\bigcap_{\lambda\in\Lambda} \cC_\lambda$ is convex.
	\\
	
	\textbf{[STEP2]}
	Proof of (ii).
	
	Because the assumption $\cC_\lambda\supset\cC$ holds for any $\lambda\in\Lambda$,
	$\left(\bigcap_{\lambda\in\Lambda} \cC_\lambda\right)\supset\cC$ holds.
	Also, because $\cC$ is a positive cone, $\cC$ has an inner point,
	which also belongs to the interior of $\bigcap_{\lambda\in\Lambda} \cC_\lambda$.
	\\
	
	\textbf{[STEP3]}
	Proof of (iii).
	
	This is shown because the set $\bigcap_{\lambda\in\Lambda} \cC_\lambda\cap\left(-\bigcap_{\lambda\in\Lambda} \cC_\lambda\right)$ is written as
	\begin{align}
		\bigcap_{\lambda\in\Lambda} \cC_\lambda\cap\left(-\bigcap_{\lambda\in\Lambda} \cC_\lambda\right)
		=&\bigcap_{\lambda\in\Lambda} \left(\cC_\lambda\cap(-\cC_\lambda)\right)
		=\bigcap_{\lambda\in\Lambda} \{0\}=\{0\}.
	\end{align}
\end{proof}

Next,
we discuss the sum of uncountable sets.
We remark the definition of  the sum of uncountable sets.

\begin{definition}\label{def:cone-2}
	Let $\{X_\lambda\}_{\lambda\in\Lambda}$ be a family of sets $X_\lambda$ with an uncountably infinite set $\Lambda$.
	We define the set $\sum_{\lambda\in\Lambda} X_\lambda$ as
	\begin{align}
		&\sum_{\lambda\in\Lambda} X_\lambda
		:=\mathrm{Clo}\left(\left\{\sum_{i\in I} x_i\middle| x_i\in X_i, I\subset\Lambda \mbox{ is a finite subset set}\right\}\right),
	\end{align}
	where $\mathrm{Clo}(Y)$ is the closure of a set $Y$.
\end{definition}

Then, the following proposition guarantees that the sum of uncountable positive cones is also a positive cone.

\begin{proposition}\label{prop:def-cone2}
	Let $\{\cC_\lambda\}_{\lambda\in\Lambda}$ be a family of positive cones with an uncountably infinite set $\Lambda$.
	There exists a positive cone $\cC$ such that the positive cone $\cC_\lambda$ satisfies $\cC_\lambda\subset\cC$ for any $\lambda\in\Lambda$.
	Then, the set $\sum_{\lambda\in\Lambda} \cC_\lambda$ is a positive cone,
	i.e.,
	$\sum_{\lambda\in\Lambda} \cC_\lambda$ satisfies the following three conditions:
	\begin{enumerate}[(i)]
		\item $\sum_{\lambda\in\Lambda} \cC_\lambda$ is closed and convex.
		\item $\sum_{\lambda\in\Lambda} \cC_\lambda$ has an inner point.
		\item $\sum_{\lambda\in\Lambda} \cC_\lambda\cap\left(-\sum_{\lambda\in\Lambda} \cC_\lambda\right)=\{0\}$.
	\end{enumerate}
\end{proposition}

\begin{proof}
	\textbf{[STEP1]}
	Proof of (i).
	
	By the deinifion~\ref{def:cone-2}, $\sum_{\lambda\in\Lambda} \cC_\lambda$ is closed.
	Take any two elements $x,y\in\left(\sum_{\lambda\in\Lambda} \cC_\lambda\right)$.
	By the definition~\ref{def:cone-2}, the elements $x,y$ are written as $x=\lim_{n\to \infty} x_n$ and $y=\lim_{n\to \infty} y_n$,
	where $x_n,y_n\in\cC_n$.
	Therefore, the element $z_n(p)=px_n+(1-p)y_n$ belongs to $\cC_n$ for $p\in[0,1]$.
	Because $\lim_{n\to\infty} z_n(p)=px+(1-p)y$,
	the element $px+(1-p)y$ belongs to $\sum_{\lambda\in\Lambda} \cC_\lambda$.
	\\
	
	\textbf{[STEP2]}
	Proof of (ii).
	
	Because the inclusion relation $\cC_{\lambda_0}\subset\sum_{\lambda\in\Lambda} \cC_\lambda$ holds for any element $\lambda_0\in\Lambda$,
	$\sum_{\lambda\in\Lambda} \cC_\lambda$ has an inner point that is also an inner point of $\cC_{\lambda_0}$.
	\\
	
	\textbf{[STEP3]}
	Proof of (iii).
	
	By the definition~\ref{def:cone-2},
	the element $0$ belongs to $\sum_{\lambda\in\Lambda} \cC_\lambda$,
	which implies the relation $\{0\}\subset \sum_{\lambda\in\Lambda} \cC_\lambda\cap\left(-\sum_{\lambda\in\Lambda} \cC_\lambda\right)$.
	Then, we show $\{0\}\supset \sum_{\lambda\in\Lambda} \cC_\lambda\cap\left(-\sum_{\lambda\in\Lambda} \cC_\lambda\right)$ as follows.
	Because the assumption $\cC_\lambda\subset\cC$ holds for any $\lambda\in\Lambda$
	and because the positive cone $\cC$ is closed,
	the set $\sum_{\lambda\in\Lambda} \cC_\lambda$ defined as a closure satisfies the inclusion relation
	$\cC\supset\sum_{\lambda\in\Lambda} \cC_\lambda$.
	The inclusion relation $-\cC\supset-\sum_{\lambda\in\Lambda} \cC_\lambda$ also holds.
	Therefore, the following inclusion relation holds:
	\begin{align}
		&\sum_{\lambda\in\Lambda} \cC_\lambda\cap\left(-\sum_{\lambda\in\Lambda} \cC_\lambda\right)
		\subset\cC\cap(-\cC)=\{0\}.
	\end{align}
\end{proof}

Finally,
the following lemma gives a relation between an intersection and a sum over uncountable positive cones.

\begin{lemma}\label{lem:sum-cone}
	Let $\{\cC_\lambda\}_{\lambda\in\Lambda}$ be a family of positive cones with an uncountably infinite set $\Lambda$,
	and let $\cC_1,\cC_2$ be positive cones satisfying $\cC_1\subset\cC_\lambda\subset\cC_2$ for any $\lambda\in\Lambda$.
	Then, the dual cone $\left(\bigcap_{\lambda\in\Lambda} \cC_\lambda\right)^\ast$
	is given by
	$\sum_{\lambda\in\Lambda} \cC_\lambda^\ast$.
\end{lemma}

\begin{proof}[Proof of lemma~\ref{lem:sum-cone}]
	Because of the assumption $\cC_1\subset\cC_\lambda\subset\cC_2$,
	Proposition~\ref{prop:dual-inc} implies  $\cC_1^\ast\supset\cC_\lambda^\ast\supset\cC_2^\ast$.
	Therefore, proposition~\ref{prop:def-cone1} and proposition~\ref{prop:def-cone2} imply that
	the two sets
	$\left(\bigcap_{\lambda\in\Lambda} \cC_\lambda\right)^\ast$ and $\sum_{\lambda\in\Lambda} \cC_\lambda^\ast$ are positive cones.

	Now, we show the duality.
	Take an arbitrary element $y\in\bigcap_{\lambda\in\Lambda} \cC_\lambda$ and
	an arbitrary element  $x\in\sum_{\lambda\in\Lambda} \cC_\lambda^\ast$.
	Then, there exists a sequence $x_n\in\sum_{\lambda\in\Lambda} \cC_\lambda^\ast$ such that
	$\lim_{n\to\infty} x_n=x$ and $x_n$ is a finite sum of elements in $\cC_\lambda$ for each $n$.
	Because $x_n$ is a finite sum of elements in $\cC_\lambda$, the inequality $\langle x_n,y\rangle\ge0$ holds.
	Therefore,
	$x,y$ satisfies the following inequality:
	\begin{align}\label{eq:lem:sum-cone}
		\langle x,y\rangle=\lim_{n\to\infty} \langle x_n,y\rangle\ge\lim_{n\to\infty} 0=0,
	\end{align}
	and thus,
	we obtain the relation $y\in\left(\sum_{\lambda\in\Lambda} \cC_\lambda^\ast\right)^\ast$,
	i.e.,
	$\bigcap_{\lambda\in\Lambda} \cC_\lambda\subset\left(\sum_{\lambda\in\Lambda} \cC_\lambda^\ast\right)^\ast$.
	The opposite inclusion relation is shown as follows.
	For any $\lambda\in\Lambda$,
	the inclusion relation $\sum_{\lambda\in\Lambda} \cC_\lambda^\ast\supset\cC_\lambda^\ast$,
	and Proposition~\ref{prop:dual-inc} implies the inclusion relation
	$\left(\sum_{\lambda\in\Lambda} \cC_\lambda^\ast\right)^\ast\subset\left(\cC_\lambda^\ast\right)^\ast=\cC_\lambda$.
	This inclusion relation holds for any $\lambda\in\Lambda$,
	and therefore,
	we obtain the inclusion relation 
	$\left(\sum_{\lambda\in\Lambda} \cC_\lambda^\ast\right)^\ast\subset\bigcap_{\lambda\in\Lambda} \cC_\lambda$.
	As a result, we obtain $\bigcap_{\lambda\in\Lambda} \cC_\lambda=\left(\sum_{\lambda\in\Lambda} \cC_\lambda^\ast\right)^\ast$.
\end{proof}

\subsection{Group Actions on Positive Cone}\label{sect:2-1-3}

In this thesis,
we discuss group symmetry on positive cones.
In this thesis,
we mainly consider a subgroup $G$ of $\mathrm{GL}(\cV)$.

At first, we introduce the following symmetry (called \textit{$G$-symmetry}) for a set $X$ (or a positive cone $\cC$) under a subgroup $G$ of $\mathrm{GL}(\cV)$:
\begin{itemize}
	\item[$G$-symmetry] a set $X$ is $G$-symmetric $\Leftrightarrow$ $g(x)\in X$ for any $x \in X$ and any $g\in G$.
\end{itemize}
Also, we say that a set of families $\cX$ is $G$-symmetric
if any element $g\in G$ and any family $\{X_\lambda\}_{\lambda\in\Lambda}\in\cX$ satisfy
$\{g(X_\lambda)\}_{\lambda\in\Lambda}\in\cX$.

Next, we define the following condition about a subgroup $G\subset\mathrm{GL}(\cV)$.
\begin{itemize}
	\item[adjoint-closed group] $G$ is a closed subgroup of $\mathrm{GL}(\cV)$ including the identity map and $G$ satisfies $g^* \in G$ for any element $g \in G$.
\end{itemize}
\begin{lemma}\label{lem:g-closed}
	Let $G$ be an adjoint-closed subgroup of $\mathrm{GL}(\cV)$,
	and let $\cC$ be a $G$-symmetric positive cone.
	Then, $\cC^\ast$ also satisfies $G$-symmetry.
\end{lemma}

\begin{proof}[Proof of lemma~\ref{lem:g-closed}]
	Take arbitrary elements $x\in\cC^\ast$, $y\in\cC$, and $g\in G$.
	Because $G$ is adjoint-closed, the relation $g^\ast\in G$ holds.
	Also, because $\cC$ is $G$-symmetric, the relation $g^\ast(y)\in\cC$ holds.
	Therefore, we obtain the inequality
	\begin{align}
		\langle g(x),y \rangle=\langle x,g^\ast(y)\rangle\ge 0,
	\end{align}
	which implies the relation $g(x)\in\cC^\ast$.
\end{proof}
Simply speaking,
lemma~\ref{lem:g-closed} shows that
adjoint-closedness transmits $G$-symmetry from $\cC$ to $\cC^\ast$.

\section{A Model of GPTs}\label{sect:2-2}

In this section,
we give a definition and examples of models of GPTs.

\subsection{Definition of a Model of GPTs}\label{sect:2-2-1}

First,
we define a model of GPTs by a proper cone in real vector space.

\begin{definition}[A Model of GPTs]\label{def:model}
	A model of GPTs is defined by a tuple $\bm{G}=(\cV,\langle\ ,\ \rangle,\cC,u)$,
	where $(\cV, \langle\ ,\ \rangle)$, $\cC$, and $u$ are a real-vector space with inner product, a proper cone, and an order unit of $\cC^\ast$, respectively.
\end{definition}

Given a model of GPTs,
a state is defined as normalized element in the cone by order unit.

\begin{definition}[State Space of GPTs]
	Given a model of GPTs $\bm{G}=(\cV,\langle\ ,\ \rangle,\cC,u)$,
	the state space of $\bm{G}$ is defined as
	\begin{align}\label{def:state}
		\cS(\cC,u):=\left\{\rho\in\cV\middle|\langle\rho,u\rangle=1\right\}.
	\end{align}
	Here, we call an element $\rho\in\cS(\cC,u)$ a state of $\bm{G}$.
\end{definition}

\begin{proposition}[{\cite[special case of Section 2.3.2]{BoydBook2004}}]\label{prop:state-sp}
	Given a model $\bm{G}=(\cV,\langle\ ,\ \rangle,\cC,u)$,
	the state space $\cS(\cC,u)$ is convex.
\end{proposition}

Due to Proposition~\ref{prop:state-sp},
the state space $\cS(\cC,u)$ has extremal points.
Hereinafter,
we call an element $\mathrm{EP}(\cS(\cC,u))$ a pure state of $\bm{G}$.

Next,
we define effects and measurements.

\begin{definition}[Effect Space]
	Given a model  of GPTs $\bm{G}=(\cV,\langle\ ,\ \rangle,\cC,u)$,
	the effect space of $\bm{G}$ is defined as
	\begin{align}\label{def:eff}
		\cE(\cC,u):=\left\{E\in\cV\middle|E\ge_{\cC^\ast}0\right\}.
	\end{align}
	Here, we call an element $E\in\cE(\cC,u)$ an effect of $\bm{G}$.
	Also, we say that an effect $E$ is proper
	if $E$ satisfies $0\le_{\cC^\ast} E\le_{\cC^\ast}u$.
\end{definition}

\begin{definition}[Measurements of GPTs]\label{def:measurement}
	Given a model  of GPTs $\bm{G}=(\cV,\langle\ ,\ \rangle,\cC,u)$,
	we say that a family $\{M_i\}_{i\in I}$ is a measurement
	if $M_i\in\cE(\cC,u)$ and $\sum_{i\in I} M_i=u$.
	The index $i$ is called an outcome of the measurement.
	Here, we denote the set of all measurements as $\cM(\cC,u)$.
	Especially, we denote the set of measurements with $n$-number of outcomes as
	$\cM_n(\cC,u)$.
\end{definition}

Hereinafter,
we assume that the set $I$ is finite.
This assumption is also usual in GPTs
when the dimension of $\cV$ is finite.

Now, we give a proposition about the effect space and the measurement space.
\begin{proposition}
	Let $\bm{G}=(\cV,\langle\ ,\ \rangle,\cC,u)$ be a model of GPT.
	For any effect $E\in\cE(\cC,u)$,
	there exists a number $r>0$ such that $rE$ is proper effect
	and the family $\{rE,u-rE\}$ belongs to $\cM(\cC,u)$.
\end{proposition}
This proposition holds because $u$ is order unit of $\cC^\ast$ and Proposition~\ref{prop-unit1} holds.

Next, we describe the measurement processing.
Let us consider the case a state $\rho$ is measured by a measurement $\{M_i\}_{i\in I}$ (Figure~\ref{figure_measurement-2}).
Then, an outcome $i\in I$ is obtained with probability $\langle \rho,M_i\rangle$.
In this setting,
the family $\{\langle \rho,M_i\rangle\}_{i\in I}$ constitutes a probability distribution
because the above definitions imply the inequality 
\begin{align}\label{eq:prob-1}
	\langle \rho,M_i\rangle\ge0 \ (\forall i\in I)
\end{align}
and the equality
\begin{align}\label{eq:prob-2}
	\sum_{i\in I} \langle \rho,M_i\rangle
	=& \left\langle \rho,\left(\sum_{i\in I}M_i\right) \right\rangle
	= \langle \rho,u\rangle=1.
\end{align}
The definitions of states and measurements come from the postulate that
the family $\{\langle \rho,M_i\rangle\}_{i\in I}$ constitutes a probability distribution,
i.e.,
the family satisfies two relations \eqref{eq:prob-1} and \eqref{eq:prob-2},
and the mathematical structure of proper cone is a typical minimal structure to discuss such processing.

\begin{figure}[t]
	\centering
	\includegraphics[width=6cm]{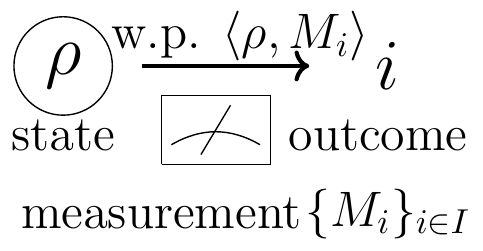}
	\caption[Measurement Processing]{
	When a state $\rho$ is measured by a measurement $\bm{M}=\{M_i\}_{i\in I}$,
	an outcome $i\in I$ is obtained with probability $\langle \rho,M_i\rangle$.
	}
	\label{figure_measurement-2}
\end{figure}

The above definitions are direct generalization of the model of classical and quantum theory
as seen in Section~\ref{sect:2-2-2}.

Finally, we define a transformation of GPTs.

\begin{definition}[Transformation of GPTs]\label{def:trans}
	Given a model  of GPTs $\bm{G}=(\cV,\langle\ ,\ \rangle,\cC,u)$,
	the transformation space of $\bm{G}$ is defined as
	\begin{align}
		\cT(\cC):=\left\{f\in\mathrm{GL}(\cV)\middle|f(\cC)=\cC\right\}.
	\end{align}
	Here, we call an element $f\in\cT(\cC)$ a transformation of $\bm{G}$.
	Also, we say that a transformation $f\in\cT(\cC)$ is a channel
	if $f$ satisfies $f(\cS(\cC,u))\subset\cS(\cC,u)$.
\end{definition}

In GPTs,
a transformation sometimes is regarded as a time-evolution of states.
In order to discuss dynamics on physical systems,
we need to deal with time-evolutions.
However,
Definition~\ref{def:trans} is not a direct generalization of time-evolution of traditional theory.
Therefore, some studies impose additional assumption for transformations (for example reversibility \cite{KBBM2017}) for the aim to deal with a transformation as a time-evolution.
Because this thesis does not aim to deal with time-evolution,
we apply Definition~\ref{def:trans} for convenience.

\subsection{Examples of Model of GPTs}\label{sect:2-2-2}

Next,
we give two examples of models of GPTs.

The first example is classical probability theory,
which corresponds to probabilistic structure of physical systems obeying classical theory.
\begin{example}[Classical probabilistic theory]
	Consider $\mathbb{R}^N$ as a vector space with the standard inner product $(\ , \ )$,
	take $\mathbb{R}^N_+:=\{(x_j)_{j=1}^N\in\mathbb{R}^N \mid x_j\ge0 \ (\forall j)\}$ and $\bm{e}=(1,1,\cdots,1)$,
	and consider the model $(\mathbb{R}^N,(\ ,\ ),\mathbb{R}^N_+,\bm{e})$.
	Then, the set of all states $\cS(\mathbb{R}^N_+,\bm{e})$ satisfies
	\begin{align}
		\cS(\mathbb{R}^N_+,\bm{e})=\{(x_j)_{j=1}^N\in\mathbb{R}^N \mid x_j\ge0 \ (\forall j), \ \sum_{j=1}^N x_j=1\}.
	\end{align}
	Therefore, $\cS(\mathbb{R}^N_+,\bm{e})$ is the set of all random variables on $\mathbb{R}^N$.
	Next, we consider measurements of $(\mathbb{R}^N,\mathbb{R}^N_+,\bm{e})$.
	Because $(\mathbb{R}^N_+)^\ast=\mathbb{R}^N_+$,
	the set of all measurements is given as follows:
	\begin{align}
		\cM(\mathbb{R}^N_+,\bm{e})=\left\{\{M_i\}_{i\in\Omega},M_i\in \mathbb{R}^N_+ \middle|
		\sum_{i=1}^N M_i=\bm{e}\right\}.
	\end{align}
	Therefore,  in this model, a measurement is equivalent to an event of the random variables.
\end{example}
For example, we regard the model $(\mathbb{R}^N, (\ ,\ ),\mathbb{R}^N_+,\bm{e})$ as the model of rolling dices with the case $N=6$.
First, a state $(x_j)_{j=1}^6$ corresponds to the skewed dice that the probability of the pip $j$ is given as $x_j$.
Next, we take $M_1=(1,0,1,0,1,0),M_2=(0,1,0,1,0,1)$, then the family $\{M_1,M_2\}$ is a measurement.
Now $\{M_1,M_2\}$ corresponds to the events of parity of pips,
that is, the probability of to roll odd pips is given as $(x_j)_{j=1}^6\cdot M_1=x_1+x_3+x_5$.
In this way, the model $(\mathbb{R}^N,(\ ,\ ),\mathbb{R}^N_+,\bm{e})$ corresponds to classical probabilistic theory.
Here, we simply call classical theory.

The second example is quantum theory.
\begin{example}[Quantum theory]
	Let $\cH$ be finite-dimensional Hilbert space, and let $\Her{\cH}$ be the set of all Hermitian matrices on $\cH$.
	Regard $\Her{\cH}$ as a real vector space with the inner product $\langle x,y\rangle:=\Tr xy$.
	Moreover, let $\Psd{\cH}$ be the set of all positive semi-definite matrices.
	Then we consider the model $(\Her{\cH},\Tr, \Psd{\cH},I)$, where $I$ is the identity matrices on $\cH$.
	First, the set of all states $\cS(\Psd{\cH},I)$ satisfies
	\begin{align}
		\cS(\Psd{\cH},I)=\{\rho\in\cT_+(\cH)\mid \Tr\rho=1\}.
	\end{align}
	Therefore, the set of all states $\cS(\Psd{\cH},I)$ is the set of all density matrices on $\cH$.
	Next, we consider measurements of $(\Her{\cH},\Tr, \Psd{\cH},I)$.
	Because $\left(\Psd{\cH}\right)^\ast=\Psd{\cH}$,
	the set of all measurements is given as follows:
	\begin{align}
		\cM(\Psd{\cH},I)=\left\{\{M_i\}_{i\in\Omega},\ M_i\in\Psd{\cH}\middle| \sum_i M_i=I\right\}.
	\end{align}
	Therefore, in this model, a measurement is equivalent to a Positive Operator Valued Measures (POVMs).
\end{example}
It is a standard definition of quantum theory whose states and measurements are defined as density matrices and POVMs, respectively.
In the perspective of their physical implementation,
density matrices and POVMs are available in a physical system.
In this way, the model $(\Her{\cH},\Tr, \Psd{\cH},I)$ corresponds to standard quantum theory.

Of course,
there are many other examples than classical and quantum theory.
Some of such models are seen in Section~\ref{sect:2-3}.

An important aim of studies of GPTs is to derive the above two examples.
For classical theory,
many operationally reasonable derivations has been found \cite{BBLW2007,ALP2019}.
On the other hand,
it is an open problem to derive quantum theory by reasonable postulates.
Especially, the essential problem is separation from entanglement structures as seen in Section~\ref{sect:2-3}.

\section{Entanglement Structures}\label{sect:2-3}

In this section,
we aim to define our main target, \textit{Entanglement Structures} (ESs).
For this aim,
we introduce a composite system in GPTs in Section~\ref{sect:2-3-1}.
Then, we give the definition of entanglement structure as a model of a typical composite system in GPTs in Section~\ref{sect:2-3-1}.

\subsection{Composite System in GPTs}\label{sect:2-3-1}

This thesis considers bipartite systems of two models $(\cV_A,\langle\,\ \rangle_A,\cC_A,u_A)$ and $(\cV_B,\langle\,\ \rangle_B,\cC_B,u_B)$.
A model of bipartite system is defined as follows.

\begin{definition}[Bipartite Model]\label{def:composite}
	Given two submodels $(\cV_A,\langle\,\ \rangle_A,\cC_A,u_A)$ and $(\cV_B,\langle\,\ \rangle_B,\cC_B,u_B)$,
	we say that a model $(\cV,\langle\,\ \rangle,\cC,u)$ is a model of bipartite composite system
	if the model satisfies the following conditions:
	\begin{align}
		\cV&=\cV_A\otimes\cV_B,\label{eq:com1}\\
		\langle x_1,x_2 \rangle&=\sum_{i,j}\langle a_1^{(i)},a_2^{(j)}\rangle_A \langle b_1^{(i)},b_2^{(j)}\rangle_B\nonumber\\
		&\mbox{for $x_1=\sum_i a_1^{(i)}\otimes b_1^{(i)}$ and $x_2=\sum_j a_2^{(j)}\otimes b_2^{(j)}$},\label{eq:com2}\\
		\cC_A\otimes\cC_B\subset&\cC\subset\left(\cC_A^\ast\otimes\cC_B^\ast\right)^\ast,\label{eq:composite}\\
		u&=u_A\otimes u_B\label{eq:com4},
	\end{align}
	where the tensor product $\cC_A\otimes\cC_B$ is defined as
	\begin{align}\label{eq:tensor}
		\cC_A\otimes\cC_B:=\{\sum_i a_i\otimes b_i \mid a_i\in\cC_A, b_i\in\cC_B\}.
	\end{align}
\end{definition}

As seen in Definition~\ref{def:composite},
a model of composite system is given by the tensor product.
The above conditions \eqref{eq:com1}, \eqref{eq:com2}, and \eqref{eq:com4} are not only natural
but also come from the following discussions.
Consider the case that Alice and Bob measure their local states independently.
On Alice's system, a measurement $\{e_A^{(i)}\}_{i\in I}\in\cM(\cC_A,u_A)$ affects a state $\rho_A\in\cS(\cC_A,u_A)$,
and also a measurement $\{e_B^{(j)}\}_{j\in J}\in\cM(\cC_B,u_B)$ affects a state $\rho_B\in\cS(\cC_B,u_B)$ on Bob's system.
There are $I\times J$ possibilities of obtained outcomes.
These measurements and states are not correlated,
and therefore, the possibility to get an outcome $(i,j)$ is given as
\begin{align}
	\langle \rho_A,e_A^{(i)}\rangle_A \langle \rho_B,e_B^{(j)}\rangle_B,
\end{align}
which is the same possibility that the product measurement $\{e_A^{(i)}\otimes e_B^{(j)}$ affects the product state $\rho_A\otimes\rho_B$ on a model of bipartite composite system
because of the conditions \eqref{eq:com2} and \eqref{eq:com4}.
In order to describe such an independent operation in tensor vector space,
we need the condition \eqref{eq:com1}.
Also, the reference \cite{Barrett2007,Janotta2014} shows that the condition \eqref{eq:com1} is derived from \textit{local tomography}, which states that any element in composite system is determined by only the joint probability of product measurements.
In this way, the above three conditions \eqref{eq:com1}, \eqref{eq:com2}, and \eqref{eq:com4} are reasonable for the minimal request about local operations.

On the other hand, because a positive cone is not a vector space (more strictly does not satisfies the universality of tensor product),
the condition \eqref{eq:composite} is not a trivial condition.
However, 
the definition of bipartite system in GPTs is so motivative that the condition \eqref{eq:composite} is derived from operational postulates.
Here, we give two ways to derive the definition of models of bipartite composite system.

The first way is  derived from availability of product elements.
\begin{postulate}[Availability of Product Elements \cite{Plavala2021}]\label{post:com1}
	In bipartite system, any product state is available,
	i.e.,
	the state $\rho=\rho_A\otimes\rho_B$ belongs to $\cS(\cC,u)$ for any $\rho_A\in\cS(\cC_A,u_A)$ and any $\rho_B\in\cS(\cC_B,u_B)$.
	Also, any product effect is available,
	i.e.,
	the effect $e=e_A\otimes e_B$ belongs to $\cE(\cC,u)$ for any $e_A\in\cE(\cC_A,u_A)$ and any $e_B\in\cE(\cC_B,u_B)$.
\end{postulate}
Postulate~\ref{post:com1} implies the condition \eqref{eq:composite} as follows.
Because $\cC$ is convex and because of the definition of positive cone,
the relation $\rho=\rho_A\otimes\rho_B\in\cS(\cC,u)$ implies the inclusion relation $\cC\supset\cC_A\otimes \cC_B$.
Similarly, the relation $e=e_A\otimes e_B\in\cE(\cC,u)$ implies the inclusion relation $\cC^\ast\supset \cC^\ast_A\otimes\cC^\ast_B$,
and we obtain $\cC\subset \left(\cC^\ast_A\otimes\cC^\ast_B\right)^\ast$ by apprying Proposition~\ref{prop:dual-inc} for the above inclusion relation.

For the second way to derive the condition~\eqref{eq:composite},
we define projection onto subsystem by elements.
\begin{definition}[Projection onto Subsystem by Elements]
	Define the projection onto $\cV_B$ by an element $x_A\in\cV_A$ as
	\begin{align}\label{def:projection}
	\begin{aligned}
		P_{x_A}:&\cV_A\otimes \cV_B\to \cV_B,\\
		P_{x_A}:&\sum_i \lambda_i a^{(i)}\otimes b^{(i)} \mapsto \sum_i\lambda_i\langle a^{(i)}, x_A\rangle_1 b^{(i)},
	\end{aligned}
	\end{align}
	where $a^{(i)}\in\cV_A$ and $b^{(i)}\in\cV_B$.
	Also, define the projection onto $\cV_A$ by the effect $x_B\in\cV_B$, similarly.
\end{definition}
By using the above projections,
we state the following postulate.
\begin{postulate}[Equivalence of Projections onto Subsystems]\label{post:com2}
	In bipartite system, any effect $e_A\in\cE(\cC_A,u_A)$ satisfies the equation $P_{e_A}(\cC)=\cC_B$.
	Also, any effect $e_B\in\cE(\cC_B,u_B)$ satisfies the equation $P_{e_B}(\cC)=\cC_A$.
	The same relations also hold for any states,
	i.e., any state $\rho_A\in\cS(\cC_A,u_A)$ satisfies the equation $P_{\rho_A}(\cC)=\cC_B$.
	Also, any state $\rho_B\in\cS(\cC_B,u_B)$ satisfies the equation $P_{\rho_B}(\cC)=\cC_A$.
\end{postulate}
Postulate~\ref{post:com2} also derives the condition~\eqref{eq:composite},
i.e., the following proposition holds.
\begin{proposition}\label{prop:def-com}
	If a model of bipartite system $(\cV,\langle\,\ \rangle,\cC,u)$ satisfies Postulate~\ref{post:com2},
	the inclusion relation~\eqref{eq:composite} holds.
\end{proposition}
\begin{proof}[Proof of Proposition~\ref{prop:def-com}]
	First, we prove the inclusion relation $\cC\subset\left(\cC_A^\ast\otimes\cC_B^\ast\right)^\ast$ by contradiction.
	Assume that there exists an element $x\in\cC$ such that $x\not\in\left(\cC_A^\ast\otimes\cC_B^\ast\right)^\ast$,
	and $x$ can be written as $\sum_i\lambda_i a^{(i)}\otimes b^{(i)}$, where $\lambda_i\in\mathbb{R}$, $a^{(i)}\in\cC_A$, $b^{(i)}\in\cC_B$.
	Because of the relation $x\not\in\left(\cC_A^\ast\otimes\cC_B^\ast\right)^\ast$,
	there exists a separable effect $e\in \cC_A^\ast\otimes\cC_B^\ast$ such that $\langle x,e\rangle<0$.
	Because $\cC_A^\ast\otimes\cC_B^\ast$ is spanned by product elements,
	we choose $e$ as the product element $e_A\otimes e_B$, where $e_A\in\cC^\ast$ and $e_B\in\cC^\ast$,  without loss of generality.
	However, we obtain the following inequality:
	\begin{align}
		0&>\langle x, e_A\otimes e_B\rangle
		=\left\langle \sum_i \lambda_i a^{(i)}\otimes b^{(i)}, e_A\otimes e_B\right\rangle
		=\sum_i\lambda_i \langle a^{(i)},e_A\rangle_A \langle b^{(i)},e_B\rangle_B\nonumber\\
		&=\left\langle \sum_i\lambda_i \langle a^{(i)},e_A\rangle_A b^{(i)},e_B\right\rangle_B
		=\langle P_{e_A}(x),e_B\rangle_B.
	\end{align}
	This inequality implies $P_{e_A}(x)\not\in\cC_B$,
	which contradicts to the assumption.
	
	The opposite inclusion relation $\cC_A\otimes\cC_B\subset\cC$ is similarly shown by the equation $P_{\rho_A}(\cC)=\cC_B$ and $P_{\rho_B}(\cC)=\cC_A$.
\end{proof}

Both Postulate~\ref{post:com1} and Postulate~\ref{post:com2} are reasonable requests about local operations.
Then, a model of bipartite composite system is defined as Definition~\ref{def:composite}.
Because of the condition~\eqref{eq:composite},
a model of bipartite composite system is not uniquely determined from submodels in general,
which is most important fact for this thesis.
As seen before,
no-correlated operation corresponds to tensor product.
Especially, an element is called \textit{entangled} if the element cannot be written as any convex combination of tensor product elements.
In quantum information theory,
entangled elements are main resource for whole of informational tasks.
The cone $\cC$ of a model of a composite system (in Definition~\ref{def:composite}) rules the diversity of entangled elements,
i.e.,
the cone $\cC$ determines the limit of available resources in the system.
This is the reason why the above non-uniqueness of $\cC$ is important.

On the other hand,
it is empirically known that classical theory does not includes entanglement elements,
which is shown by the above definition of composite system in GPTs.
Let us consider a model of two classical-subsystems $(\mathbb{R}^{N_A},(\ ,\ ),\mathbb{R}^{N_A}_+,\bm{e})$ and $(\mathbb{R}^{N_B},(\ ,\ ),\mathbb{R}^{N_B}_+,\bm{e})$.
The left-hand-side of the inclusion relation \eqref{eq:composite} is given as
\begin{align}\label{eq:com-classical}
	\mathbb{R}^{N_A}_+\otimes \mathbb{R}^{N_B}_+=\mathbb{R}^{N_AN_B}_+.
\end{align}
Also, because of the equation $(\mathbb{R}^{N}_+)^\ast=\mathbb{R}^{N}_+$,
the right-hand-side of the inclusion relation \eqref{eq:composite} is given as
\begin{align}
	\left((\mathbb{R}^{N_A}_+)^\ast\otimes (\mathbb{R}^{N_B}_+)^\ast\right)^\ast
	=&\left(\mathbb{R}^{N_A}_+\otimes \mathbb{R}^{N_B}_+\right)^\ast
	=(\mathbb{R}^{N_AN_B}_+)^\ast=\mathbb{R}^{N_AN_B}_+.
\end{align}
Therefore, the both sides of inclusion relation \eqref{eq:composite} are equivalent.
In other words,
The model of composite system of two classical-subsystems is uniquely determined as
classical theory on large system,
which implies that the model of classical composite system has no entangled elements.
Moreover,
the reference \cite{ALP2019} shows that such no-entanglement property derives classical theory uniquely,
i.e.,
the left-hand-side and the right-hand side in \eqref{eq:composite} are equal if and only if
one of cones $\cC_A$ or $\cC_B$ is equal to $\mathbb{R}^{N}_+$ for some $N$.

As the above discussion,
the model of composite system of classical theory is naturally determined as classical theory on a large system.
On the other hand,
in the case of quantum theory,
the model of composite system is not unique.
This thesis addresses this problem
and investigate the diversity of the model of quantum composite systems.

\subsection{Diversity of Entanglement Structures}\label{sect:2-3-2}

In this thesis,
we consider models of bipartite composite system of two quantum subsystems.
Let $(\Her{\cH_A},\Tr, \Psd{\cH_A},I)$ and $(\Her{\cH_B},\Tr, \Psd{\cH_B},I)$ be models of quantum theory on Alice's system and Bob's system.
In this case,
a model of bipartite composite system is given as $(\Her{\cH_A\otimes\cH_B},\Tr, \cC,I)$
satisfying the condition \eqref{eq:composite}.
Because of the equation $(\Psd{\cH})^\ast=\Psd{\cH}$,
the condition \eqref{eq:composite} is modified as
\begin{align}\label{eq:quantum}
	\mathrm{SEP}(A;B)\subset\cC\subset\mathrm{SEP}(A;B)^\ast,
\end{align}
where the proper cone $\mathrm{SEP}(A;B)$ is defined as
\begin{align}\label{eq:sep}
	\mathrm{SEP}(A;B):=\Psd{\cH_A}\otimes\Psd{\cH_B}.
\end{align}
In this thesis,
the model $(\Her{\cH},\Tr, \cC,I)$ satisfying \eqref{eq:composite} is called an \textit{Entanglement Structure} (ES),
and we denote the model as $\cC$ by omitting other objects for simplicity.
Of course, there is the diversity of ESs.
In other words,
an entanglement structure is not uniquely determined by the postulates in Section~\ref{sect:2-3-1}.

On the other hand,
it is strongly believed that physical systems obey the model $(\Her{\cH},\Tr, \Psd{\cH},I)$.
In bipartite composite system,
the model is defined as the cone $\Psd{\cH_A\otimes\cH_B}$,
which is neither the smallest one nor the largest one in \eqref{eq:quantum}.
Moreover, it is not completely clarified how the entanglement structure $\Psd{\cH_A\otimes\cH_B}$ is derived.
In this thesis, we call the entanglement structure $\Psd{\cH_A\otimes\cH_B}$ the \textit{Standard Entanglement Structure} (SES),
and we denote the model $\Psd{\cH_A\otimes\cH_B}$ as $\mathrm{SES}(A;B)$.
Our interest is the question what uniquely determines ESs as the SES.

The diversity of ESs is so large that some ESs are counterexample to some important mathematical properties.
For example, the model $\mathrm{SEP}(A;B)$ is a typical example that does not satisfy entropy preserving spectrality \cite{MAB2022}.
Also, an ES satisfies 1-symmetry but does not satisfy 2-symmetry \cite{MAB2022}.
We explain the details of the above two examples and its importance in Appendix~\ref{appe:2}.
The SES satisfies the above mathematical structures,
and therefore, the class of ESs contains various models separated from the SES.
On the other hand,
there exist many near ESs to the SES, called Pseudo Standard Entanglement Structures (PSES)s,
which are introduced in Section~\ref{sect:4-4}.
In this way,
there are variable types of ESs,
and therefore, the derivation of the SES is important and difficult problem in GPTs.

\chapter{State Discrimination in GPTs}\label{chap:3}

In this chapter,
we investigate the performance for state discrimination tasks in ESs.
State discrimination is a fundamental information task considered in quantum information theory \cite{HayashiBook2017,Chernoff1952,Holevo1972,Helstrom1979,HiaiPetz1991,OgawaNagaoka2000,OgawaHayashi2002}.
Because some performance for many other information tasks is derived from the performance for state discrimination,
it is important to investigate the performance for state discrimination.
Therefore, the performance for discrimination tasks is one of candidates to derive the SES from ESs.
In this thesis, we investigate how drastically an extraordinary performance for discrimination tasks determines ESs.

First, 
we introduce state discrimination and its performance in Section~\ref{sect:3-1}.
In quantum information theory,
there are many different types of discrimination tasks.
This thesis mainly discusses two types of them (Definition~\ref{def:dist-per} and Definition~\ref{def:dist-min}).
Also, we review important preceding studies about state discrimination in Section~\ref{sect:3-1}.
The two types of discrimination tasks has been studied very well in quantum theory \cite{HayashiBook2017,Chernoff1952,Holevo1972,Helstrom1979,HiaiPetz1991,OgawaNagaoka2000,OgawaHayashi2002,Chefles2004}.
Especially, some types have been studied in GPTs\cite{Kimura2010,Bae2016,Yoshida2021} and certain ESs \cite{Arai2019,YAH2020}.

Second,
as Theme A,
we investigate the performance for discrimination tasks of Dual-Operator-Valued Measurements (DOVMs) in Section~\ref{sect:3-2}.
The preceding studies \cite{Arai2019,YAH2020} showed that DOVMs with a certain form has extraordinary performance for the discrimination task of Definition~\ref{def:dist-per}.
This thesis investigates the performance of general DOVMs,
and we give equivalent conditions when a DOVM has extraordinary performance for discrimination tasks.
Furthermore, we show how drastically this equivalent conditions determine ESs.

Third,
as Theme B,
this thesis discusses simulability of DOVMs in Section~\ref{sect:3-2-3}.
A DOVM (especially non POVM) cannot be simulated in standard quantum theory when the dimension of the system is equivalent.
However, there is a possibility to simulate a given DOVM in a high-dimensional system.
In this section,
we show that a certain class of DOVMs cannot be simulated in any high dimensional quantum system as an application of the result in Section~\ref{sect:3-2-3}.

The proofs of statements in this chapter are written in Section~\ref{sect:3-4}.

\section{Introduction and Preceding Studies}\label{sect:3-1}

In this section,
we introduce state discrimination and briefly review preceding studies about state discrimination.
In Section~\ref{sect:3-1-1},
we define two types of discrimination tasks (Definition~\ref{def:dist-per} and Definition~\ref{def:dist-min}),
and see their performance in standard quantum theory.
In Section~\ref{sect:3-1-2},
we explain our preceding studies about the performance for discrimination tasks in certain ESs,
and we point out the reason why the preceding studies restrict to certain ESs.

\subsection{General Definition of Discrimination Tasks}\label{sect:3-1-1}

In discrimination tasks,
a player is given a number of candidates $\{\rho_i\}_{i=1}^n$ of unknown state $\rho$.
The player apply arbitrary one-shot measurement $\{M_i\}_{i=1}^n$ and obtain an outcome $i$.
The player identifies the unknown state $\rho$ from the outcome $i$ with large probability
(Figure~\ref{figure-discrimination}).

\begin{figure}[t]
	\centering
	\includegraphics[width=10cm]{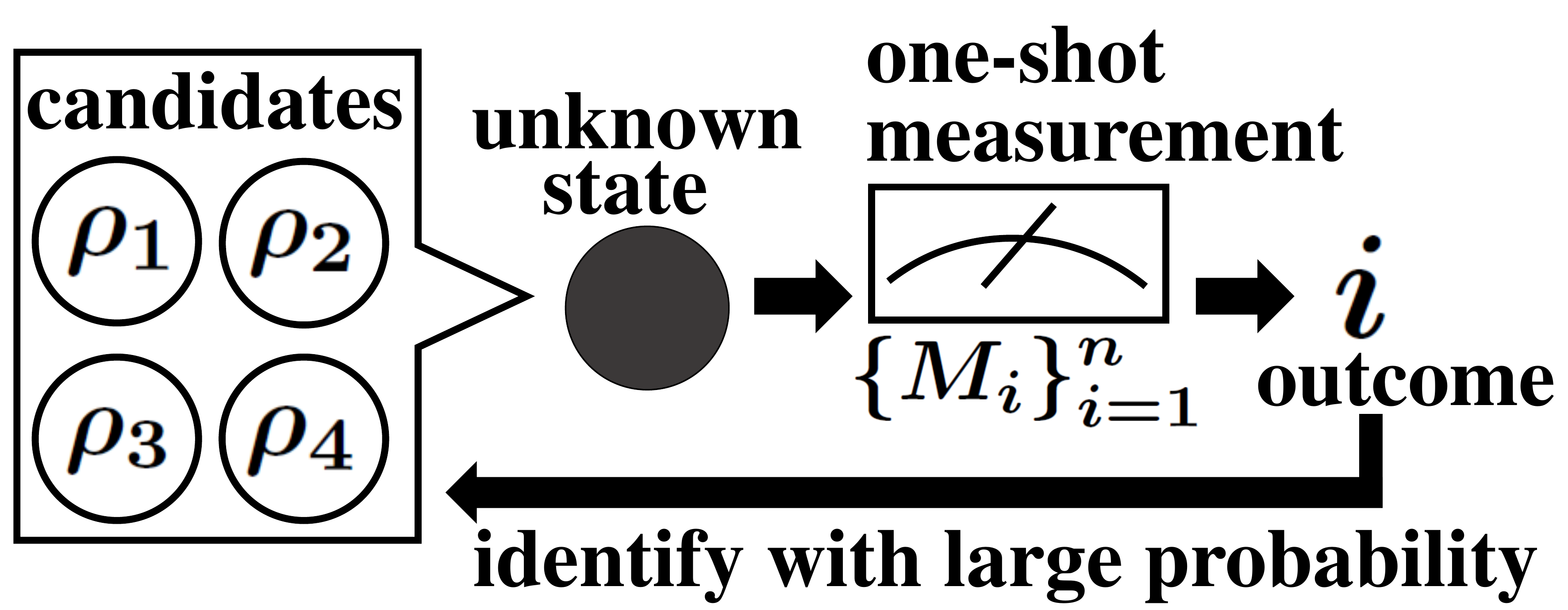}
	\caption[Discrimination Task]{
	When  number of candidates $\{\rho_i\}_{i=1}^n$ of unknown state $\rho$ is given,
	the player aims to identify $\rho$ with large probability.
	}
	\label{figure-discrimination}
\end{figure}

The aim of discrimination tasks is to find a measurement that minimize the error probability,
and there are many ways to describe the tasks mathematically.
In this thesis,
we mainly discuss the following two settings.

The first one is \textit{perfect discrimination}.

\begin{definition}[perfect distinguishablity]\label{def:dist-per}
	Let $\{\rho_k\}_{k=1}^n$ be a family of states $\rho_k\in\cS(\cC,u)$.
	Then, $\{\rho_k\}_{k=1}^n$ are perfectly distinguishable
	if there exists a measurement $\{M_k\}_{k=1}^n\in\cM(\cC,u)$ such that
	$\langle \rho_k,M_l\rangle=\delta_{kl}$.
\end{definition}

Because of the equation $\langle \rho_k,M_l\rangle=\delta_{kl}$ in Definition~\ref{def:dist-per},
it is the only possibility to get outcome $k$ that the unknown state $\rho$ is equal to $\rho_k$.
In other words,
if and only if the candidates of state $\{\rho_k\}_{k=1}^n$ are perfectly distinguishable,
the player estimates the unknown state $\rho$ with probability 1.

The second one is \textit{discimination with the minimization of the sum of error probabilities}.
In this setting, we consider the situation that the number of candidates is two.

\begin{definition}[Minimization of the Sum of Error Probabilities]\label{def:dist-min}
	Given a two-elements family of states $\{\rho_1,\rho_2\}$ with $\rho_k\in\cS(\cC,u)$,
	the sum of errors by $\bm{M}=\{M_1,M_2\}$ is defined as
	\begin{align}\label{def:error-1}
		\mathrm{Err}(\rho_1;\rho_2;\bm{M})=\Tr \rho_1M_2+\Tr \rho_2M_1.
	\end{align}
	Then, the minimization of the sum of errors in an entanglement structure $\cC$ is defined as
	\begin{align}\label{def:error-2}
		\mathrm{Err}_{\cC}(\rho_1;\rho_2):=\min_{\bm{M}\in\cM_2(\cC,u)} \mathrm{Err}(\rho_1;\rho_2;\bm{M}).
	\end{align}
\end{definition}

In this task, the player aims to find a measurement that minimize the sum of error probability.
In statistics, the part $\Tr \rho_1M_2$ and $\Tr \rho_2M_1$ are called Type I error and Type II error, respectively.
We remark that $\rho_1$ and $\rho_2$ are perfectly distinguishable if and only if $\mathrm{Err}_{\cC}(\rho_1;\rho_2)=0$
because the following equation holds:
\begin{align}
	\Tr \rho_i(M_1+M_2)=\Tr \rho_i I=1.
\end{align}

In the SES, more generally in standard quantum theory,
the discrimination tasks \ref{def:dist-per} and \ref{def:dist-min} are well-studied as follows \cite[Section~3.2]{HayashiBook2017}.
First, in the model $(\Her{\cH},\Tr, \Psd{\cH},I)$,
the following two conditions are equivalent:
\begin{enumerate}
	\item $\{\rho_i\}$ is perfectly distinguishable.
	\item $\Tr \rho_i \rho_j=0$ for any $i\neq j$.
\end{enumerate}
Second, in the model $(\Her{\cH},\Tr, \Psd{\cH},I)$,
especially in the SES,
the minimization of the sum of errors is given as
\begin{align}\label{eq:trace-norm}
	\mathrm{Err}_{\Psd{\cH}}(\rho_1;\rho_2)=1-\frac{1}{2}\|\rho_1-\rho_2\|_1,
\end{align}
whose minimizer measurement $\{M_1,M_2\}$ is given by
\begin{align}
	M_1&=\frac{1}{2}\left(\rho_1-\rho_2+|\rho_1-\rho_2|\right),\\
	M_2&=\frac{1}{2}\left(-\rho_1+\rho_2+|\rho_1-\rho_2|\right).
\end{align}

\subsection{Discrimination Tasks in Entanglement Structures}\label{sect:3-1-2}

Roughly speaking,
the width of measurement space $\cM(\cC,u)$ determines the performance for the above discrimination tasks.
Therefore, it is a possibility in an ES $\cC$ that the performance is further improved than that of $\mathrm{SES}(A;B)$
when its measurement space $\cM(\cC,I)$ is larger than $\cM(\mathrm{SES},I)$.
Recently, the preceding studies \cite{Arai2019,YAH2020} have investigated the performance for discrimination task \ref{def:dist-per} in some entanglement structures.

The reference \cite{Arai2019} has investigated the performance for the discrimination task \ref{def:dist-per} in the entanglement structure $\mathrm{SEP}(A;B)$,
and it has given an equivalent condition to discriminate two pure states in $\mathrm{SEP}(A;B)$ perfectly.

\begin{theorem}[\cite{Arai2019}]\label{preceding-1}
	Let $\rho^{(1)}=\rho_A^{(1)}\otimes\rho_B^{(1)}$ and $\rho^{(2)}=\rho^{(2)}_A\otimes\rho^{(2)}_B$ be pure states in $\mathrm{SEP}(A;B)$.
	$\rho_1$ and $\rho_2$ are perfectly distinguishable in $\mathrm{SEP}(A;B)$,
	i.e., there exists a measurement $\bm{M}\in\cM_2(\mathrm{SEP}(A;B),I)E$ such that $\Tr \rho_iM_j=\delta_{ij}$
	if and only if the following inequality holds:
	\begin{align}\label{eq:sep-dist}
		\Tr \rho_A^{(1)}\rho_A^{(2)}+\Tr \rho_B^{(1)}\rho_B^{(2)}\le 1.
	\end{align}
\end{theorem}

Also, the reference \cite{YAH2020} has investigated the performance for the discrimination task \ref{def:dist-per} in more general classes.
The reference \cite{YAH2020} defined the following two One-parameter family of entanglement structures.

\begin{definition}[One-parameter family of entanglement structures (I)]
	For $s\ge0$, we define the positive cone $\cC_s^{\nege}$ as 
	\begin{equation}
		\cC_s^{\nege} = \set{X\in\cT(AB) | \Tr\rho X\ge0 \;(\forall \rho:\mbox{separable}), \;\nege(X)\le s\Tr X},
	\end{equation}
	where the function $\nege\colon \cT(AB)\to[0,\infty)$ is defined as 
	\begin{align}
	\nege(X) = \max_{\lambda\text{ eigenvalue}\text{ of }X} \{-\lambda,0\}.
	\end{align}
\end{definition}

\begin{definition}[One-parameter family of entanglement structures (II)]
	For a vector $v\in\cH_A\otimes\cH_B$, let $\sco(v)$ be the value 
	\begin{align}
	\sco(v) = 
	\begin{cases}
		\lambda_1\lambda_2 & v\not=0,\\
		0 & v=0,
	\end{cases}
	\end{align}
	where $\lambda_1 \ge \lambda_2 \ge \cdots \ge \lambda_d$ are Schmidt coefficients of $v/\|v\|$.\\
	Then, for $s\ge0$, we define
	\begin{align}
	\begin{aligned}
		\cC_s^{(0)} &:= \conv\{ \ketbra{v}{v} \mid v\in\cH_A\otimes\cH_B,\ \sco(v)\le s \},\\
		\cC_s^{\sco} &:= \cT_+(AB) + \Gamma(\cC_s^{(0)}),
	\end{aligned}
	\end{align}
	where $\Gamma$ is partial transposition, i.e., $\Gamma=\id\otimes \top$.
\end{definition}

Then, the reference \cite{YAH2020} gave sufficient conditions to discriminate two separable pure states in the above entanglement structures.

\begin{theorem}[\cite{YAH2020}]\label{preceding-2}
	Given a pair of two pure separable states $\rho_1,\rho_2$,
	the states $\rho_1$ and $\rho_2$ are perfectly distinguishable by a measurement $\cM(\cC_s^{\nege},I)$
	if the point $(\Tr\rho_1^A\rho_2^A,\Tr\rho_1^B\rho_2^B)$ belongs to the set
		\begin{equation}
			\set{(x,y)\in[0,1]^2 | xy \le 16s^2(1-x)(1-y)} 
		\end{equation}
	for $s\in[0,1/4]$.
	Also,
	given a pair of two pure separable states $\rho_1,\rho_2$,
	the states $\rho_1$ and $\rho_2$ are perfectly distinguishable by a measurement $\cM(\cC_s^{\sco},I)$
	if the point $(\Tr\rho_1^A\rho_2^A, \Tr\rho_1^B\rho_2^B)$ belongs to the set
	\begin{equation}
		\set{(x,y)\in[0,1]^2 | xy \le t(1-x)(1-y)}
	\end{equation}
	for $t\in[0,1]$ with $s=\sqrt{t}/(1+t)$.
\end{theorem}

The preceding studies \cite{Arai2019,YAH2020} showed that the above certain ESs have extraordinary performance for discrimination task \ref{def:dist-per}.
However, the preceding studies \cite{Arai2019,YAH2020} cannot address more general ESs even if we consider only the ESs $\cC$ satisfying $\cC\subset\mathrm{SES}(A;B)$.

On the other hand,
there exists an indicator that does not change in any entanglement structure.
As an example,
we introduce the capacity of a model.
Define the number $\mathrm{Cap}(\cC)$ of a model $\cC$ as
the maximum number $m$ of perfectly distinguishable states $\{\rho_k\}_{k=1}^m$ in the model $\cC$.
The following proposition is known for the capacity
of entanglement structures.
\begin{proposition}[{\cite[proposition4.5]{Yoshida2021}}]\label{prop:cap}
	For any cone $\cC$,
	$\mathrm{Cap}(\cC)=\dim(\cH_A\otimes\cH_B)$ holds if $\cC$ satisfies $\mathrm{SEP}(A;B)\subset\cC\subset\mathrm{SEP}^\ast(A;B)$.
\end{proposition}

The proofs of both Theorem~\ref{preceding-1} and Theorem~\ref{preceding-2} essentially depend on the following certain formed measurement $\{M_1,M_2\}$.
\begin{align}
\begin{aligned}\label{eq:measurement-preceding}
	M_1&=T_1+\Gamma(T_1),\quad M_2=T_2+\Gamma(T_2),\\
	T_1 &= \frac{1}{2\gamma}
	\begin{bmatrix}
		\gamma & 0 & 0 & -\beta_1\beta_2\gamma/\alpha_1\alpha_2\\
		0 & \gamma-1 & 0 & -(\gamma-1)\beta_1/\alpha_1\\
		0 & 0 & \gamma-1 & -(\gamma-1)\beta_2/\alpha_2\\
		-\beta_1\beta_2\gamma/\alpha_1\alpha_2 &
		-(\gamma-1)\beta_1/\alpha_1 & -(\gamma-1)\beta_2/\alpha_2 &
		 2-\gamma
	\end{bmatrix}
	,\\
	T_2 &= \frac{1}{2\gamma}
	\begin{bmatrix}
		\qquad 0\qquad&\qquad 0\qquad&\qquad 0\qquad&\qquad 0\\
		\qquad 0\qquad&\qquad 1\qquad&\quad \beta_1\beta_2\gamma/\alpha_1\alpha_2 \quad&\quad (\gamma-1)\beta_1/\alpha_1\\
		\qquad 0\qquad&\quad \beta_1\beta_2\gamma/\alpha_1\alpha_2 \quad&\qquad 1\qquad&\quad (\gamma-1)\beta_2/\alpha_2\\
		\qquad 0\qquad&\quad (\gamma-1)\beta_1/\alpha_1 \quad&\quad (\gamma-1)\beta_2/\alpha_2 \quad&\quad 2(\gamma-1)
	\end{bmatrix},
\end{aligned}
\end{align}
where $\Gamma=\id\otimes \top$, $\alpha_i\in[0,1]$, and $\gamma = \alpha_1+\alpha_2$.
Of course, the above measurement does not generate all measurements in entanglement structures.
In the next section,
we generalize the preceding studies \cite{Arai2019,YAH2020},
i.e.,
this thesis investigates general measurement and its performance for the discrimination tasks \ref{def:dist-per} and \ref{def:dist-min} in any entanglement structures.

\section{Theme A : Characterization of DOVMs via Discrimination Tasks}\label{sect:3-2}

In this section,
we classify Dual-Operator-Valued Measurements (DOVMs) and characterize them by the performance for information tasks as Main Result 1.
First, we define DOVMs and classify them into five types (Definition~\ref{def:bq}-\ref{def:povm}) by their eigenvalues in Section~\ref{sect:3-2-1}.
A DOVM is an available measurement in a certain ESs,
which consists of Hermitian matrices (not necessarily positive semi-definite).
This thesis classifies DOVMs by the interval between the maximum and the minimum eigenvalues.
Second, we characterize DOVMs by discrimination tasks (Definition~\ref{def:dist-per} and \ref{def:dist-min})
as Theorem~\ref{theorem:bq} and ~\ref{theorem:aq} in Section~\ref{sect:3-2-2}.
It is known that 
some types of DOVMs have extraordinary performance for perfect discrimination \cite{Arai2019,YAH2020}.
This thesis generalizes the preceding studies,
i.e.,
this thesis investigates the performance for discrimination tasks for general DOVMs.
As a result,
we give equivalent conditions for extraordinary performance for two types of discrimination tasks (Theorem~\ref{theorem:bq} and Theorem~\ref{theorem:aq}).
Moreover,
as an application of one of the equivalent conditions,
the performance for discrimination task derives the SES from ESs with an inclusion relation (Theorem~\ref{theorem:ses-dis}).
Furthermore,
as an application of Main Theme A,
in section~\ref{sect:3-2-3},
we define simulability of DOVMs and show non-simulability of BQ measurement (Theorem~\ref{theorem:non-sim-1}).

\subsection{Classes of DOVMs}\label{sect:3-2-1}

Because the condition \eqref{eq:quantum} implies the inclusion relation $\cC^\ast\subset\mathrm{SEP}^\ast(A;B)$ for any entanglement structure $\cC$,
any measurement in $\cM_2(\cC,I)$ belongs to $\cM_2(\mathrm{SEP}(A;B),I)$.
In this thesis, we call a measurement in $\cM_2(\mathrm{SEP}(A;B),I)$ a \textit{Dual-Operator-Valued Measure} (DOVM),
and we use the following notation:
\begin{align}\label{def:dovm}
	\rm{DOVM}(A;B):=\cM_2(\mathrm{SEP}(A;B),I).
\end{align}
Any effect of DOVMs is Hermitian,
and therefore, any effect has only real eigenvalues.
In this thesis,
DOVMs are classified into the following four classes of measurements by the relation about eigenvalues.
Hereinafter,
we denote the $k$-th eigenvalue of a Hermitian matrix $X$ in ascending order as $\lambda_k(X)$.

\begin{definition}[Beyond Quantum (BQ)]\label{def:bq}
	We say that a measurement $\cM=\{M_i\}_{i=1}^2\in\rm{DOVM}(A;B)$ is Beyond Quantum (BQ)
	if one of the effects $M_i$ satisfies the inequalities $\lambda_1(M_i)<0$ and $\lambda_d(M_i)\ge1$.
\end{definition}

\begin{definition}[Advantage Quantum (AQ)]\label{def:aq}
	We say that a measurement $\cM=\{M_i\}_{i=1}^2\in\rm{DOVM}(A;B)$ is Advantage Quantum (AQ)
	if one of the effects $M_i$ satisfies the inequalities $\lambda_1(M_i)<0$ and $1+\lambda_1(M_i)<\lambda_d(M_i)<1$.
\end{definition}

\begin{definition}[Non-Advantage Quantum (NAQ)]\label{def:naq}
	We say that a measurement $\cM=\{M_i\}_{i=1}^2\in\rm{DOVM}(A;B)$ is Non-Advantage Quantum (NAQ)
	if one of the effects $M_i$ satisfies the inequalities $\lambda_1(M_i)<0$ and $\lambda_d(M_i)\le1+\lambda_1(M_i)$.
\end{definition}

\begin{definition}[Positive-Operator-Valued-Measure (POVM)]\label{def:povm}
	We say that a measurement $\cM=\{M_i\}_{i=1}^2\in\rm{DOVM}(A;B)$ is a Positive-Operator-Valued-Measure (POVM)
	if any effects $M_i$ satisfies the inequalities $\lambda_1(M_i)\ge0$.
\end{definition}

Here, we remark that the set of POVM is equal to the set of (two-valued) measurements in the SES.
Also, we remark that the names of the above classes $\rm{BQ}(A;B)$, $\rm{AQ}(A;B)$, and $\rm{NAQ}(A;B)$ come from the results in Theme~A and Theme~B.
As seen in the following sections,
any measurement in $\rm{BQ}(A;B)$ cannot be simulated in the SES,
any measurement in $\rm{AQ}(A;B)$ has an advantage for the discrimination task \ref{def:dist-min} over any POVM,
and any measurement in $\rm{NAQ}(A;B)$ has no advantages for the discrimination tasks \ref{def:dist-per} and \ref{def:dist-min} over any POVM.

A typical example of DOVMs is the measurement \eqref{eq:measurement-preceding}.
The reference \cite{Arai2019} shows that the matrices $M_1$ and $M_2$ in \eqref{eq:measurement-preceding} are not positive semi-definite,
i.e., $\lambda_1(M_i)<0$.
Because of the equation $M_1+M_2=I$,
the matrix $M_i$ satisfies $\lambda_d(M_i)>1$,
and therefore, the measurement \eqref{eq:measurement-preceding} is BQ.

Now, we denotes each set of all measurements belonging to the above classes as $\rm{BQ}(A;B)$, $\rm{AQ}(A;B)$, $\rm{NAQ}(A;B)$, and $\rm{POVM}(A;B)$, respectively.
By the above definitions, the following relation holds:
\begin{align}
\begin{aligned}
	&\rm{DOVM}(A;B)\\
	=&\rm{BQ}(A;B)\uplus\rm{AQ}(A;B)\uplus\rm{NAQ}(A;B)\uplus\rm{POVM}(A;B).
\end{aligned}
\end{align}
Our purpose is to characterize each of the classes by the discrimination tasks \ref{def:dist-per} and \ref{def:dist-min}.
In the next section, we give a complete characterization for the classes.

\subsection{Extraordinary Performance of Discrimination Tasks}\label{sect:3-2-2}

Now, we give a characterization of the classes by discrimination tasks.
First, we give a necessary and sufficient condition when a DOVM has superior performance for the discrimination task \ref{def:dist-per}.

\begin{theorem}[Main Result A-1]\label{theorem:bq}
	Given a measurement $\bm{M}=\{M_i\}_{i=1,2}\in\rm{DOVM}(A;B)$,
	the following two conditions are equivalent
	\begin{enumerate}
		\item $\bm{M}\in\rm{BQ}(A;B)$
		\item There exists a pair of two pure states $\rho_1$ and $\rho_2$ in $\cS(\mathrm{SES}(A;B),I)$ such that $\Tr \rho_i M_j=\delta_{ij}$ and $\Tr \rho_1\rho_2>0$.
	\end{enumerate}
\end{theorem}

Theorem~\ref{theorem:bq} states that the class $\rm{BQ}(A;B)$ is 
characterized by the superior performance for discrimination task \ref{def:dist-per} to that of the SES.

Similarly,
we give a necessary and sufficient condition
when a DOVM has superior performance for the discrimination task \ref{def:dist-per}.

\begin{theorem}[Main Result A-2]\label{theorem:aq}
	Given a measurement $\bm{M}=\{M_i\}_{i=1,2}\in\rm{DOVM}(A;B)$,
	the following two conditions are equivalent
	\begin{enumerate}
		\item $\bm{M}\in\rm{BQ}(A;B)\cup\bm{AQ}(A;B)$
		\item There exists a pair of two states $\rho_1$ and $\rho_2$ in $\cS(\mathrm{SES}(A;B),I)$ such that 
		\begin{align}\label{eq-aq}
			\mathrm{Err}(\rho_1;\rho_2;\bm{M}) < 1-\frac{1}{2}\|\rho_1-\rho_2\|_1.
		\end{align}
	\end{enumerate}
	Moreover, if the condition 1 holds, $\rho_1$ and $\rho_2$ can be chosen as a state $\cS(\mathrm{SEP}(A;B),I)$ in condition 2.
\end{theorem}

Theorem~\ref{theorem:bq} states that the class $\rm{BQ}(A;B)\cup\bm{AQ}(A;B)$ is 
characterized by the superior performances for discrimination task \ref{def:dist-min} to that of the SES.

We summarize the characterization as a table (Table~\ref{table-dovm}).

\begin{table}[h]
	\centering
	\caption[Characterization of DOVMs by discrimination tasks]{Characterization of DOVMs by the performance for discrimination tasks. When the class has a performance superior to POVM, we denote checkmarks.}
	\label{table-dovm}
	\begin{tabular}{|c||c|c|c|c|c|}
		\hline
		discrimination task & BQ  & AQ & NAQ & POVM \\ \hline \hline
		perfect discrimination (\ref{def:dist-per}) & $\checkmark$  & $\times$ & $\times$ & -\\
		discrimination with minimum error (\ref{def:dist-min}) & $\checkmark$ & $\checkmark$ & $\times$ & - \\ \hline
	\end{tabular}
\end{table}

Applying Theorem~\ref{theorem:aq} for the characterization of the SES,
we obtain the following theorem.

\begin{theorem}[Main Result A-3]\label{theorem:ses-dis}
	Given an entanglement structure $\cC\subset\mathrm{SES}(A;B)$,
	the following conditions are equivalent:
	\begin{enumerate}
		\item $\cC=\mathrm{SES}(A;B)$
		\item Any pair of two state $\rho_1,\rho_2\in\cS(\cC,I)$ satisfies $\mathrm{Err}_{\cC}(\rho_1;\rho_2)=1-\frac{1}{2}\|\rho_1-\rho_2\|_1$.
	\end{enumerate}
\end{theorem}

Theorem~\ref{theorem:ses-dis} states that the SES is uniquely determined by discrimination tasks \ref{def:dist-per} from any ES with smaller state space than that of the SES.
In other words,
an ES has extraordinary performance for discrimination task
if the ES has smaller state space than that of the SES,
which is regarded as a generalization of the result of \cite{Arai2019,YAH2020}
in the viewpoint of characterization of ESs by discrimination tasks.

\subsection{Theme B : Non-Simulability of BQ Measurements}\label{sect:3-2-3}

In this section,
as Theme B,
we define simulability of DOVMs and reveal the impossibility to simulate BQ measurement.
In Section~\ref{sect:3-2-3},
we define simulability of DOVMs and show non-simulability of BQ measurement (Theorem~\ref{theorem:non-sim-1}).

It is believed and well-verified that our physical systems obey standard quantum theory or the SES.
Therefore, a measurement in $\rm{DOVM}\setminus\rm{POVM}$ cannot be implemented in physical system.
However, there is a possibility that such a measurement beyond standard quantum theory can be implemented in high-dimensional standard quantum theory.
An effect $e$ of a measurement in $\rm{DOVM}\setminus\rm{POVM}$ is a non-positive matrix,
and there exists  a state of the SES such that $\Tr \rho e<0$.
In order to exclude such a ``negative probability'',
we introduce the domain of a measurement as follows.
\begin{definition}[Domain of an element in $\mathrm{SEP}^\ast$]
	Given an element $X\in\mathrm{SEP}^\ast$,
	we define the domain of $X$ as
	\begin{align}\label{def:domain-1}
		\cD(X):=\{\rho\in\cL_{\mathrm{H}}^+(\cH_A\otimes\cH_B)\mid\Tr \rho X\ge0\}.
	\end{align}
\end{definition}
\begin{definition}[Domain of DOVM]
	We define the domain of a DOVM $\bm{M}=\{M_i\}_{i\in I}$ as
	\begin{align}\label{def:domain-2}
		\cD(\bm{M}):=\bigcap_{i\in I} \cD(M_i).
	\end{align}
\end{definition}
By definition,
any DOVM $\bm{M}$ satisfies $\mathrm{SEP}(A;B)\subset\cD(\bm{M})\subset\mathrm{SEP}^\ast(A;B)$.

Next, we define $n$-simulability of a DOVM in standard quantum theory as follows.
\begin{definition}[quantum $n$-simulability of a DOVM]\label{def:simulable}
	Let $\cM=\{M_i\}_{i\in I}$ be a DOVM.
	We say that $\cM$ is quantum $n$-simulable
	if there exists a natural number $n$ and a POVM $\cN=\{N_i\}_{i\in I}$ on $(\cH_A\otimes\cH_B)^{\otimes n}$ such that
	\begin{align}\label{eq:sim}
		\Tr \rho^{\otimes n} N_i=\Tr \rho M_i \quad ^\forall\rho\in\cD(\cM)
	\end{align}
\end{definition}
If a DOVM $\bm{M}=\{M_i\}_{i\in I}$ is $n$-simulable,
the probability distribution $\{\Tr \rho M_i\}$ is simulated by $\{\Tr \rho^{\otimes n} N_i\}$ independent of a given (unknown) state $\rho$.
Of course,
it does not attain $n$-simulability to simulate a probability distribution $\{\Tr \rho M_i\}$ for a certain state.

In general,
any unknown state $\rho$ cannot be copied.
However, in physical situation,
an initial state is prepared by a certain way.
The same preparation generates the same state $\rho$.
In this situation,
we can apply the measurement $\{N_i\}_{i\in I}$ for $\rho^{\otimes n}$
even if the state $\rho$ is unknown.
On the other hand,
someone wants to consider a situation that we cannot copy the given unknown $\rho$.
In this case,
we can apply only adaptive measurements (or one way Local Operation and Classical Information measurement) \cite{Chefles2004,HayashiBook2017},
which is implemented by $n$-times sequence of measurement on the local system.
If we want to consider such setting,
we restrict the class of resource measurements.

Also,
we note that it is useless to consider ``infinite-simulability''.
When we apply an infinite number of operation for the copies of an unknown state $\rho$,
we completely extract the information of $\rho$.
Then, we can easily simulate the probability $\{\Tr\rho M_i\}$
because we can determine the value $\Tr\rho M_i$ completely.
Therefore, this thesis considers $n$-simulability for an exact finite number $n$.

This thesis considers POVMs as resource measurements.
Even though we consider the class of POVMs,
any measurement in $\bm{BQ}$ is not $n$-simulable for any natural number $n$.

\begin{theorem}[Main Result B-1]\label{theorem:non-sim-1}
	Let us consider $\cH=\cH_A\otimes\cH_B$.
	Any DOVM $\cM\in\bm{BQ}$ is not quantum $n$-simulable for any natural number $n$.
\end{theorem}

This theorem is obtained from Theorem~\ref{theorem:bq} as follows.
Take an arbitrary DOVM $\cM\in\bm{BQ}$.
Because of Theorem~\ref{theorem:bq},
there are two non-orthogonal states $\rho_1,\rho_2\in\cS(\Psd{\cH_A\otimes\cH_B},I)$ such that
$\rho_1,\rho_2$ are perfectly distinguishable by $\bm{M}=\{M_1,M_2\}$,
i.e.,
$\Tr \rho_iM_j=\delta_{i,j}$.
The relation $\Tr \rho_iM_j=\delta_{i,j}$ implies that $\rho_i$ belongs to the domain of $\bm{M}$.
However,
because the states $\rho_1$ and $\rho_2$ are non-orthogonal,
i.e.,
$\Tr \rho_1\rho_2>0$,
the following inequality holds:
\begin{align}
	\Tr \rho_1^{\otimes n}\rho_2^{\otimes n}
	=\left(\Tr \rho_1\rho_2\right)^n>0
\end{align}
In other words,
the two states $\rho_1^{\otimes n}$ and $\rho_2^{\otimes n}$ are non-orthogonal.
Therefore, any POVM $\bm{N}$ does not discriminate $\rho_1^{\otimes n}$ and $\rho_2^{\otimes n}$  perfectly,
which implies the equality \eqref{eq:sim} never holds.

In this way,
this thesis clarifies that
any measurement in $\bm{BQ}$ is not $n$-simulable for any natural number $n$.
On the other hand,
it is an open problem whether other classes of DOVMs are $n$-simulable.

Here, we remark that non-simulability is derived from the possibility to discriminate non-orthogonal states perfectly as seen in the above proof.
This relation between non-simulability and perfect discrimination of non-orthogonal states holds not only in ESs but also in more general models.
In this paper, we give an example of such models that contains a non-simulable measurement in Appendix~\ref{appe:2-3}.

\section{Proofs of theorems in Chapter~\ref{chap:3}}\label{sect:3-4}

In this section,
we prove statements in Chapter~\ref{chap:3}.

\subsection{Proof of Theorem~\ref{theorem:bq}}\label{sect:3-4-1}

\begin{proof}
	\textbf{[STEP1]} $1\Rightarrow 2$
	
	Without loss of generality,
	we assume that $M_1$ satisfies $\lambda_1(M_i)<0$ and $\lambda_{d}(M_1)\ge1$.
	By spectral decomposition, the matrix $M_1$ is decomposed into projections $E_k$ as
	\begin{align}
		M_1:=\sum_{k=1}^d \lambda_k(M_1)E_k.
	\end{align}
	Here, we denote the eigenvector of $E_k$ as $\ket{\psi_k}$.
	Because $\{M_1,M_2\}$ is a measurement, i.e., $M_1+M_2=I$,
	the matrix $M_2$ is also decomposed into projections $E_k$ as
	\begin{align}
		M_2:=\sum_{k=1}^d \left(1-\lambda_k(M_1)\right)E_k.
	\end{align}
	Take a pair of matrices $\rho_1$ and $\rho_2$ as
	\begin{align}
		\rho_i:&=\ketbra{\phi_i}{\phi_i} \quad(i=1,2)\\
		\ket{\phi_1}:&=\sqrt{\cfrac{\lambda_d(M_1)-1}{\lambda_d(M_1)-\lambda_1(M_1)}}\ket{\psi_1}
		+\sqrt{\cfrac{1-\lambda_1(M_1)}{\lambda_d(M_1)-\lambda_1(M_1)}}\ket{\psi_d},\\
		\ket{\phi_2}:&=\sqrt{\cfrac{\lambda_d(M_1)}{\lambda_d(M_1)-\lambda_1(M_1)}}\ket{\psi_1}
		+\sqrt{\cfrac{-\lambda_1(M_1)}{\lambda_d(M_1)-\lambda_1(M_1)}}\ket{\psi_d}.
	\end{align}
	Then, the matrices $\rho_1$ and $\rho_2$ are rank 1 and positive semi-definite,
	i.e., $\rho_1$ and $\rho_2$ belong to $\cS(\mathrm{SES},I)$.
	Also, the choice of $\rho_i$ implies the following equations:
	\begin{align}
		\Tr \rho_i M_j=\delta_{ij}.
	\end{align}
	Therefore, $\rho_1$ and $\rho_2$ are perfectly distinguishable by $\bm{M}$.
	Finally, we obtain the equation $\Tr \rho_1\rho_2>0$ as follows:
	\begin{align}
		&\Tr \rho_1\rho_2=|\braket{\phi_1|\phi_2}|^2\nonumber\\
		=&\cfrac{\left(\lambda_d(M_1)^2+\lambda_1(M_1)^2-\lambda_d(M_1)-\lambda_1(M_1) \right)^2}{\left(\lambda_d(M_1)-\lambda_1(M_1)\right)^2}\nonumber\\
		=&\cfrac{\left(\lambda_d(M_1)\left(\lambda_d(M_1)-1\right)+\lambda_1(M_1)^2-\lambda_1(M_1) \right)^2}{\left(\lambda_d(M_1)-\lambda_1(M_1)\right)^2}
		\stackrel{(a)}{>}0.
	\end{align}
	The inequality $(a)$ is shown by the inequalities $\lambda_1(M_1)<0$ and $\lambda_d(M_1)\ge1$.
	\\
	
	\textbf{[STEP2]} $2\Rightarrow 1$
	
	At first,
	because a POVM perfectly discriminates only orthogonal states,
	$\bm{M}\not\in \rm{POVM}(A;B)$.
	Then, one of the effects $M_1$ is not positive semi-definite.
	Without loss of generality,
	we assume that $M_1$ is not positive semi-definite,
	which implies $\lambda_1(M_1)<0$.
	
	Now, we show $\cM\in\rm{BQ}(A;B)$ by contradiction.
	Then, we assume that $\lambda_d(M_1)<1$.
	Because of the equation $M_1+M_2=I$,
	we obtain the inequality
	\begin{align}
		\lambda_1(M_2)=1-\lambda_d(M_1)>0,
	\end{align}
	which implies that $\lambda_k(M_2)>0$ for any $k=1,\cdots,d$.
	Therefore, $\Tr M_2\rho>0$ for any non-zero $\rho\in\Psd{\cH_A\otimes\cH_B}$.
	This contradicts the existence of $\rho_1\in\cS(\mathrm{SES}(A;B),I)$ such that $\Tr \rho_2M_1=0$.
	Therefore, we obtain $\lambda_d(M_1)\ge1$,
	which implies that $\bm{M}$ belongs to $\rm{BQ}(A;B)$.
\end{proof}

\subsection{Proof of Theorem~\ref{theorem:aq}}\label{sect:3-4-2}

For the proof of Theorem~\ref{theorem:aq},
we apply the following facts.

\begin{proposition}[{\cite[equation (3.59)]{HayashiBook2017}}]\label{prop:trace-norm}
	Given a pair of states $\rho_1$ and $\rho_2$ in $\cS(\Psd{\cH},I)$,
	the following equation holds:
	\begin{align}
		\min_{0\le T\le I}\Tr \rho_1(I-T)+\Tr \rho_2T=1-\frac{1}{2}\|\rho_1-\rho_2\|_1,
	\end{align}
	where the order relation $\le$ is defined by $T\ge0\Leftrightarrow T\in\Psd{\cH}$.
\end{proposition}

\begin{proposition}[\cite{GB2003}]\label{prop:sep-ball}
	If a Hermitian matrix $X\in\Her{\cH_A\otimes\cH_B}$ satisfies $\|I-X\|_2\le1$,
	then $X$ belongs to $\mathrm{SEP}(A;B)$.
\end{proposition}

\begin{proof}[Proof of Theorem~\ref{theorem:aq}]
	\textbf{[STEP1]} $1\Rightarrow 2$
	
	Because of the relation $\bm{M}\in\rm{BQ}(A;B)\cup\bm{AQ}(A;B)$,
	one of the matrices $M_i$ satisfies $\lambda_d(M_i)-\lambda_1(M_i)>1$.
	Without loss of generality,
	we assume that the matrix $M_1$ satisfies $\lambda_d(M_1)-\lambda_1(M_1)>1$.
	By spectral decomposition, the matrix $M_1$ is decomposed into projections $E_k$ as
	\begin{align}
		M_1:=\sum_{k=1}^d \lambda_k(M_1)E_k.
	\end{align}
	Then, we take two states $\rho_1$ and $\rho_2$ as
	\begin{align}
		\rho_1:=&\frac{1}{d}I,\\
		\rho_2:=&\frac{1}{d}I+\frac{1}{\sqrt{2}d}\left(E_1-E_d\right).
	\end{align}
	The state $\rho_1$ belongs to $\mathrm{SEP}(A;B)$.
	Because the following inequality holds:
	\begin{align}
		\|d\rho_2-I\|_2
		=&\|\frac{1}{\sqrt{2}}\left(E_1-E_d\right)\|_2=1,
	\end{align}
	Proposition~\ref{prop:sep-ball} implies $\rho_2\in\mathrm{SEP}(A;B)$.
	Then, the following inequality shows \eqref{eq-aq}.
	\begin{align}
		&\mathrm{Err}(\rho_1;\rho_2;\bm{M})
		=\Tr \rho_1M_2+\rho_2M_1\nonumber\\
		=&\Tr \rho_1 +\Tr\left(\rho_2-\rho_1\right)M_1
		=1+\Tr \frac{1}{\sqrt{2}d}\left(E_1-E_d\right)M_1\nonumber\\
		=&1+\frac{1}{\sqrt{2}d}\lambda_1-\frac{1}{\sqrt{2}d}\lambda_d
		=1+\frac{1}{\sqrt{2}d}\left(\lambda_1-\lambda_d\right)\nonumber\\
		<&1-\frac{1}{\sqrt{2}d}
		=1-\frac{1}{2}\|\rho_1-\rho_2\|_1.
	\end{align}
	\\
	
	\textbf{[STEP2]} $2\Rightarrow 1$
	
	We show the contraposition,
	i.e.,
	the statement that the relation $\bm{M}\in\rm{NAQ}\cup\rm{POVM}$ implies
	the inequality
	\begin{align}\label{eq:naq}
		\mathrm{Err}(\rho_1;\rho_2;\bm{M}) \ge 1-\frac{1}{2}\|\rho_1-\rho_2\|_1.
	\end{align}
	Proposition~\ref{prop:trace-norm} shows the inequality \eqref{eq:naq} in the case $\bm{M}\in\rm{POVM}$.
	Therefore, the remaining case is only $\bm{M}\in\rm{NAQ}$.
	Without loss of generality,
	$\lambda_1(M_1)<0$ and $\lambda_d(M_i)\le1+\lambda_1(M_i)$ hold.
	
	Because $\lambda_1(M_1)<0$ and $\lambda_d(M_1)\le1+\lambda_1(M_1)$ hold,
	the inequality $\lambda_d(M_1)<1$ holds,
	which implies that $M_2=I-M_1$ is positive semi-definite.
	Therefore, the relation $0\le \frac{1}{\lambda_d(M_2)}M_2\le I$,
	i.e., $\bm{M'}\in\cM(\mathrm{SES}(A;B),I)$ holds,
	and we define $\bm{M'}$ as $\bm{M'}=\{I-\frac{1}{\lambda_d(M_2)}M_2,\frac{1}{\lambda_d(M_2)}M_2\}$.
	Then, the following inequality holds:
	\begin{align}
		&\mathrm{Err}(\rho_1;\rho_2;\bm{M})
		=\lambda_d(M_2)\left(\mathrm{Err}(\rho_1;\rho_2;\bm{M}')\right)\nonumber\\
		\ge&\lambda_d(M_2)\mathrm{Err}_{\mathrm{SES}(A;B)}(\rho_1;\rho_2)
		=\lambda_d(M_2)\left(1-\frac{1}{2}\|\rho_1-\rho_2\|_1\right)\nonumber\\
		\stackrel{(a)}{>}&1-\frac{1}{2}\|\rho_1-\rho_2\|_1.
	\end{align}
	The inequality $(a)$ is shown by the inequality $\lambda_d(M_2)=1-\lambda_1(M_1)>1$.
	As a result,
	$\bm{M}\in\rm{NAQ}$ satisfies \eqref{eq:naq}.
\end{proof}

\subsection{Proof of Theorem~\ref{theorem:ses-dis}}\label{sect:3-4-3}

\begin{proof}[Prof of Theorem~\ref{theorem:ses-dis}]
	The statement $1\Rightarrow 2$ has already been shown by \eqref{eq:trace-norm}.
	Then, we show the contraposition of $2\Rightarrow1$ as follows.
	
	Assume that the inclusion relation $\cC\subsetneq\mathrm{SES}(A;B)$ holds,
	which implies the inclusion relation $\cC^\ast\supsetneq\mathrm{SES}(A;B)$.
	Then, there exists a matrix $T\in\cC^\ast\setminus\mathrm{SES}(A;B)$.
	Especially, the matrix $T$ belongs to $\mathrm{SEP}^\ast(A;B)\setminus\mathrm{SES}(A;B)$,
	which implies that $\lambda_1(T)<0$ and $\lambda_d(T)>0$ hold.
	Now, we take the matrix $T':=\frac{1}{\lambda_d(T)}T$,
	which belongs to $\cC^\ast$.
	Because of the equation $\lambda_d(T')=1$,
	the matrix $I-T$ belongs to $\mathrm{SES}(A;B)\subset\cC^\ast$.
	As a result, the family $\bm{T'}=\{T',I-T'\}$ belongs to $\cM(\cC,I)$.
	Also, $T'$ satisfies $\lambda_1(T')=\frac{\lambda_1(T)}{\lambda_d(T)}<0$.
	As a result, the measurement $\bm{T'}$ is AQ.
	Hence, Theorem~\ref{theorem:aq} shows the negation of the condition 2.
\end{proof}

\chapter{Self-duality and Existence of PSESs}\label{chap:4}

In this chapter,
we explore ESs with self-duality and group symmetry.
Self-duality and homogeneity play an important role to derive algebraic structure in physical systems.
When a model satisfies self-duality and homogeneity, 
the proper cone of the model is characterized by Euclidean Jordan Algebras \cite{Jordan1934,Koecher1957,Barnum2019,BMA2020},
which essentially lead limited types of models including classical and quantum theory \cite{Jordan1934,Koecher1957}.
Also, the successful result \cite{Barnum2019} derives the models corresponding to Jordan Algebras.
However,
it is an open problem how drastically one of the above two properties determines ESs.
This thesis attacks this problem and investigates the diversity of ESs with self-duality and symmetric conditions.

First,
in Section~\ref{sect:4-1},
we briefly introduce self-duality and homogeneity,
and we investigate ESs with group symmetric conditions as Theme C.
In this section,
we show that an ES with self-duality and homogeneity is limited to the SES (Theorem~\ref{prop:sym1}).
Also, we show that an ES with global unitary symmetry is limited to the SES (Theorem~\ref{prop:global2}).
These two results imply that global unitary symmetry is a weaker condition than the condition in the reference \cite{Barnum2019}.

Second,
in Section~\ref{sect:4-2},
as Theme D,
in order to investigate ESs with self-duality, 
we give a general theory about self-duality.
In this section,
we define a pre-dual cone and show that any pre-dual cone can be modified to a self-dual cone (Theorem~\ref{theorem:sd}).
Next, we clarify the relation between exact hierarchy of pre-dual cones and independent family of self-dual cones (Theorem~\ref{theorem:hie1}).

The above general theory implies there are infinitely many self-dual ESs.
Moreover, this thesis attacks harder problem,
i.e.,
we investigate the existence of a self-dual ES that is near the SES in Section~\ref{sect:4-3-1}.
This thesis defines a condition called \textit{$\epsilon$-undistinguishability},
where the model cannot be distinguished from the SES by verification of maximally entangled states \cite{HMT,Ha09-2,HayaM15,PLM18,ZH4,Markham,MST,KSKWW,Bavaresco,FVMH,JWQ}.
In this thesis,
we call a model with self-duality and $\epsilon$-undistinguishability \textit{$\epsilon$-Pseudo Standard Entanglement Structures} ($\epsilon$-PSESs),
and we show that infinite existence of $\epsilon$-PSESs for any $\epsilon>0$ (Theorem~\ref{theorem:main}) as Theme E.
In contrast to the similarity in terms of $\epsilon$-undistinguishability,
we see operational difference between the SES and PSESs in terms of state discrimination tasks.
In Section~\ref{sect:4-3-2},
we show that some types of PSESs have non-orthogonal perfectly distinguishable states (Theorem~\ref{theorem:dist}).

The proofs of statements in Chapter~\ref{chap:4} are written in Section~\ref{sect:4-4}.

\section{Topics about Self-Duality and Group Symmetry in ESs}\label{sect:4-1}

In this section,
we review the preceding studies about self-duality and group symmetric conditions in GPTs,
and we investigates the diversity of ESs with group symmetric conditions.

First,
we introduce self-duality and group symmetric conditions, including homogeneity in Section~\ref{sect:4-1-1}.
Also, we introduce cones with a pair of conditions, i.e., symmetric cones,
which is essentially classified into five types.

Second,
as Main Result 3,
we clarify the uniqueness of the ES with certain group symmetric conditions
in Section~\ref{sect:4-1-2}.
We show that a model with symmetric cones is limited to the SES (Theorem~\ref{prop:sym1}).
Furthermore, we show that a weak group symmetric condition derives the SES in ESs (Theorem~\ref{prop:global2}).

\subsection{Self-Duality and Homogeneity}\label{sect:4-1-1}

In the studies of proper cones,
two properties \textit{self-duality} and \textit{homogeneity} play an important role.

First, we define \textit{self-dual cone}.
\begin{definition}[Self-Duality]\label{def:self-dual}
	We say that a positive cone $\cC$ is self-dual
	when the cone $\cC$ satisfies $\cC=\cC^\ast$.
\end{definition}
In some studies of GPTs,
Definition~\ref{def:self-dual} is called strong self-duality.
On the other hand,
some studies say that a positive cone $\cC\subset\cV$ is weakly self-dual
when there exists a linear automorphism $f$ on $\cV$ such that
$f(\cC^\ast)=\cC$.
It is known that any proper cone $\cC$ is weakly self-dual \cite{Janotta2013}.
The reference \cite{Janotta2013} has shown that the state space and the effect space
of any model can be transformed by linear map from one to another, where the effect space is considered as the subset of $\cV^\ast$. 
That is, the result \cite{Janotta2013} can be interpreted in our setting as follows; the state space and effect space become equivalent by changing inner product.
This process is called self-dualization in \cite{Janotta2013}. 
However, this thesis addresses entanglement structures,
whose local structures are completely equivalent to standard quantum theory.
Standard quantum theory fixes the inner product,
and therefore, entanglement structures also possess the fixed inner product $\Tr$.
Therefore, this thesis discusses strong self-duality,
and the result \cite{Janotta2013} cannot be used for our purpose.
Hereinafter, we call the property $\cC=\cC^\ast$ self-dual (omitting ``strong'').

Another important symmetric property is \textit{homogeneity}.
\begin{definition}
	For a positive cone $\cC$ in a vector space $\cV$, define the set $\mathrm{Aut}(\cC)$ as
	\begin{align}\label{def:auto}
		\mathrm{Aut}(\cC):=\{f\in\mathrm{GL}(\cV) \mid f(\cC)=\cC\}.
	\end{align}
	Then, we say that a positive cone $\cC$ is homogeneous
	if there exists a map $g\in\mathrm{Aut}(\cC)$ for any two elements $x,y\in \cC^\circ$ such that $g(x)=y$.
\end{definition}

A proper cone with self-duality and homogeneity is called a \textit{symmetric cone},
which is essentially classified into finite kinds of cones including the SES \cite{Jordan1934,Koecher1957}.
In order to introduce the classification of symmetric cones,
we define the direct sum of cones as follows:
\begin{definition}
	If a family of positive cones $\{\cC_i\}_{i=1}^k$ satisfies $\cC_i\cap\cC_j=\{0\}$ for any $i\neq j$,
	the direct sum of $\cC_i$ is defined as
	\begin{align}\label{def:direct-sum}
		\bigoplus_{i=1}^k \cC_i:=\left\{\sum_{i=1}^k x_i\middle| x_i\in\cC_i\right\}.
	\end{align}
\end{definition}
Here, we say that a positive cone $\cC$ is irreducible if
the cone $\cC$ cannot be decomposed by a direct sum over more than 1 positive cones as
\begin{align}
	\cC=\bigoplus_{i=1}^k \cC_i.
\end{align}

Irreducible symmetric cones are classified into the following five cases \cite{Jordan1934,Koecher1957}:
\begin{inparaenum}[(i).]
	\item $\mathrm{PSD}(m,\mathbb{R})$,
	\item $\mathrm{PSD}(m,\mathbb{C})$,
	\item $\mathrm{PSD}(m,\mathbb{H})$,
	\item $\mathrm{Lorentz}(1,n-1)$,
	\item $\mathrm{PSD}(3,\mathbb{O})$,
\end{inparaenum}
where $n$ and $m$ are arbitrary positive integers,
$\mathrm{PSD}(m,\mathbb{K})$ denotes the set of positive semi-definite matrices on a $m$-dimensional Hilbert space over a field $\mathbb{K}$ and $\mathrm{Lorentz}(1,n-1)$ is defined as
\begin{align}
	&\mathrm{Lorentz}(1,n-1)
	:=\{(z,x)\in\mathbb{R}\oplus\mathbb{R}^{n-1}\mid |z|^2\ge|x|^2,\ z\ge0\}.
\end{align}
About reducible symmetric cones,
it is known that
any symmetric cone $\cC$ can be decomposed by a direct sum over irreducible symmetric cones $\cC_i$ as
\begin{align}
	\cC=\bigoplus_{i=1}^k \cC_i
\end{align} \cite{FKsymmetric}.

In this way,
symmetric cones have been studied well as a general theory of proper cones.
However,
it is an open problem
how drastically symmetric cones determine entanglement structures.
Also, it is an open problem
how drastically either self-duality or homogeneity determine entanglement structures.
In this thesis,
we explore these problems.

\subsection{Theme C : ESs with group symmetry}\label{sect:4-1-2}

Here, we investigate ESs with group symmetric conditions.

First,
we show that an ES $\cC$ is equal to $\mathrm{SES}(A;B)$ if the cone $\cC$ is symmetric.
\begin{theorem}[Main Result C-1]\label{prop:sym1}
	Assume a symmetric cone $\cC$ satisfies \eqref{eq:quantum}.
	Then, $\cC=\mathrm{SES}(A;B)$.
\end{theorem}
In other words,
the combination of self-duality and homogeneity uniquely determine the SES.
The preceding study \cite{Barnum2019} gives a condition about symmetry to derive symmetric cone.
Therefore, restricting ESs, i.e.,
under the condition \eqref{eq:composite},
the condition in \cite{Barnum2019} derives the SES.

Moreover, the SES is uniquely determined by another group symmetric condition.
Here, we define the class of global unitary maps given as
\begin{align}
	\mathrm{GU}(A;B):=&\{g\in\mathrm{GL}(\cT(\cH_A\otimes\cH_B)) \mid g(\cdot):=U^\dag (\cdot) U,\nonumber\\
	& U \ \mbox{is a unitary matrix on $\cH_A\otimes\cH_B$}\}.\label{eq:gu}
\end{align}
Then, the condition $\mathrm{GU}(A;B)$-symmetry uniquely derives the SES from all ESs.
\begin{theorem}[Main Result C-2]\label{prop:global2}
	Assume that a model $\cC$ satisfies \eqref{eq:quantum} and $\mathrm{GU}(A;B)$-symmetric.
	Then, $\cC=\mathrm{SES}(A;B)$.
\end{theorem}
The condition $\mathrm{GU}(A;B)$-symmetry is weaker condition than homogeneity as the following proposition.
\begin{proposition}\label{prop:sym2}
	A symmetric cone $\cC$ with \eqref{eq:quantum} satisfies that
	$\mathrm{Aut}(\cC)\supset\mathrm{GU}(A;B)$.
\end{proposition}
Theorem~\ref{prop:sym1} implies that $\mathrm{Aut}(\cC)$ is larger than $\mathrm{GU}(A;B)$ under the condition that $\cC$ is a symmetric cone with \eqref{eq:quantum}.
Proposition~\ref{prop:sym2} is shown by Theorem~\ref{prop:sym1} and the inclusion relation $\mathrm{Aut}(\mathrm{SES}(A;B))\supset\mathrm{GU}(A;B)$.
Since Theorem~\ref{prop:sym1} requires a larger symmetry than 
theorem~\ref{prop:global2} under the condition \eqref{eq:quantum},
we can conclude that
the assumption of Theorem~\ref{prop:sym1} is stronger than that of Theorem~\ref{prop:global2},
which implies the condition in \cite{Barnum2019} is stronger than $\mathrm{GU}(A;B)$-symmetry under the condition \eqref{eq:quantum}.

In this way,
a kind of group symmetric conditions can determine entanglement structures well.
On the other hand,
it is an open and difficult problem how drastically self-duality restricts entanglement structures.
In order to investigate self-dual entanglement structures,
we state general theory in the next section.

\section{Theme D : Pre-Dual Cone and Self-Dual Modification}\label{sect:4-2}

In this section,
we introduce pre-dual cones and state general theories for the construction of self-dual models satisfying \eqref{eq:quantum}.
In Section~\ref{sect:4-2-0},
we introduce pre-dual cone and discuss the meaning of pre-duality in physics.
In this thesis, we introduce pre-duality as a mathematical representation of the existence of projective measurements.
In Section~\ref{sect:4-2-1},
we show that pre-dual cone can be modified to a self-dual cone (Theorem~\ref{theorem:sd}).
Also, we show the relation between exact hierarchy of pre-dual cones and independent family of self-dual cones (Theorem~\ref{theorem:hie1}).

\subsection{Pre-dual Cone and Its Meaning in Physics}\label{sect:4-2-0}


In standard quantum theory,
there exists a measurement $\{e_i\}_{i\in I}$ such that the post-measurement state with the outcome $i$ is given as $e_i/\Tr e_i$ independently of the initial state when the effect $e_i$ is pure.
Such a measurement is called a projective measurement \cite{Neumann1932,Luders1951,Davies1970,Ozawa1984,Chefles2003,Buscemi2004}\footnote{The above property is sometimes called repeatability, but this thesis defines repeatability as the latter property (Figure~\ref{figure-repeatable}). As the following discussion, the two properties are strictly distinguished from each other in general. However, in finite-dimensional standard quantum theory, the two properties are equivalent, as shown in \cite{Buscemi2004}.}.
The measurement projectivity is one of the postulates of standard quantum theory \cite{Neumann1932,Luders1951,Davies1970,Ozawa1984}.
Therefore, in this paper, we impose that any model $\cC$ satisfies the following condition:
for any pure effect $e\in\cE(\cC,u)$, there exists a measurement $\{e_i\}$ such that
an element $e_{i_0}$ is equal to $e$, and the post-measurement state is given as $\overline{e_{i_0}}:=e_{i_0}/\Tr e_{i_0}$.

Pure effects span the effect space $\cE(\cC,u)$ with convex combination,
and the effect space $\cE(\cC,u)$ generates the dual cone $\cC^\ast$ with constant time.
Therefore, the existence of projective measurement implies the inclusion relation $\cC\supset\cC^\ast$.
In this paper, this property $\cC\supset\cC^\ast$ is called pre-duality.
\begin{definition}[Pre-Dual Cone]\label{def:pre-dual}
	Given a proper cone $\cC$ in $\cV$,
	we say that $\cC$ is pre-dual
	if $\cC$ satisfies the inclusion relation $\cC\supset\cC^\ast$.
\end{definition}

\begin{figure}[t]
	\centering
	\includegraphics[width=4cm]{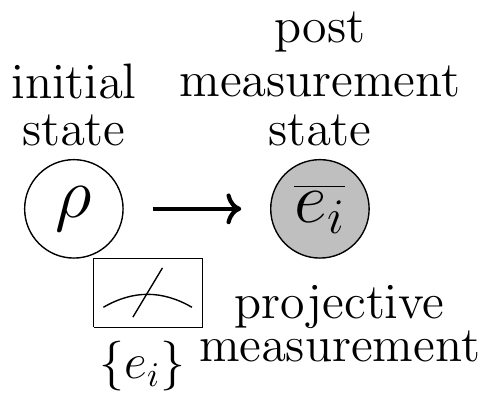}
	\caption[Projective Measuremets]{
	When a projective measurement $\{e_i\}$ is applied to the system with an initial state $\rho$,
	we obtain an outcome $i$ and the corresponding post-measurement state $\overline{e_i}=e_i/\Tr e_i$ independent of the initial state $\rho$.
	}
	\label{figure-projection}
\end{figure}

Here, we remark on the relation between projectivity and repeatability.
The following relation between projectivity and repeatability is also discussed in \cite[Discussion]{Luders1951} and \cite{Buscemi2004}.
Repeatability is a postulate of standard quantum theory, sometimes included in the projection postulate \cite{Neumann1932,Luders1951}\footnote{One may think that the settings in \cite{Neumann1932} and \cite{Luders1951} should not be categorized as repeatability in our setting.
However, the reference \cite{Neumann1932} starts from a similar assumption like repeatability defined in this thesis and derives projectivity of measurements from the assumption.
Also,
the discussion of the translated version of \cite[Discussion]{Luders1951} emphasizes distinguishability.
Therefore, we divide the references \cite{Neumann1932,Luders1951} from the references \cite{Davies1970,Ozawa1984,Chefles2003}.
As mentioned in the former footnote, these two properties are equivalent in finite-dimensional quantum theory.
On the other hand, this thesis distinguishes projectivity and repeatability from the viewpoint of GPTs,
and the main statement in this section is that ``projectivity derives pre-duality''.}.
Repeatability ensures that 
the same effect is observed with probability 1 in the sequence of the same measurements,
and the effects do not change the post-measurement state (figure~\ref{figure-repeatable}).
\begin{figure}[t]
	\centering
	\includegraphics[width=6cm]{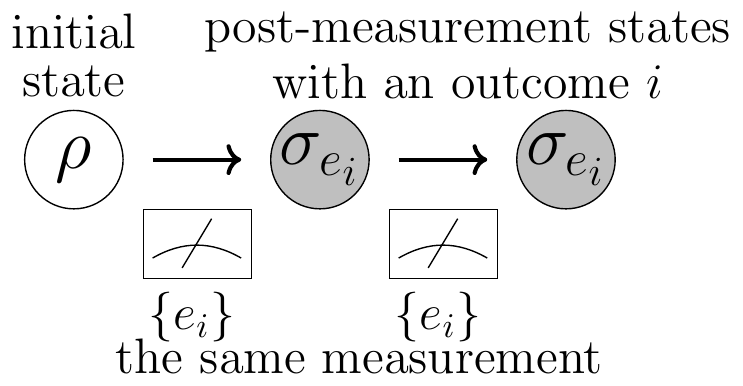}
	\caption[Repeatable Measurement]{When an initial state is measured by a measurement $\{e_i\}$ twice,
	the post-measurement states of first and second measurement with an outcome $i$ are equivalent.
	}
	\label{figure-repeatable}
\end{figure}
Repeatability requests that the tuple of post-measurement states $\{\sigma_{e_i}\}_{i\in I}$ is perfectly distinguishable,
i.e.,
the equation $\Tr \sigma_{e_i} e_j=\delta_{i,j}$ holds.
In other words,
repeatability requests the $|I|$ number of constraints for the post-measurement state $\sigma_{e_i}$.
On the other hand,
projectivity determines post-measurement states completely.
In other words,
projectivity requests the same number of constraints for the post-measurement state as the dimension of $\cC^\ast$.
In general, the number of outcomes $|I|$ is smaller than the dimension of $\cC^\ast$;
therefore, projectivity is a stronger postulate than repeatability in terms of the number of constraints.

Now, we consider pre-dual models of composite systems $(\cT(\cH_A\otimes \cH_B),\cC,I_{A;B})$.
For example,
let us consider the model that contains only separable measurements.
In such a model, the dual cone is given as $\cC^\ast=\mathrm{SEP}(A;B)$,
and therefore, the model satisfies $\cC=\mathrm{SEP}(A;B)^\ast$ because the dual of a dual cone is equal to the original cone.
However, the state space $\cS(\mathrm{SEP}(A;B)^\ast,I_{A;B})$ has excessive many states;
the state space $\cS(\mathrm{SEP}(A;B)^\ast,I_{A;B})$ has not only all quantum states but also all entanglement witnesses with trace 1.
Then, there exist two state $\rho_1,\rho_2\in\cS(\cC,I_{A;B})$ such that they satisfy $\Tr\rho_1\rho_2<0$.
Not only the case $\cC=\mathrm{SEP}(A;B)^\ast$,
but also any pre-dual model has two states $\rho_1,\rho_2$ with $\Tr\rho_1\rho_2<0$ unless $\cC=\cC^\ast$.
In this way, pre-dual models have a gap between the state space and the effect space unless $\cC=\cC^\ast$.
In order to remove such a gap,
we apply a theorem called \textit{self-dual modification} in the next section,
i.e.,
we show the possibility to extend the measurement effect space and restrict the state space with satisfying $\tilde{\cC}\supset\tilde{\cC}^\ast$.

\subsection{Self-Dual Modification and Hierarchy of Pre-Dual Cones}\label{sect:4-2-1}

Then, we show that any pre-dual model can be always modified to self-dual model.
\begin{theorem}[Main Result D-1]\label{theorem:sd}
	Let $\cC$ be a pre-dual cone in $\cV$.
	Then, there exists a positive cone $\tilde{\cC}$ such that 
	\begin{align}\label{eq:sd}
		\cC\supset\tilde{\cC}=\tilde{\cC}^\ast\supset\cC^\ast.
	\end{align}
\end{theorem}
A self-dual cone $\tilde{\cC}$ satisfing \eqref{eq:sd} is called a Self-Dual Modification (SDM) of $\cC$.
Here, we remark that the reference \cite{BF1976} has also shown a result essentially similar to Theorem~\ref{theorem:sd}.
In the reference \cite{BF1976}, a cone is defined as a closed convex set satisfying only the property that $rx\in\cC$ for any $r\ge0$ and any $x\in\cC$.
This thesis assumes additional properties, $\cC$ has non-empty interior and $\cC\cap(-\cC)=\{0\}$.
Actually, we can easily modify the proof in \cite{BF1976} for our definition,
but this thesis gives another proof based on Zorn's Lemma for reader's convenience in Section~\ref{sect:proof-sd}.

In this way,
when we want to investigate self-dual cones,
it is sufficient to focus on pre-dual cones.
Here, we remark that
SDM is not uniquely determined by a given pre-dual cone
because the proof of Theorem~\ref{theorem:sd} is derived from Zorn's Lemma,
also because the proof in \cite{BF1976} is neither constructive nor deterministic
\footnote{Because a SDM $\tilde{\cC}$ is not uniquely determined by $\cC$,
the notation $\tilde{\cC}$ is slightly confusing,
but
for a convenience in the latter discussion,
we often denote a SDM of $\cC$ as $\tilde{\cC}$ in this thsis.
In other words, when we use the notation $\tilde{\cC}$,
the set $\tilde{\cC}$ is a self-dual cone satisfying $\cC\supset\tilde{\cC}\supset\cC^\ast$.}.
However, in order to examine the uniqueness even with the above symmetric condition,
we prepare the following two concepts among several cones.
Indeed, even when two self-dual cones are self-dual modifications of different pre-dual cones, they are not necessarily different self-dual cones in general.
For example, when we have three different self-dual cone $\cC_1,\cC_2,\cC_3$,
then $\cC_1+\cC_2$ and $\cC_2+\cC_3$ are pre-dual cones,
but $\cC_2$ is regarded as a modification of $\cC_1+\cC_2$ and $\cC_2+\cC_3$.
Hence, the following two concepts are useful to clarify the difference among self-dual modifications.
\begin{definition}[$n$-independence]
	For a natural number $n$, we say that a family of sets $\{\cC_i\}_{i=1}^n$ is $n$-independent if no sets $\cC_i\ (1\le i\le n)$ satisfy that $\cC_i \subset \sum_{j\neq i} \cC_j$.
	Especially, we say that $\{\cC_i\}_{i=1}^n$ is $n$-independent family of cones when any $\cC_i$ is a positive cone.
\end{definition}
\begin{definition}[exact hierarcy with depth $n$]
	For a natural number $n$, we say that pre-dual cone $\cC$ has an exact hierarchy with depth $n$
	if there exists a family of sets $\{\cC_i\}_{i=1}^n$ such that
	\begin{align}
			\cC\supset\cC_1\supsetneq\cC_2\supsetneq\cdots\supsetneq\cC_n
			\supset\cC_n^\ast\supsetneq\cdots\supsetneq\cC_1^\ast\supset\cC^\ast.
	\end{align}
	Especially, we say that $\{\cC_i\}_{i=1}^n$ is an exact hierarchy of cones when any $\cC_i$ is a positive cone.
\end{definition}

Then, as an extension of theorem~\ref{theorem:sd},
the following theorem shows the equivalence between
the existence of an $n$-independent family of self-dual cones
and the existence of an exact hierarchy of pre-dual cones with depth $n$.
\begin{theorem}[Main Result D-2]\label{theorem:hie1}
	Let $\cC$ be a positive cone.
	The following two statements are equivalent:
	\begin{enumerate}
		\item there exists an exact hierarchy of pre-dual cones $\{\cC_i\}_{i=1}^n$ satisfying $\cC\supset\cC_i\supset\cC^\ast$.
		\item there exists an $n$-independent family of self-dual cones $\{\tilde{\cC}_i\}_{i=1}^n$ satisfying that $\tilde{\cC}_i$ is a self-dual modification of $\cC_i$,
		i.e., 
		$\tilde{\cC}_i$ is a self-dual cone satisfying $\cC_i\supset\tilde{\cC}_i\supset\cC_i^\ast$.
	\end{enumerate}
\end{theorem}
Because of Theorem~\ref{theorem:hie1},
when we want to show the existence of different self-dual cones,
it is sufficient to show the existence of an exact hierarchy of pre-dual cones.
When we apply Theorem~\ref{theorem:hie1} to the existence of self-dual ESs,
we need to show the existence of an exact hierarchy of pre-dual ESs.
Such a hierarchy is easily constructed by adding an extremal ray to $\mathrm{SEP}(A;B)$
because the set $\mathrm{SES}(A;B)\setminus\mathrm{SEP}(A;B)$ is an infinite set.
In Section~\ref{sect:4-3},
we apply Theorem~\ref{theorem:hie1} to the existence of self-dual ESs with an important additional condition.

\section{Theme E : Existence of PSESs and Difference from the SES}\label{sect:4-3}

In this section,
we investigate PSESs.
PSESs are introduced by verification of maximally entangled states.
Recently, many studies discussed verification of maximally entangled states 
from theory \cite{HMT,Ha09-2,HayaM15,PLM18,ZH4,Markham} to experiment \cite{MST,KSKWW,Bavaresco,FVMH,JWQ}.
However, their verification ensures only that the constructed state is close to 
the maximally entangled state
because their verification inevitably contains small errors.
Therefore,
an ES attains verification of maximally entangled states
if the ES has all states that are close to 
the maximally entangled states.
Simply considering,
ESs with large state space attains the verification,
but self-duality bans such surplus of states.
Then,
this thesis investigates a class of ESs, called PSESs, with self-duality and attainment of the verification with small errors.

First, in Section~\ref{sect:4-3-1},
we define $\epsilon$-PSESs and explain the importance of the definition.
Next, in Section~\ref{sect:4-3-2},
we show the infinite existence of $\epsilon$-PSESs for any $\epsilon>0$ (Theorem~\ref{theorem:main}).
Finally,
in Section~\ref{sect:4-3-3},
we show that there exist $\epsilon$-PSESs with non-orthogonal perfect discrimination for any $\epsilon>0$ (Theorem~\ref{theorem:dist}).

\subsection{Definition of PSES and Its Importance}\label{sect:4-3-1}

Here, we want to consider self-dual ESs near the SES,
i.e.,
self-dual ESs that have a small distance from the SES.
In this thesis,
as an operational meaningful distance between models,
we introduce experimental verification of a given model.
To consider the experimental verification of a given model,
we consider the distinguishability of two state spaces of two given models $\cC_1$ and $\cC_2$.
Given a state $\sigma\in\cS(\cC_2,u_2)$, the quantity 
\begin{align}\label{def:distance1}
D(\cC_1\|\sigma):= \min_{\rho \in \cS(\cC_1,u_1)}\| \rho-\sigma\|_1
\end{align}
expresses how well 
the state $\sigma$ is distinguished from states in $\cC_1$.
Optimizing the state $\sigma$, we consider the quantity
\begin{align}\label{def:distance2}
D(\cC_1\|\cC_2):=\max_{\sigma \in \cS(\cC_2,u_2)} D(\cC_1\|\sigma),
\end{align}
which expresses the optimum distinguishability of the model $\cC_2 $
from the model $\cC_1$.
Hence, the quantity $D(\mathrm{SES}(A;B) \|\cC)$ expresses
how the standard model $\mathrm{SES}(A;B)$
can be distinguished from a model $\cC$.

However, we often consider the verification of 
a maximally entangled state
because a maximally entangled state is the furthest state from separable states.
That is, when the range of the above maximization \eqref{def:distance2} is restricted to maximally entangled states, 
the distinguishability of the standard model $\mathrm{SES}(A;B)$ from
the model $\cC$ is measured by the following quantity:
\begin{align}\label{def:distance}
D(\cC)
:= &\max_{\sigma \in \mathrm{ME}(A;B) } D(\cC\|\sigma),
\end{align}
where the set $\mathrm{ME}(A;B)$ is denoted as the set of all maximally entangled states on $\cH_A\otimes\cH_B$.
Given a model $\cC$,
we introduce \textit{$\epsilon$-undistinguishable condition} as
\begin{align}\label{cd:epsilon}
	D(\cC)\le\epsilon.
\end{align}
That is, if a model $\cC$ satisfies $\epsilon$-undistinguishablity,
even when we pass verification test for a maximally entangled state with $\epsilon$-errors,
we cannot deny the possibility that our system is the model $\cC$ (figure~\ref{figure-verification}).
Clearly, there are many models satisfying this condition,
for example, 
$\mathrm{SEP}^\ast$ satisfies it because $D(\mathrm{SEP}^\ast)=0$.
However, such an ES is not self-dual,
therefore, we can easily deny the possibility of $\mathrm{SEP}^\ast$ by imposing self-duality.

\begin{figure}[t]
	\centering
	\includegraphics[width=12cm]{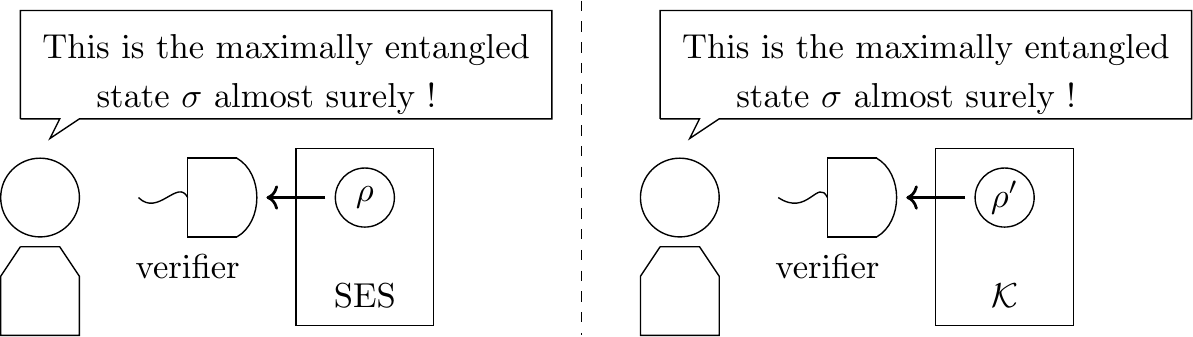}
	\caption[Verification Tasks]{
	Even if the verifier's system is subject to an $\epsilon$-undistinguishable entanglement structure $\cC\neq\mathrm{SES}(A;B)$,
	the verifier achieves the verification task of a given maximally entangled state $\sigma$ with error $\epsilon$ by preparing a state $\rho'\in\cC$ satisfying $\|\rho'-\sigma\|_1\le\epsilon$.
	In this sense, such verification tasks can not distinguish the entanglement structures $\mathrm{SES}(A;B)$ and $\cC$ when $\cC$ satisfies $\epsilon$-undistinguishability.
	}
	\label{figure-verification}
\end{figure}

Then,
we call an ES \textit{$\epsilon$-Pseudo Standard Entanglement Structures} ($\epsilon$-PSESs)
if the ES satisfies self-duality and $\epsilon$-undistinguishablity.
In other words,
$\epsilon$-PSESs are models that cannot be denied by self-duality and verification test for a maximally entangled state with $\epsilon$-errors.
A typical example of $\epsilon$-PSESs is of course the SES,
but
another example of $\epsilon$-PSESs is not known,
especially in the case when $\epsilon$ is very small.
In this section,
we investigate the problem whether there exists another example of $\epsilon$-PSESs.
As a result,
we give an infinite number of examples of PSESs by applying a general theory given in the next section.

\subsection{Existence of PSESs}\label{sect:4-3-2}

In this section,
in order to discuss variety of $\epsilon$-PSESs,
we apply Theorem~\ref{theorem:hie1} to a model of the quantum composite system on $\cH_A\otimes\cH_B$ with $\dim(\cH_A\otimes\cH_B)=d^2$.
Then, we show infinite number of exactly different $\epsilon$-PSESs.


First,
we denote
$\mathrm{MEOP}(A;B)$
as the set of all maximally entangled orthogonal projections on $\cH_A\otimes\cH_B$,
i.e.,
\begin{align}\label{eq:proj}
\begin{aligned}
	&\mathrm{MEOP}(A;B):=\Bigl\{\vec{E}=\{\ketbra{\psi_k}{\psi_k}\}_{k=1}^{d^2} \Bigm\vert \braket{\psi_k|\psi_l}=\delta_{kl}, \\
	&\ketbra{\psi_k}{\psi_k} :\mbox{ maximally entangled state on }\cH_A\otimes\cH_B\Bigr\}.
\end{aligned}
\end{align}
Now, we define the followin sets for the construction of PSESs.
\begin{definition}\label{definition:Kr}
	Given a subset $\cP\subset\mathrm{MEOP}(A;B)$ and a parameter $r\ge0$,
	we define the following set of non-positive matrices:
	\begin{align}
	\begin{aligned}
		&\mathrm{NPM}_r(\cP)\\
		:=&\Bigl\{\rho=-\lambda E_1+(1+\lambda)E_2+\frac{1}{2}\sum_{k=3}^{d^2}E_k\Big|
		0\le\lambda\le r,\ \vec{E}=\{E_k\}\in\cP
		\Bigr\}.\label{def:NPM}
	\end{aligned}
	\end{align}
	Using  the above set $\mathrm{NPM}_r(\cP)$,
	given a parameter $r\ge0$,
	we define the following two cones
	$\cC^{(0)}_r(\cP)$
	and
	$\cC_r(\cP)$
	as
	\begin{align}
		\cC^{(0)}_r(\cP)
		:&=\mathrm{SES}(A;B)+
		\mathrm{NPM}_r(\cP)
		,\\
		\cC_r(\cP)
		:&=\left(
		\cC^{(0)\ast}_r(\cP)+\mathrm{NPM}_r(\cP)
		\right)^\ast.\label{def:Kr}
	\end{align}
\end{definition}
Then, the following proposition holds.
\begin{proposition}\label{prop:construction1}
	Given $\cH_A$, $\cH_B$,
	define a real number $r_0(A;B)$ as
	\begin{align}\label{def:r0}
		r_0(A;B):=\left(\sqrt{2d}-2\right)/4.
	\end{align}
	When two parameters $r_1$ and $r_2$ satisfy $r_2\le r_1\le r_0(A;B)$,
	two cones
	$\cC_{r_1}(\cP)$
	and
	$\cC_{r_2}(\cP)$
	are pre-dual cones satisfying \eqref{eq:quantum}
	and the inclusion relation
	\begin{align}\label{eq:con-hie}
		\cC_{r_2}(\cP)\subsetneq\cC_{r_1}(\cP).
	\end{align}
\end{proposition}
Proposition~\ref{prop:construction1} guarantees that
$\cC_r(\cP)$
is pre-dual for any $r\le r_0$.
Therefore, Theorem~\ref{theorem:sd} gives a self-dual modification of
$\cC_r(\cP)$
with \eqref{eq:quantum}.
Next, we calculate the value $D(\tilde{\cC}_r(\cP))$.
The following proposition estimates the value $D(\tilde{\cC}_r(\cP))$.
\begin{proposition}\label{prop:construction2}
	Given a parameter $r$ with $0<r\le r_0(A;B)$
	and a self-dual modification
	$\tilde{\cC}_r(\cP)$,
	the following inequality holds:
	\begin{align}\label{eq:est-distance}
		D(\tilde{\cC}_r(\cP))
		\le2\sqrt{\cfrac{2r}{2r+1}}.
	\end{align}
\end{proposition}
For the latter use, we define the parameter $\epsilon_r$ as
\begin{align}\label{def:er}
	\epsilon_r:=2\sqrt{\cfrac{2r}{2r+1}}.
\end{align}
Proposition~\ref{prop:construction2} implies that the model
$\tilde{\cC_r}(\cP)$
is an $\epsilon_r$-PSES with \eqref{eq:quantum}.
Also,
due to \eqref{eq:con-hie} in Proposition~\ref{prop:construction1},
for an arbitrary number $n$,
an exact inequality
\begin{align}\label{eq:ri}
	0<r_n<\cdots<r_1\le r_0(A;B)
\end{align}
gives an exact hierarchy of pre-dual cones
$\{\cC_{r_i}(\cP)\}_{i=1}^n$
with \eqref{eq:quantum}.
Thus,
Theorem~\ref{theorem:hie1} gives an independent family
$\{\tilde{\cC}_{r_i}(\cP)\}$
with \eqref{eq:quantum},
and the distance
$D(\tilde{\cC}_{r_i}(\cP))$
is estimated as
\begin{align}
	D(\tilde{\cC}_{r_i}(\cP))
	\le2\sqrt{\cfrac{2r_i}{2r_i+1}}<2\sqrt{\cfrac{2r_1}{2r_1+1}}=\epsilon_{r_1}
\end{align}
by inequalities \eqref{eq:est-distance} and \eqref{eq:ri}.
In other words, the family
$\{\tilde{\cC}_{r_i}(\cP)\}$
is an $n$-independent family of $\epsilon_{r_1}$-PSESs.
Because $n$ is arbitrary and $\epsilon_{r_1}\to0$ holds with $r_1\to0$,
we obtain the following theorem.
\begin{theorem}[Main Result E-1]\label{theorem:main}
	For any $\epsilon>0$,
	there exists an infinite number of $\epsilon$-PSESs.
\end{theorem}
In other words,
there exist infinite entanglement structures that cannot be distinguished from the SES by a verification of a maximally entanglement state with small errors
even if the entanglement structure is self-dual.

\subsection{Non-Orthogonal Discrimination in PSESs}\label{sect:4-3-3}

In this section,
we discuss the difference between $\epsilon$-PSESs and the SES in terms of informational tasks.
We see the difference between the behaviors of perfect discrimination in
$\tilde{\cC}_r(\cP)$
and $\mathrm{SES}$.
For this aim,
we show that
any self-dual modification $\tilde{\cC}_r(\cP)$ in Section~\ref{sect:4-3-1}
has a measurement to discriminate non-orthogonal states in $\tilde{\cC}_r(\cP)$ perfectly
for a certain subset $\cP\subset\mathrm{MEOP}(A;B)$.

First,
given a vector $\vec{P}=\{P_k\}_{k=1}^{d^2}\in\mathrm{MEOP}(A;B)$,
we define a vector $\vec{E_P}=\{P'_k\}_{k=1}^{d^2}\in\mathrm{MEOP}(A;B)$ as
\begin{align}\label{eq:E-E'}
	P'_1:=P_2,\quad P'_2:=P_1, \quad P'_k=P_k \ (k\ge3).
\end{align}
Then,
given a vector $\vec{P}=\{P_k\}_{k=1}^{d^2}\in\mathrm{MEOP}(A;B)$,
we define a subset $\cP_0(\vec{P})\subset\mathrm{MEOP}(A;B)$ as
\begin{align}\label{def:PE}
	\cP_0(\vec{P}):=\{\vec{P},\vec{E_P}\}.
\end{align}

Now, we consider perfect discrimination in a self-dual modification $\tilde{\cC}_r(\cP_0(\vec{E}))$.
By the equations \eqref{def:NPM} and \eqref{def:PE},
the following two matrices belong to $\mathrm{NPM}_r(\cP_0(\vec{E}))$ for any $\vec{E}$ and any $0\le\lambda\le r$:
\begin{align}
\begin{aligned}
	\label{eq:measurement}
	M_1(\lambda;\vec{P})&:=-\lambda P_1+(1+\lambda)P_2+\frac{1}{2}\sum_{k\ge3}P_k,\\
	M_2(\lambda;\vec{P})&:=-\lambda P_1'+(1+\lambda)P_2'+\frac{1}{2}\sum_{k\ge3}P_k'
	=(1+\lambda) P_1-\lambda P_2+\frac{1}{2}\sum_{k\ge3}P_k,
\end{aligned}
\end{align}
which implies that $M_i(\lambda;\vec{P})\in\cC_r^{(0)\ast}(\cP_0(\vec{P}))\subset\tilde{\cC}_r(\cP_0(\vec{P}))$ for $i=1,2$.
Also, because of the equation \eqref{eq:E-E'},
the equation $M_1(\lambda;\vec{P})+M_2(\lambda;\vec{P})=I$ holds.
Therefore, the family $M(\lambda;\vec{P})=\{M_i(\lambda;\vec{P})\}_{i=1,2}$ is a measurement in $\tilde{\cC}_r(\cP_0(\vec{P}))$ when $0\le\lambda\le r$.

Next, we choose a pair of distinguishable states by $M(\lambda;\vec{P})$.
Let $\ket{\psi_k}$ be a normalized eigenvector of $P_k$.
Then, we define two states $\rho_1,\rho_2$ as follows:
\begin{align}
\begin{aligned}
	\rho_1:&=\ketbra{\phi_1}{\phi_1},\quad \rho_2:=\ketbra{\phi_2}{\phi_2},\\
	\ket{\phi_1}:&=\sqrt{\cfrac{r}{2r+1}}\ket{\psi}_1+\sqrt{\cfrac{r+1}{2r+1}}\ket{\psi}_2,\\
	\ket{\phi_2}:&=\sqrt{\cfrac{r+1}{2r+1}}\ket{\psi}_1+\sqrt{\cfrac{r}{2r+1}}\ket{\psi}_2.
\end{aligned}
\end{align}
Because of the relation $\vec{P}\in\mathrm{MEOP}(A;B)$, the projections $P_i$ and $P_j$ are orthogonal for $i\neq j$,
which implies
the equations
\begin{align}\label{eq:psi1-2}
	\braket{\psi_i|\psi_j}&=\delta_{i,j},\\
	\braket{\psi_i|P_j|\psi_i}&=\delta_{i,j}.\label{eq:psi-E}
\end{align}
Therefore, the following relation holds for $i,j=1,2$:
\begin{align}
	\Tr \rho_iM_j(r;\vec{P})=\delta_{i,j},
\end{align}
i.e.,
the states $\rho_1$ and $\rho_2$ are distinguishable by the measurement $M(\lambda;\vec{P})$.

Next, we show that $\rho_1,\rho_2\in\tilde{\cC}_r(\cP_0(\vec{P}))$,
which is shown as follows.
Because of the equation
$\mathrm{NPM}_r(\cP_0(\vec{P})):=\{M_i(\lambda;\vec{P}) | 0\le \lambda\le r,\ i=1,2\}$,
any extremal element $x\in\cC_r^{(0)}(\cP_0(\vec{P}))$ can be written as $x=\sigma+M_i(\lambda;\vec{P})$,
where $\sigma\in\mathrm{SES}(A;B)$, $0\le \lambda\le r$, $i=1,2$.
Moreover, the following two inequalities hold:
\begin{align}
	\Tr \rho_1 M_1(\lambda;\vec{P})&=-\lambda\cfrac{r}{2r+1}+(1+\lambda)\cfrac{r+1}{2r+1}
	=(\lambda+r+1)\cfrac{1}{2r+1}\nonumber\\
	&\stackrel{(a)}{\ge}\cfrac{r+1}{2r+1}\ge0,\label{eq:rho1M1}\\
	\Tr \rho_1 M_2(\lambda;\vec{P})&=-\lambda\cfrac{r+1}{2r+1}+(1+\lambda)\cfrac{r}{2r+1}=(-\lambda+r)\cfrac{1}{2r+1}\stackrel{(b)}{\ge}0.\label{eq:rho1M2}
\end{align}
The equations $(a)$ and $(b)$ are shown by the inequality $0\le \lambda \le r$.
Because the inequality $\Tr \rho_1\sigma\ge0$ holds for any $\sigma\in\mathrm{SES}(A;B)$,
we obtain $\Tr \rho_1 x\ge0$ for any $x\in\cC_r^{(0)}(\cP_0(\vec{P}))$,
which implies $\rho_1\in\cC^{(0)\ast}_r(\cP_0(\vec{P}))$.
Therefore, $\rho_1\in\tilde{\cC}_r(\cP_0(\vec{P}))$ because of the inclusion relation $\cC^{(0)\ast}_r(\cP_0(\vec{P}))\subset \tilde{\cC}_r(\cP_0(\vec{P}))$.
The same discussion derives that $\rho_2\in\tilde{\cC}_r(\cP_0(\vec{P}))$.
As a result,
we obtain a measurement and a distinguishable pair of two states by the measurement in
$\tilde{\cC}_r(\cP_0(\vec{P}))$.

Finally, the following equality implies that $\rho_1$ and $\rho_2$ are non-orthogonal for $r>0$:
\begin{align}
	\Tr\rho_1\rho_2
	&=2\cfrac{r(r+1)}{(2r+1)^2}>0.\label{eq:orthogonal1}
\end{align}
That is to say,
$\rho_1$ and $\rho_2$ are perfectly distinguishable non-orthogonal states.
Here, we apply Proposition~\ref{prop:construction2} to the case with
$\epsilon=2\sqrt{(2r)/(2r+1)}$.
Then,
$\tilde{\cC}_r(\cP_0(\vec{P}))$ is an $\epsilon$-PSES that contains a pair of two perfectly distinguishable states  $\rho_1$ and $\rho_2$ with
\begin{align}\label{eq:orthogonal2}
	\Tr\rho_1\rho_2
	&\stackrel{(a)}{\ge}\cfrac{\epsilon^2(\epsilon^2+8)}{32}
\end{align}
if $r$ satisfies $\epsilon=2\sqrt{(2r)/(2r+1)}$.
The inequality $(a)$ is shown by simple calculation as seen in Section~\ref{append-4-2} (Proposition~\ref{prop:dist2}).
We summarize the result as the following theorem.
\begin{theorem}[Main Result E-2]\label{theorem:dist}
	For any $\epsilon>0$ satisfying $\epsilon=2\sqrt{(2r)/(2r+1)}$,
	there exists an $\epsilon$-PSES that has a measurement and a pair of two perfectly distinguishable states $\rho_1,\rho_2$ with \eqref{eq:orthogonal2}.
\end{theorem}

In this way,
$\epsilon$-PSESs are different from the SES in terms of state discrimination.
This result implies the possibility that orthogonal discrimination can characterize the standard entanglement structure rather than self-duality.
Such a discussion provides an important conjecture as seen in Section~\ref{sect:5-2}.

\section{Proofs in Chapter~\ref{chap:4}}\label{sect:4-4}

\subsection{Proof of theorem~\ref{prop:sym1}}

In order to prove Theorem~\ref{prop:sym1},
we consider the capacity of a model.
About the symmetric cones in the list in Section~\ref{sect:4-1-1},
preceding studies investigated the capacity of models and the dimension of vector spaces as follows \cite{FKsymmetric}.
In this table,
the dimension is defined as the dimension of vector space including the symmetric cone,
and dim/cap means the ratio of the dimension / the capacity.

\begin{table}[htb]
	\caption{List about irreducible symmetric cones}
	\label{table2}
	\centering
	\begin{tabular}{cccc}
	\hline
	symmetric cone  & capacity  & dimension & dim/cap \\ \hline \hline
	$\mathrm{PSD}(m,\mathbb{R})$ & $m$&$m(m+1)/2$ & $(m+1)/2$ \\ \hline
	$\mathrm{PSD}(m,\mathbb{C})$ &$m$&$m^2$ & $m$  \\ \hline
	$\mathrm{PSD}(m,\mathbb{H})$ &$m$ &$m(2m-1)$ &$2m-1$ \\ \hline
	$\mathrm{Lorentz}(1,n-1)$ &2 &$n$& $n/2$ \\ \hline
	$\mathrm{PSD}(3,\mathbb{O})$ &3 &8 &$8/3$ \\ \hline
	\end{tabular}
\end{table}

Then, we obtain the following lemma.

\begin{lemma}\label{lem:sym}
	Let $\cH_A,\cH_B$ be finite-dimensional Hilbert spaces with dimension larger than 1.
	If an irreducible symmetric cone $\cC$ satisfies \eqref{eq:quantum},
	i.e.,
	$\mathrm{SEP}(A;B)\subset\cC\subset\mathrm{SEP}^\ast(A;B)$,
	the cone $\cC$ is the SES,
	i.e.,
	$\cC=\mathrm{PSD}(m,\mathbb{C})$ for $m=\dim(\cH_A\otimes\cH_B)$.
\end{lemma}

\begin{proof}[Proof of Lemma~\ref{lem:sym}]
	At first, $\cC$ is restricted to the five cases in the list.
	proposition~\ref{prop:cap} implies that $\cC$ has the capacity $m=\dim(\cH_A\otimes\cH_B)$,
	which denies the possibilities $\cC=\mathrm{Lorentz}(1,n-1)$ and $\cC=\mathrm{PSD}(3,\mathbb{O})$
	because $\dim(\cH_A\otimes\cH_B)\ge4$.
	Also, the cone $\cC$ is contained by the vector space of Hermitian matrices on $\cH_A\otimes\cH_B$ with $\mathbb{C}$-valued entries.
	Therefore, the dimension is given by $m^2$.
	Only the case with $\cC = \mathrm{PSD}(m,\mathbb{C})$ satisfies the ratio of the dimension and the capacity of $\cC$
	among the cones listed in Table~\ref{table2}, which shows the desired statement.
\end{proof}

Then, we prove Theorem~\ref{prop:sym1}.

\begin{proof}[Proof of Theorem~\ref{prop:sym1}]
	First, we decompose $\cC$ by a direct sum over irreducible symmetric cones $\cC_i$ as
	\begin{align}
		\cC=\bigoplus_{i=1}^k \cC_i.
	\end{align}
	Because each $\mathrm{ER}(\cC_i)$ is disjoint
	and because any pure state $\rho$ cannot be written as $\rho=\rho_1+\rho_2$ for any Hermitian matrices $\rho_1,\rho_2$ that are not transformed by multiplying any real number,
	the inclusion relation
	\begin{align}
		\mathrm{ER}(\mathrm{SEP}(A;B))\subset\bigcup_{i=1}^k \mathrm{ER}(\cC_i)
	\end{align}
	holds.
	Because the set $\mathrm{ER}(\mathrm{SEP}(A;B))$ is topologically connected,
	$\mathrm{ER}(\mathrm{SEP}(A;B))$ cannot be written as the disjoint sum of closed sets.
	Therefore, there exists an index $i_0$ such that $\mathrm{ER}(\mathrm{SEP}(A;B))\subset\mathrm{ER}(\cC_{i_0})$,
	which implies $\mathrm{SEP}(A;B)\subset\cC_{i_0}$.
	Because $\cC_{i_0}$ is self-dual, the inclusion relation $\mathrm{SEP}^\ast(A;B)\supset\cC_{i_0}$ holds.
	Hence, we apply lemma~\ref{lem:sym} to $\cC_{i_0}$,
	and we obtain $\cC_{i_0}=\mathrm{SES}(A;B)$.
	Thus, we obtain the inclusion relation $\cC\supset\mathrm{SES}(A;B)$,
	which implies $\cC=\mathrm{SES}(A;B)$ because $\cC$ is self-dual.
\end{proof}

\subsection{Proof of Theorem~\ref{prop:global2}}

\begin{proof}[Proof of Theorem~\ref{prop:global2}]
	We show the statement by contradiction.
	Assume that $\cC\neq\mathrm{SES}(A;B)$.
	If $\cC$ satisfies $\cC\subsetneq\mathrm{SES}(A;B)$,
	$\cC$ is not self-dual because of the inclusion relation $\cC\subsetneq\mathrm{SES}(A;B)\subsetneq\cC^\ast$.
	Therefore, we assume the existence of the element $x\in\cC\setminus\mathrm{SES}(A;B)$ without loss of generality.
	Because $x\in\mathrm{SEP}^\ast(A;B)\setminus\mathrm{SES}(A;B)$,
	there exists a pure state $\rho\in\mathrm{SES}(A;B)$ such that $\Tr \rho x<0$.
	Because $\rho$ is pure,
	there exists a unitary map $g\in\mathrm{GU}(A;B)$ such that $g(\rho)\in\mathrm{SEP}(A;B)$.
	Also, because $\cC$ is $\mathrm{GU}(A;B)$-symmetry,
	$g(x)\in\cC$.
	However, $\Tr g(\rho)g(x)=\Tr \rho x<0$,
	and therefore,
	$g(x)\not\in\mathrm{SEP}^\ast(A;B)$.
	This contradicts to $g(x)\in\cC\subset\mathrm{SEP}^\ast(A;B)$.
\end{proof}

\subsection{Proof of theorem~\ref{theorem:sd}}\label{sect:proof-sd}

\begin{proof}[Proof of theorem~\ref{theorem:sd}]
	\textbf{[OUTLINE]}
	First, as STEP1, we define $\cX$ as a set of all pairs of pre-dual cone and its dual cone,
	and we also define an order relation on $\cX$.
	Next, as STEP2,
	we show the existence of a maximal element in $\cX$ by Zorn's lemma,
	i.e.,
	we show that any totally ordered subset $\cB\subset\cX$ has an upper bound in $\cX$.
	Finally, as STEP3, we show that any maximal element corresponds to self-dual cone.\\
	
	\textbf{[STEP1]}
	Definition of $\cX$ and an order relation on $\cX$.
	
	Define the set $\cX$ of all pairs of pre-dual cone and its dual as:
	\begin{align}
	\begin{aligned}\label{eq:set}
		\cX:=\Bigl\{X:=&(\cC_X,\cC_X^\ast)\subset\cV\times\cV \Big| \cC\supset \cC_X, \ \cC_X \mbox{ is pre-dual cone}\Bigr\}.
	\end{aligned}
	\end{align}
	Also, we define an order relation $\preceq$ on $\cX$ as
	\begin{quote}
		$X\preceq Y \ \Leftrightarrow \cC_X\supseteq \cC_Y, \mbox{ and } \cC_X^\ast \subseteq \cC_Y^\ast$
		for any $X=(\cC_X,\cC_X^\ast)$, $Y=(\cC_Y,\cC_Y^\ast)$.
	\end{quote}
	
	\textbf{[STEP2]}
	The existence of the maximal element.
	
The aim of this step is showing
the existence of the maximal element of $\cX$
by applying Zorn's lemma.
For this aim, we need to show the existence of an upper bound
for every totally ordered subset in $\cX$.
That is, it is needed to show 
that the element written as
	\begin{align}
		X':=\left(\bigcap_{B\in\cB} \cC_B, \Big(\bigcap_{B\in\cB} \cC_B\Big)^\ast\right)
	\end{align}
is an upper bound in $\cX$ 
for a totally ordered subset $\cB\subset\cX$.
Since any $X\in\cB$ satisfies $X\preceq X'$ by definition of $X'$,
	non-trivial thing is $X'\in\cX$. 
Therefore, we show this membership relation in the following.
	
	Because any $\cC_B$ satisfies $\cC_B\supset\cC^\ast$, the subset
	$\bigcap_{B\in\cB}\cC_B$ has non-empty interior.
	Therefore, it is sufficient to show that $\bigcap_{B\in\cB} \cC_B$ is pre-dual in order to show $X'\in\cX$.
That is, the condition $X'\in\cX$ follows from 
the relation $\Big(\bigcap_{B\in\cB} \cC_B\Big)^\ast\subset
\bigcap_{B\in\cB} \cC_B$.

	For any $X=(\cC_X,\cC_X^\ast)\in\cB$ and $Y=(\cC_Y,\cC_Y^\ast)\in\cB$,
	one of the following inclusion relations holds by total order of $\cB$:
	\begin{align}
		\cC_X\supseteq \cC_X^\ast \supseteq \cC_Y^\ast  &\quad \left(X\preceq Y\right),\\
		\cC_X\supseteq \cC_Y \supseteq \cC_Y^\ast &\quad \left(Y\preceq X\right).
	\end{align}
	Therefore, $\cC_X\supset\cC_Y^\ast$ holds for any $X,Y\in\cB$,
	which implies $\bigcap_{B\in\cB} \cC_B\supset \cC_X$ for any $X\in\cB$.
	Hence, we have
	$\sum_{B\in\cB} \cC_B^\ast \subset \bigcap_{B\in\cB} \cC_B$ because
	the set $\bigcap_{B\in\cB} \cC_B^\ast\supset \cC_X^\ast$ is a positive cone,
	i.e.,
	closed under linear combination of non-negative scalars.
	Because lemma~\ref{lem:sum-cone} guarantees 
	the relation
	$\Big(\bigcap_{B\in\cB} \cC_B\Big)^\ast=\sum_{B\in\cB} \cC_B^\ast$, the above discussion implies
	$X'$ is pre-dual,
	and therefore, we obtain the relation $X'\in\cX$.
	
	Consequently, we have finished showing that every totally ordered in $\cX$ has an upper bound in $\cX$.
	Therefore, Zorn's lemma ensures the existence of the maximal element $X\preceq \tilde{X}\in\cX$.
	\\
	
	\textbf{[STEP3]}
	Self-duality of any maximal element.
	
	We consider maximal element of $X=(\cC,\cC^\ast)$ and write the maximal element as $\tilde{X}=(\tilde{\cC},\tilde{\cC}^\ast)$.
	Here, we will show $\tilde{\cC}$ is self-dual by contradiction.
	Assume $\tilde{\cC}$ is not self-dual, i,e,, $\tilde{\cC}\supsetneq\tilde{\cC}^\ast$,
	and, we take an element $x_0\in\tilde{\cC}\setminus\mathrm{Clo}\left({\cC}^\ast\right)$.
	Then, $\cC'^\ast:=\tilde{\cC}^\ast+\{x_0\}$ satisfies $\cC\supsetneq\cC'$ because $\cC^\ast\subsetneq\cC'^{\ast}$.
	Hence, $\cC'\in\cB$ and $\tilde{X}\prec(\cC',\cC'^\ast)$ hold.
	However, this contradicts to the maximality of $\tilde{X}$.
	As a result, $\tilde{\cC}$ is self-dual.
\end{proof}

\subsection{Proof of Theorem~\ref{theorem:hie1}}

\begin{proof}[Proof of theorem~\ref{theorem:hie1}]
	\textbf{[STEP1]}
	$(i)\Rightarrow(ii)$.
	
	Let $\{\cC_i\}_{i~1}^n$ be an exact hierarchy of pre-dual cones with $\cC\supset\cC_i\supset\cC^\ast$.
	By fixing an element $\rho_i\in\cC_i\setminus\cC_{i+1}$,
	define cones $\cL_i$ as the self-dual modification of $\cC_i'$
	\begin{align}\label{def:cone-sd}
		\cC_i':=\left(\cC_i^\ast + \{\rho_i\}\right)^\ast.
	\end{align}
	Let us show the pre-duality of $\cC_i'$.
	Take any two elements $x',y'\in\cC_i^{\prime\ast}$.
	Because of \eqref{def:cone-sd},
	the elements $x',y'$ is written as $x'=x+\rho_i$, $y'=y+\rho_i$,
	where $x,y\in\cC_i^\ast$.
	Pre-duality of $\cC_i$ implies
	that $\langle x,y \rangle\ge0$.
	Also, the definition of dual implies $\langle x,\rho_i\rangle\ge0$ and $\langle y,\rho_i\rangle\ge0$.
	Because $\langle \rho_i,\rho_i\rangle=||\rho_i||>0$,
	$\langle x',y'\rangle\ge0$ holds,
	which implies that $\cC_i^{\prime\ast}\subset(\cC_i^{\prime\ast})^\ast=\cC_i'$.
	Hence, $\cC_i'$ is a pre-dual cone,
	and Theorem~\ref{theorem:sd} guarantees the existence of a SDM $\tilde{\cC}_i$ satisfying $\cC_i'\supset\tilde{\cC'}_i\supset\cC_i^{\prime\ast}$.
	Also, the definition \eqref{def:cone-sd} implies the inclusion relation $\cC_i^\ast\subset\cC_i^{\prime\ast}$,
	and therefore, we obtain the inclusion relation
	\begin{align}
		\cC_i\supset\cC_i'\supset\tilde{\cC'}_i\supset\cC_i^{\prime\ast}\supset\cC_i^\ast.
	\end{align}
	
	Now, we show the independence of $\{\tilde{\cC'}_i\}$,
	i.e.,
	we show $\tilde{\cC'}_i\not\subset\tilde{\cC'}_j$ and $\tilde{\cC'}_i\not\supset\tilde{\cC'}_j$ for any $i>j$.
	We remark that any two elements $a,b$ in a self-dual cone satisfies $\langle a,b\rangle\ge0$.
	Because $\rho_i$ belongs to $\cC_i\setminus\cC_{i+1}$,
	$\rho_i\not\in\cC_{i+1}\supset\cC_j\supset\tilde{\cC'}_j$ holds,
	which implies $\tilde{\cC'}_i\not\subset\tilde{\cC'}_j$.
	The opposite side $\tilde{\cC'}_i\not\supset\tilde{\cC'}_j$ is shown by contradiction.
	Assume that $\tilde{\cC'}_i\supset\tilde{\cC'}_j$,
	and therefore, Proposition~\ref{prop:dual-inc} implies $\tilde{\cC'}_i\subset\tilde{\cC'}_j$.
	This contradicts to $\tilde{\cC'}_i\not\subset\tilde{\cC'}_j$.
	As a result, we obtain $\tilde{\cC'}_i\not\supset\tilde{\cC'}_j$.
	\\
	
	\textbf{[STEP2]}
	$(ii)\Rightarrow(i)$.

	Let $\{\tilde{\cC}_i\}_{i=1}^n$ be an independent family of self-dual cones with $\cC\supset\tilde{\cC}_i\supset\cC^\ast$.
	Now, we define a cone $\cC_i$ as
	\begin{align}
		\cC_i:=\sum_{j\ge i} \tilde{\cC}_j.
	\end{align}
	The choice of $\cC_i$ implies the inclusion relation $\cC_i\supset\tilde{\cC}_i$,
	i.e.,
	$\tilde{\cC}_i$ is a self-dual modification of $\cC_i$.
	Also, because of the inclusion relation $\cC\supset\tilde{\cC}_i\supset\cC^\ast$,
	the choice of $\cC_i$ implies the inclusion relation $\cC\supset\cC_i\supset\cC^\ast$.
	Moreover, the independence of $\tilde{\cC}_i$ implies the inclusion relation
	\begin{align}
	\cC_i=\sum_{j\ge i} \tilde{\cC}_j\supsetneq\sum_{j\ge i+1} \tilde{\cC}_j=\cC_{i+1},
	\end{align}
	which implies that $\{\cC_i\}_{i=1}$ is an exact hierarchy of pre-dual cones.
\end{proof}

\subsection{Proof of Proposition~\ref{prop:construction1}}\label{append-construction}

For the proof of Proposition~\ref{prop:construction1},
we give the following lemmas.

\begin{lemma}\label{lem:con1}
	For given $\cH_A$ and $\cH_B$,
	the relation
	\begin{align}\label{eq:rn}
		\mathrm{NPM}_r(\cP)
		\subset\mathrm{SEP}(A;B)^\ast
	\end{align}
	holds for $0\le r \le (\sqrt{d}-1)/2$.
\end{lemma}

\begin{proof}[Proof of lemma~\ref{lem:con1}]

	The aim of this proof is showing that any element
	$x\in\mathrm{NPM}_r(\cP)$
	satisfies $x\in\mathrm{SEP}^\ast(A;B)$,
	i.e.,
	the element $x$ satisfies $\Tr xy\ge0$ for any $y\in\mathrm{SEP}(A;B)$.
	Take an arbitrary element
	$x\in\mathrm{NPM}_r(\cP)$.
	Then, the element $x$ is written as
	\begin{align}\label{def:Nr}
		x=N(r;\{E_k\}):=-rE_1+(1+r)E_2+\frac{1}{2}\sum_{k=3}^{d^2} E_k,
	\end{align}
	$\{E_k\}\in\cP$
	and $r>0$.
	Here, we remark that any $E_k$ is a maximaly entangled state.
	Therefore, any separable pure state $y$ satisfies the following inequality:
	\begin{align}
		\Tr xy
		=&\Tr y\left(-rE_1+(1+r)E_2+\frac{1}{2}\sum_{j=3}^{d^2} E_j\right)\nonumber\\
		=&\Tr y\left(-\left(r+\frac{1}{2}\right)E_1+\left(r+\frac{1}{2}\right)E_2+\frac{1}{2}I\right)\nonumber\\
		\stackrel{(a)}{\ge}&\Tr y\left(-\left(r+\frac{1}{2}\right)E_1+\frac{1}{2}I\right)\nonumber\\
		\stackrel{(b)}{\ge}&-\cfrac{r}{\sqrt{d}}-\cfrac{1}{2\sqrt{d}}+\frac{1}{2}
		\stackrel{(c)}{\ge}-\cfrac{\sqrt{d}-1}{2\sqrt{d}}-\cfrac{1}{2\sqrt{d}}+\cfrac{1}{2}=0.
	\end{align}
	The inequality $(a)$ is shown by the inequalities $\Tr y E_1\le1$ and $\Tr y E_2\ge0$.
	The inequality $(b)$ is shown by the fact that the inequality $\Tr\sigma\rho\le(1/\sqrt{d})$ holds for any separable pure state $\sigma$ and any maximally entangled state $\rho$ \cite[Eq. (8.7)]{HayashiBook2017}.
	The inequality $(c)$ is shown by the assumption $0\le r \le (\sqrt{d}-1)/2$.
	Therefore, $\Tr xy\ge0$ holds,
	which implies that $x\in\mathrm{SEP}(A;B)^\ast$.
\end{proof}

\begin{lemma}\label{lem:con2}
	For given $\cH_A$ and $\cH_B$,
	define the dimension $d=\dim(\cH_A)\dim(\cH_B)$.
	Then, any two elements
	$x,y\in\mathrm{NPM}_r(\cP)$
	satisfy $\Tr xy\ge0$ if the parameter $r$ satisfies
	\begin{align}\label{ineq:d}
		0\le r\le \cfrac{\sqrt{2d}-2}{4}.
	\end{align}
\end{lemma}

\begin{proof}[Proof of lemma~\ref{lem:con2}]
	Take two arbitrary elements
	$x,y\in\mathrm{NPM}_r(\cP)$.
	By the definition \eqref{def:NPM},
	the two elements $x,y$ are written as
\begin{align}
\begin{aligned}
	x&=-rE_1+(1+r)E_2+\frac{1}{2}\sum_{k=3}^{d^2} E_k,\\
	y&=-rE_1'+(1+r)E_2'+\frac{1}{2}\sum_{l=3}^{d^2} E_l',
\end{aligned}
\end{align}
where
$\{E_k\},\{E_l'\}\in\cP$.
Then, the following inequality holds:
\begin{align}
	&\Tr xy \nonumber \\
	=&\Tr\left(-rE_1+(1+r)E_2+\frac{1}{2}\sum_{k=3}^{d^2} E_k\right)\left(-rE_1'+(1+r)E_2'+\frac{1}{2}\sum_{l=3}^{d^2} E_l'\right)\nonumber\\
	=&\Tr\left(-rE_1+(1+r)E_2+\frac{1}{2}(I-E_1-E_2)\right) \nonumber\\
	& \rightline{$\displaystyle \left(-rE_1'+(1+r)E_2'+\frac{1}{2}(I-E_1'-E_2')\right)$}\nonumber\\
	=&\Tr\left(-\left(r+\frac{1}{2}\right)E_1+\left(r+\frac{1}{2}\right)E_2+\frac{1}{2}I\right) \nonumber\\
	&\rightline{$\displaystyle \left(-\left(r+\frac{1}{2}\right)E_1'+\left(r+\frac{1}{2}\right)E_2'+\frac{1}{2}I\right)$}\nonumber\\
	=&\Tr\Biggl( \left(r+\frac{1}{2}\right)^2E_1E_1'-\left(r+\frac{1}{2}\right)^2E_1E_2'-\frac{1}{2}\left(r+\frac{1}{2}\right)E_1-\left(r+\frac{1}{2}\right)^2E_2E_1'\nonumber\\
	&\rightline{$\displaystyle +\left(r+\frac{1}{2}\right)^2E_2E_2'+\frac{1}{2}\left(r+\frac{1}{2}\right)E_2-\frac{1}{2}\left(r+\frac{1}{2}\right)E_1'+\frac{1}{2}\left(r+\frac{1}{2}\right)E_2'+\frac{1}{4}d \Biggr)$}\nonumber\\
	\stackrel{(a)}{\ge}&-\left(r+\frac{1}{2}\right)^2-\frac{1}{2}\left(r+\frac{1}{2}\right)-\left(r+\frac{1}{2}\right)^2+\left(r+\frac{1}{2}\right)^2\nonumber\\
	&\rightline{$\displaystyle +\frac{1}{2}\left(r+\frac{1}{2}\right)-\frac{1}{2}\left(r+\frac{1}{2}\right)+\frac{1}{2}\left(r+\frac{1}{2}\right)+\frac{1}{4}d$}\nonumber\\
	=&-2\left(r+\frac{1}{2}\right)^2+\frac{1}{4}d
	\stackrel{(b)}{\ge}-2\left(\cfrac{\sqrt{2d}-2}{4}+\frac{1}{2}\right)^2+\frac{1}{4}d\nonumber\\
	=&-2\left(\cfrac{\sqrt{2d}}{4}\right)^2+\frac{1}{4}d=-\frac{4d}{16}+\frac{1}{4}d=0.\label{eq:lem:con2}
\end{align}
	The inequality $(a)$ is shown by $E_1E_1'\ge0$, $E_1E_2'\le1$ and so on.
	The inequality $(b)$ is shown by the assumption \eqref{ineq:d} of lemma~\ref{lem:con2}.
	The inequality \eqref{eq:lem:con2} is the desired inequality.
\end{proof}

\begin{proof}[Proof of Proposition~\ref{prop:construction1}]
	We remark the following inequality:
	\begin{align}
		\cfrac{\sqrt{2d}-2}{4n}
		&\le \cfrac{\sqrt{2d}-2}{4}
		\le \cfrac{2\sqrt{d}-2}{4}=\cfrac{\sqrt{d}-1}{2}.
	\end{align}
	Therefore, we apply lemma~\ref{lem:con1} and lemma~\ref{lem:con2} to
	$\cC_r(\cP)$
	with $r\le r_0(A;B)$.
	
	First, we show pre-duality of $\cC_r(\cP)$ for $r\le r_0(A;B)$,
	i.e., any two elements $x,y\in\cC_r^\ast(\cP)$ satisfy $\Tr xy\ge0$.
	Take two elements $x,y\in\cC_r^\ast(\cP)$,
	and we need to show $\Tr xy\ge0$.
	Because of the definition~\ref{def:Kr},
	the elements $x,y$ are written as $x=x_1+x_2$ ,$y=y_1+y_2$ for $x_1,y_1\in\cC_r^{(0)\ast}(\cP)$, $x_2,y_2\in\mathrm{NPM}_r(\cP)$.
	By lemma~\ref{lem:con2},
	the inequality $\Tr x_2y_2\ge0$ holds.
	Because $\mathrm{SES}(A;B)\supset\cC_r^{(0)}(\cP)$ holds,
	$\cC_r^{(0)\ast}(\cP)$ is pre-dual,
	and therefore, the inequality $\Tr x_1y_1\ge0$ holds.
	Because 
	$\cC_r^{(0)\ast}(\cP)\subset\mathrm{NPM}_r(\cP)$ holds,
	the inequalities $\Tr x_1y_2\ge0$ and $\Tr y_1x_2\ge0$ hold.
	As a result,
	we obtain $\Tr xy\ge0$, which implies that $\cC_r^\ast(\cP)$ is pre-dual.
	
	Next, we show the exact inclusion relation $\cC_{r_2}(\cP)\subsetneq\cC_{r_1}(\cP)$,
	which is shown by $\mathrm{NPM}_{r_2}(\cP)\subsetneq\mathrm{NPM}_{r_1}(\cP)$ holds.
	Finally, $\cC_r(\cP)$ satisfies \eqref{eq:quantum} because of the definition \eqref{def:Kr} and lemma~\ref{lem:con1}.
\end{proof}

\subsection{Proof of Proposition~\ref{prop:construction2}}\label{append-max-ent}

For the proof of Proposition~\ref{prop:construction2},
we define a function $F_{\mathrm{max}}$ by fidelity $F(\rho,\sigma)$ of two states $\rho,\sigma$ as
	\begin{align}
		F_{\mathrm{max}}(\rho):&=\max_{\sigma\in\mathrm{ME}(A;B)}F(\rho,\sigma)\stackrel{(a)}{=}\max_{\sigma\in\mathrm{ME}(A;B)}\Tr \rho\sigma.\label{eq:fid-tra}
	\end{align}
The equality $(a)$ holds because any maximally entangled state is pure.
Also, we remark the relation between trace norm and fidelity.
The following inequality holds for any state $\rho,\sigma\in\mathrm{SES}(A;B)$:
\begin{align}\label{eq:F1}
	\|\rho-\sigma\|_1\le 2\sqrt{1-F(\rho,\sigma)}.
\end{align}
In order to show Proposition~\ref{prop:construction2},
we give the following lemma.
\begin{lemma}\label{lem:max-ent}
When a state $\rho\in\mathrm{SES}$ and a parameter $r$ satisfy the inequality
	\begin{align}\label{eq:index2}
		F_{\mathrm{max}}(\rho)\le\cfrac{1}{2r+1},
	\end{align}
we have 
	\begin{align}
		\rho\in\cC_r^{(0)\ast}(\cP).
	\end{align}
	\end{lemma}

\begin{proof}[Proof of lemma~\ref{lem:max-ent}]
We choose a state $\rho\in\mathrm{SES}$ and a parameter $r$ to satisfy the inequality \eqref{eq:index2}.
In the following, we show $\rho\in\cC_r^{(0)\ast}(\cP)$.

Any element of $\cC_r^{(0)}(\cP)=\mathrm{SES}(A;B)+\mathrm{NPM}_r(\cP)$ is written as 
$\sigma+N(\lambda;\{E_k\})$ with 
$\sigma \in \mathrm{SES}(A;B)$ and $N(\lambda;\{E_k\})\in \mathrm{NPM}_r(\cP)$ given in \eqref{def:Nr}.
As $\rho \in \mathrm{SES}(A;B)$, we have
\begin{align}
\Tr \rho \sigma \ge 0 \label{HH1}.
\end{align}

	Since the element $N(\lambda;\{E_k\})\in\mathrm{NPM}_r(\cP)$ is written as the following form by $\{E_k\}\in(\cP)$ and $0\le \lambda\le r$
	\begin{align*}
		N(\lambda;\{E_k\})=-\lambda E_1+(1+\lambda)E_2+\frac{1}{2}\sum_{k=3}^{d^2} E_k,
	\end{align*}
	we obtain the following inequality
	by using \eqref{eq:index2}; 
	\begin{align}
		&\Tr\rho N(\lambda;\{E_k\})\nonumber\\
		=&\Tr\rho\left(-\lambda E_1+(1+\lambda)E_2+\frac{1}{2}(I-E_1-E_2)\right)\nonumber\\
		=&\Tr\rho\left(-\left(\lambda+\frac{1}{2}\right) E_1+\left(\lambda+\frac{1}{2}\right)E_2+\frac{1}{2}I\right)
		\stackrel{(a)}{\ge}\Tr\left(-\left(\lambda+\frac{1}{2}\right)\rho E_1+\frac{1}{2}\rho I\right)\nonumber\\
		\stackrel{(b)}{\ge}&-\left(\lambda+\frac{1}{2}\right)\cfrac{1}{2r+1}+\frac{1}{2}
		\stackrel{(c)}{\ge}-\left(r+\frac{1}{2}\right)\cfrac{1}{2r+1}+\frac{1}{2}=0.\label{HH2}
	\end{align}
	The inequality $(a)$ is shown by $\Tr\rho E_2\ge0$.
	The inequality $(b)$ holds because $E_1$ is a maximally entangled state and because the equations \eqref{eq:fid-tra}, \eqref{eq:index2} hold.
	The inequality $(c)$ is shown by $\lambda\le r$.
Therefore, combining \eqref{HH1} and \eqref{HH2},
we obtain
\begin{align}
\Tr \rho (\sigma+N(\lambda;\{E_k\})) \ge 0,
\end{align}
which implies
the relation $\rho\in\cC_r^{(0)\ast}(\cP)$.
\end{proof}

By using Lemma~\ref{lem:max-ent},
we prove Proposition~\ref{prop:construction2}.

\begin{proof}[Proof of Proposition~\ref{prop:construction2}]
	\textbf{[OUTLINE]}
	First, as STEP1, we simplify the minimization of $D(\tilde{\cC_r}(A;B))$.
	Next, as STEP2, we estimate the simplified minimization.
	Finally, as STEP3,
	combining STEP1 and STEP2,
	we derive \eqref{eq:est-distance}.\\
	
	\textbf{[STEP1]}
	Simplification of the minimization.
	
	Because the inclusion relations
	\begin{align}
		\tilde{\cC_r}(\cP)\supset\cC_r^\ast(\cP)\supset\cC_r^{(0)\ast}(\cP)
	\end{align}
	hold,
	the following inequality holds for any $\sigma\in\mathrm{ME}(A;B)$:
	\begin{align}
		D(\tilde{\cC_r}(\cP) \| \sigma)
		=&\min_{\rho\in\tilde{\cC_r}(\cP)}\|\rho-\sigma\|_1
		\le \min_{\rho\in\cC_r^{(0)\ast}(\cP)}\|\rho-\sigma\|_1.\label{eq:prop:const2-1}
	\end{align}
	
	\textbf{[STEP2]}
	Estimation of the minimization.
	
	Given an arbitrary maximally entangled state $\sigma$,
	take an element $\rho_0\in\mathrm{SES}(A;B)$ satisfying the following equality:
	\begin{align}
		F(\rho_0,\sigma)=\cfrac{1}{2r+1}.
	\end{align}
	lemma~\ref{lem:max-ent} implies the relation $\rho_0\in\cC_r^{(0)\ast}(\cP)$.
	Then, we obtain the following inequality:
	\begin{align}
		&\min_{\rho\in\cC_r^{(0)\ast}(\cP)}\|\rho-\sigma\|_1
		\le\|\rho_0-\sigma\|_1\nonumber\\
		\stackrel{(a)}{\le} &2\sqrt{1-F(\rho_0,\sigma)}
		=2\sqrt{1-\cfrac{1}{2r+1}}
		=2\sqrt{\cfrac{2r}{2r+1}}.\label{eq:prop:const2-2}
	\end{align}
	The inequality $(a)$ is shown by the inequality \eqref{eq:F1}.
	\\
	
	\textbf{[STEP3]}
	Combination of STEP1 and STEP2.
	
	Because $\sigma$ is an arbitrary element in $\mathrm{ME}(A;B)$,
	the following inequality holds:
	\begin{align}
		&D(\tilde{\cC_r}(\cP))
		=\max_{\sigma\in\mathrm{ME}(A;B)}D(\tilde{\cC_r}(\cP) \| \sigma)\nonumber\\
		\stackrel{(a)}{\le} &\max_{\sigma\in\mathrm{ME}(A;B)}\min_{\rho\in\cC_r^{(0)\ast}(\cP)}\|\rho-\sigma\|_1
		\stackrel{(b)}{\le}2\sqrt{\cfrac{2r}{2r+1}}\label{eq:d-e}.
	\end{align}
	The inequality $(a)$ is shown by \eqref{eq:prop:const2-1}.
	The inequality $(b)$ holds because the inequality \eqref{eq:prop:const2-2} holds for any $\sigma\in\mathrm{ME}(A;B)$.
	Hence, we obtain \eqref{eq:est-distance}.
\end{proof}

\subsection{Proof of inequality \eqref{eq:orthogonal2}}\label{append-4-2}

Here,
we show the inequality \eqref{eq:orthogonal2}.
In other words,
we show the following Proposition.
\begin{proposition}\label{prop:dist2}
	Let $\epsilon>0$ and $0<r\le r_0(A;B)$ be parameters satisfying $\epsilon=2\sqrt{(2r)/(2r+1)}$,
	and let $\rho_1,\rho_2$ be states satisfying \eqref{eq:orthogonal1}.
	Then, the following inequality holds;
	\begin{align}
		\Tr\rho_1\rho_2\ge\cfrac{\epsilon^2(\epsilon^2+8)}{32}.
	\end{align}
\end{proposition}

\begin{proof}[Proof of Proposition~\ref{prop:dist2}]
	The equation $\epsilon=2\sqrt{(2r)/(2r+1)}$ is reduced as follows:
	\begin{align}
		\epsilon=&2\sqrt{(2r)/(2r+1)}\nonumber\\
		\epsilon^2/4=&2r/(2r+1)\nonumber\\
		\epsilon^2/4=&1-1/(2r+1)\nonumber\\
		1/(2r+1)=&(4-\epsilon^2)/4\nonumber\\
		r=&\epsilon^2/(2(4-\epsilon^2)).\label{eq:r-e}
	\end{align}
	Then, \eqref{eq:r-e} implies the following equation:
	\begin{align}
		\Tr \rho_1\rho_2 \stackrel{(a)}{\ge}& \cfrac{2r(r+1)}{(2r+1)^2}
		=\cfrac{2\epsilon^2}{2(4-\epsilon^2)}\cdot\cfrac{\epsilon^2+(2(4-\epsilon^2)}{2(4-\epsilon^2)}\cdot\left(\cfrac{4-\epsilon^2}{4}\right)^2
		=\cfrac{\epsilon^2(\epsilon^2+8)}{32}.
	\end{align}
	Here, we remark that the inequality $(a)$ is \eqref{eq:orthogonal1}.
\end{proof}

\chapter{Conclusion}\label{chap:5}

\section{Summary}\label{sect:5-2}

This thesis has investigated the diversity of entanglement structures in general probabilistic theories.
In Chapter~\ref{chap:2},
we have introduced positive cones and GPTs.
Also, we have introduced ESs and its non-uniqueness,
and we have given important examples of ESs (in Appendix~\ref{appe:2}).
In Chapter~\ref{chap:3} and Chapter~\ref{chap:4},
this thesis has addressed the following five themes.
\begin{enumerate}
	\item[A.] Characterization of Dual-Operator-Valued Measurement
	\item[B.] Non-Simulability of Beyond Quantum Measurement
	\item[C.] Entanglement Structures with Group Symmetry
	\item[D.] Self-Dual Modification
	\item[E.] Existence of PSES and Difference from the SES
\end{enumerate}

In Chapter~\ref{chap:3},
we have discussed state discrimination tasks in ESs,
and have investigated Theme A and B.
We have classified DOVMs and have characterized the classes by the performance for discrimination tasks (Theorem~\ref{theorem:bq} and Theorem~\ref{theorem:aq}).
Besides the results,
as an application of Theorem~\ref{theorem:aq},
this thesis has given a derivation of the SES from ESs with a inclusion relation (Theorem~\ref{theorem:ses-dis}).
Also, as an application of Theorem~\ref{theorem:bq},
this thesis has discussed simulability of BQ measurement (Theme B),
and this thesis has shown non-simulability of BQ measurement (Theorem~\ref{theorem:non-sim-1}).

In Chapter~\ref{chap:4},
we have discussed ESs with self-duality and group symmetry.
As a result of Theme C,
we have shown that an ES with symmetric cone is uniquely determined as the SES (Theorem~\ref{prop:sym1}).
Also, we have revealed that an ES with global unitary symmetry is limited to the SES (Theorem~\ref{prop:global2}).
Next,
this thesis has given a general theory about self-duality for the investigation of ESs with self-duality.
We have shown that any pre-dual cone can be modified to a self-dual cone (Theorem~\ref{theorem:sd})
and that an exact hierarchy of pre-dual cones corresponds to an independent family of self-dual cones (Theorem~\ref{theorem:hie1}).
Applying this general theory,
we have shown that infinite existence of $\epsilon$-PSESs for any $\epsilon>0$ (Theorem~\ref{theorem:main}).
Moreover,
as the operational difference between the SES and PSESs,
we have shown that there exist $\epsilon$-PSESs with non-orthogonal perfect discrimination for any $\epsilon>0$ (Theorem~\ref{theorem:dist}).

In this way,
we have investigated the diversity of ESs,
i.e.,
possible structure of quantum composite systems in GPTs.
This thesis has clarified that there are several types of ESs and some of them are similar to the SES but different from the SES.

\section{Open Problems}\label{sect:5-2}

Finally, we enumerate open problems in this thesis.

In Theme A,
as Theorem~\ref{theorem:ses-dis},
we have given a characterization of the SES by the condition $\mathrm{Err}_{\cC}(\rho_1;\rho_2)=1-\frac{1}{2}\|\rho_1-\rho_2\|_1$
when we impose an additional condition $\cC\subset\mathrm{SES}(A;B)$.
No counterexample is known when we relax the condition $\cC\subset\mathrm{SES}(A;B)$.
Therefore,
it is an open problem whether the condition $\mathrm{Err}_{\cC}(\rho_1;\rho_2)=1-\frac{1}{2}\|\rho_1-\rho_2\|_1$
characterizes the SES without the assumption of the condition $\cC\subset\mathrm{SES}(A;B)$.

In Theme B,
as Theorem~\ref{theorem:non-sim-1},
we have shown the non-simulability of BQ measurement.
The proof of Theorem~\ref{theorem:non-sim-1} depends on the extraordinary performance of BQ measurement.
Similarly,
it is considered that AQ measurement might be non-simulable
because AQ measurement also has an extraordinary performance.
This is an open problem.

In Theme C,
as Theorem~\ref{prop:sym2},
we have shown that a $\mathrm{GU}(A;B)$-symmetric ES is uniquely determined as the SES.
However, $\mathrm{GU}(A;B)$-symmetry is not derived reasonably from local structures.
On the other hand,
the symmetry of the local unitary group $\mathrm{LU}(A;B)$,
defined as
\begin{align}
\begin{aligned}
	\mathrm{LU}(A;B):=\{g\in\mathrm{GL}(\cT(\cH_A\otimes\cH_B))&\mid g(\cdot):=(U_A^\dag\otimes U_B^\dag) (\cdot) (U_A\otimes U_B)\\
	 U_A,&U_B \ \mbox{are unitary matrices on $\cH_A,\cH_B$}\},\label{eq:lu}
\end{aligned}
\end{align}
is naturally derived from local structures.
Therefore,
an important problem is variety of entanglement structures with $\mathrm{LU}(A;B)$-symmetry.
Here, we give the following two important examples:
\begin{enumerate}[(EI)]
	\item $\Gamma(\mathrm{SES}(A;B))$ (where $\Gamma$ is the partial transposition map that transposes Bob's system)
	\item $\cC_r^\ast(\cP)$ (where $\cP$ is an $\mathrm{LU}(A;B)$-symetric subset of $\mathrm{MEOP}(A;B)$)
\end{enumerate}
These two examples satisfy two of three conditions,
$\mathrm{LU}(A;B)$-symmetry, self-duality, and $\epsilon$-undistinguishablity.
On the other hand, no known example satisfies the above three conditions except for $\mathrm{SES}(A;B)$.
Therefore, it remains open whether there exists a model that satisfies these three conditions and that is different from SES.

In Theme D,
as Theorem~\ref{theorem:sd},
we have shown the existence of a self-dual modification.
However, because our proof depends on Zorn's Lemma,
we have not given a self-dual cone explicitly.
Also, the reference \cite{BF1976} does not give constructive self-dual cones.
It is an open problem to give a self-dual modification by an explicit form.

In Theme E,
as Theorem~\ref{theorem:dist},
we have shown that some types of PSESs have an extraordinary performance for discrimination tasks.
This result implies the possibility that orthogonal discrimination can characterize the standard entanglement structure rather than self-duality.
In other words, we propose the following conjecture as a considerable statement, which is a future work.
\begin{conjecture}\label{conj1}
	If a model of the quantum composite system $\cC$ is not equivalent to the SES,
	$\cC$ has a pair of two non-orthogonal states discriminated perfectly by a measurement in $\cC$.
\end{conjecture}

\renewcommand{\bibname}{References}

\appendix

\chapter{Examples of Models of GPTs}

\section{Examples of ESs with Important Properties}\label{appe:2}

In this section,
we give ESs that are counterexamples for some important mathematical properties.
In Section~\ref{appe:2-1},
we show that the model $\mathrm{SEP}(A;B)$ is a typical example that does not satisfy entropy preserving spectrality.
In Section~\ref{appe:2-2},
we give an ES that satisfies 1-symmetry but does not satisfy 2-symmetry.

\subsection{Counterexample of Entropy Preserving Spectrality}\label{appe:2-1}

First, we define entropy preserving spectrality.
\begin{definition}[entropy preserving spectrality \cite{MAB2022}]\label{def:ent}
	Given a state $\rho$,
	we say that $\rho$ has entropy non-preserving spectral decompositions
	if there exist two decompositions of $\rho$ over pairs of perfectly distinguishable pure state $\{\rho_i\}$ and $\{\sigma_j\}$ as
	\begin{align}
		\rho=\sum_i p_i\rho_i=\sum_j q_j \sigma_j
	\end{align}
	satisfying the relation
	\begin{align}
		-\sum_i p_i\log p_i\neq-\sum_j q_j\log q_j.
	\end{align}
	where $(p_i)$ and $(q_j)$ are probability vectors.
	
	We then say that an ES $\cC$ satisfies entropy preserving spectrality
	if any state $\rho\in\cS(\cC,I)$ does not have entropy non-preserving spectral decompositions.
\end{definition}
From the viewpoint of physics,
entropy preserving spectrality means the consistency with thermodynamics.
The SES satisfies Definition~\ref{def:ent}.
On the other hand, there exists an ES that does not satisfy Definition~\ref{def:ent}.
The ES is constructed as follows.

At first,
Theorem~\ref{preceding-1} shows that
the following two separable states are perfectly distinguishable in $\mathrm{SEP}(A;B)$:
\begin{align}
    \rho_1
    &=
    \begin{bmatrix}
        1 & 0\\
        0 & 0
    \end{bmatrix}
    \otimes
    \begin{bmatrix}
        1 & 0\\
        0 & 0
    \end{bmatrix},\\
    \rho_2
    &=
    \cfrac{1}{2}\begin{bmatrix}
        1 & 1\\
        1 & 1
    \end{bmatrix}
    \otimes
    \cfrac{1}{2}\begin{bmatrix}
        1 & 1\\
        1 & 1
    \end{bmatrix}.
\end{align}
Also, the reference \cite{Arai2019} gives the following measurement $\{e_1,e_2\}$ that discriminate $\{\rho_1,\rho_2\}$ perfectly:
\begin{align}\label{eq:meas-1}
    e_1(\rho)
    &=\Tr \left\{
    \cfrac{1}{2}
    \begin{bmatrix}
    	2 & 0 & 0 & -1\\
    	0 & 0 & -1 & 0\\
    	0 & -1 & 0 & 0\\
    	-1 & 0 & 0 & 2
    \end{bmatrix}
    \rho
    \right\},\\
    e_2(\rho)\label{eq:meas-2}
    &=\Tr \left\{
    \cfrac{1}{2}
    \begin{bmatrix}
    	0 & 0 & 0 & 1\\
    	0 & 2 & 1 & 0\\
    	0 & 1 & 2 & 0\\
    	1 & 0 & 0 & 0
    \end{bmatrix}
    \rho
    \right\}.
\end{align}
Next,
we extend $\mathrm{SEP}(A;B)$ slightly.
Consider the following density matrices with unit rank:
\begin{align}
	\sigma_1
	&=\cfrac{1}{6}
	\begin{bmatrix}
    	3 & \sqrt{3} & \sqrt{3} & \sqrt{3}\\
    	\sqrt{3} & 1 & 1 & 1\\
    	\sqrt{3} & 1 & 1 & 1\\
    	\sqrt{3} & 1 & 1 & 1
    \end{bmatrix},\\
    \sigma_2
    &=\cfrac{1}{6}
    \begin{bmatrix}
    	3 & -\sqrt{3} & -\sqrt{3} & -\sqrt{3}\\
    	-\sqrt{3} & 1 & 1 & 1\\
    	-\sqrt{3} & 1 & 1 & 1\\
    	-\sqrt{3} & 1 & 1 & 1
    \end{bmatrix}.
\end{align}
Because $\sigma_1$ and $\sigma_2$ are not separable,
$\sigma_1,\sigma_2\not\in \mathrm{SEP}(A;B)$.
Then consider the following ES $\cC$:
\begin{align}
	\cC:=
	\mathrm{Hul}\left(\mathrm{SEP}(A;B)\cup\{\sigma_1,\sigma_2\}\right).
\end{align}
We remark that $\rho_i$ and $\sigma_j$ are pure because they are rank 1 matrices.

Because the inclusion relation $\mathrm{SEP}(A;B)\subset\cC\subset\mathrm{SES}(A;B)$,
the inclusion relation $\cM(\mathrm{SES}(A;B),I)\subset\cM(\cC,I)\subset\cM(\mathrm{SEP}(A;B),I)$ holds.
In particular,
because $e_j$ given in in Eqs.~\eqref{eq:meas-1} and~\eqref{eq:meas-2} satisfies $e_j(\sigma_i)\ge0$ for all $i,j$,
the measurement $\{e_1,e_2\}$ belongs to $\cM(\cC,I)$.
Because the two states $\sigma_1,\sigma_2$ are orthogonal quantum states, they are perfectly distinguished by a measurement in $\cM(\mathrm{SES}(A;B),I)$.
Therefore, the states $\sigma_1,\sigma_2$ are perfectly distinguishable in $\cC$.
This implies that the state $\rho:=\cfrac{1}{3}\rho_1+\cfrac{2}{3}\rho_2$ can be decomposed into perfectly distinguishable pure states in two different ways, as follows:
\begin{align}
	\rho&=\frac{1}{3}\rho_1+\frac{2}{3}\rho_2,\\
	&=\frac{3+\sqrt{3}}{6}\sigma_1+\frac{3-\sqrt{3}}{6}\sigma_2\;,
\end{align}
which clearly possess two different values of entropy.

\subsection{Example of 1-Symmetry and not 2-Symmetry}\label{appe:2-2}

Next, we give an example of ES that satisfies 1-symmetry but does not satisfy 2-symmetry.
First, we define $k$-symmetry.
\begin{definition}[$k$-symmetry]
	We say that a model $\cC$ is $k$-symmetric if there exists a transformation
	$f\in\cT(\cC,u)$ such that $\rho_i = f(\sigma_i)$, for $i = 1,\cdots, k$, for any pair of $k$-tuples of perfectly distinguishable pure states $\{\rho_i\}_{i=1}^k$ and $\{\sigma_i\}_{i=1}^k$.
\end{definition}

Now we show the difference between strong symmetry and weak symmetry by giving the following counterexample.
\begin{theorem}\label{theorem:k-sym}
	$\mathrm{SEP}(A;B)$ is $1$-symmetric but not $2$-symmetric.
\end{theorem}

In order to show this theorem,
we apply the following lemma \cite[Theorem 3]{Friedland2011}.
\begin{lemma}\label{lemma:sep_transf}
	For the linear map $f$ from $\Her{\cH_{A}\otimes\cH_{B}}\to\Her{\cH_{A}\otimes\cH_{B}}$, the following are equivalent:
	\begin{enumerate}[(i)]
		\item $f\in\cT(\mathrm{SEP(A;B)})$.
		\item $f(\mathrm{Ext}(\Psd{\cH_A})\otimes\mathrm{Ext}(\Psd{\cH_B}))=\mathrm{Ext}(\Psd{\cH_A})\otimes\mathrm{Ext}(\Psd{\cH_B})$.
		\item $f(X\otimes Y)=f_{A}(X)\otimes f_{B}(Y)$, or $\dim\cH_{A}=\dim\cH_{B}$ and $f(X\otimes Y)=f_{B}(Y)\otimes f_{A}(X)$, where $f_{A}(X)=U_{A} X U_{A}^\dag$ or $U_{A} X^{\top}U_{A}^\dag$ and $f_{B}(Y)=V_{B}YV_{B}^\dag$ or $V_{B}Y^{\top}V_{B}^\dag$.
	\end{enumerate}
\end{lemma}
\begin{proof}[Proof of Theorem \ref{theorem:k-sym}]
	Since $\cF(\mathrm{SEP(A;B)})$ contains all local unitary maps, $\mathrm{SEP(A;B)}$ clearly satisfies $1$-symmetry.
	
	Now, we show that $\mathrm{SEP(A;B)}$ is not $2$-symmetric by giving a counterexample.
	Take the following four separable pure states:
	\begin{align}
		\rho_1&=\rho_1^{\mathrm{A}}\otimes\rho_1^{\mathrm{B}}=
		\begin{bmatrix}
			1 & 0\\
			0 & 0
		\end{bmatrix}\otimes
		\begin{bmatrix}
			1 & 0\\
			0 & 0
		\end{bmatrix},\label{eq:ex_sep1}\\
		\rho_2&=\frac{1}{2}
		\begin{bmatrix}
			1 & 1\\
			1 & 1
		\end{bmatrix}\otimes\frac{1}{2}
		\begin{bmatrix}
			1 & 1\\
			1 & 1
		\end{bmatrix},\label{eq:ex_sep2}\\
		\sigma_1&=
		\begin{bmatrix}
			1 & 0\\
			0 & 0
		\end{bmatrix}\otimes
		\begin{bmatrix}
			1 & 0\\
			0 & 0
		\end{bmatrix},\label{eq:ex_sep3}\\
		\sigma_2&={}
		\begin{bmatrix}
			0 & 0\\
			0 & 1
		\end{bmatrix}\otimes
		\begin{bmatrix}
			0 & 0\\
			0 & 1
		\end{bmatrix}\label{eq:ex_sep4}.
	\end{align}
	By direct inspection, we can verify that the two dichotomies $\{\rho_1,\rho_2\}$ and $\{\sigma_1,\sigma_2\}$ both satisfy condition~\eqref{eq:sep-dist} and thus, by Theorem \ref{preceding-1}, both contain perfectly distinguishable pure states in $\mathrm{SEP(A;B)}$.
	
	Assume that there is a map $f\in\cT(\mathrm{SEP}(A;B))$ where $\sigma_1=f(\rho_1)$ and $\sigma_2=f(\rho_2)$.
	From Lemma \ref{lemma:sep_transf}, the following equality should hold:
	\begin{equation}
		\begin{split}
			\Tr\{\sigma_1\sigma_2\}&=\Tr\{f_{A}(\rho_1^{A})f_{A}(\rho_2^{A})\otimes f_{B}(\rho_1^{B})f_{B}(\rho_2^{B})\}\\
			&=\Tr\{\rho_1^{A}\rho_2^{A}\otimes\rho_1^{B}\rho_2^{B}\}\\
			&=\Tr\{\rho_1\rho_2\}.\label{eq:sep_tr}
		\end{split}
	\end{equation}
	However, now we have
	\begin{equation}
		\Tr\{\rho_1\rho_2\}=\frac{1}{4},\quad \Tr\{\sigma_1\sigma_2\}=0.
	\end{equation}
	This contradicts \eqref{eq:sep_tr}.
	Thus $\mathrm{SEP}(A;B)$ is $1$-symmetric but not $2$-symmetric.
\end{proof}

\section{Example of General Models with Non-Simulable Measurements}\label{appe:2-3}

In this section,
we address a general model of GPTs that is not an ES.
As mentioned in Section~\ref{sect:3-2-3},
general models sometimes contain non-simulable measurements,
and here, we give an example of such models, called \textit{shrunk Bloch sphere}.

Let us consider the set $\Her{\cH}$ with $\dim\cH=2$.
For a parameter $p\in(0,1)$, we define a cone $\cC_p$ as
\begin{align}
	\cC_p:=\left\{p\rho+\frac{1-p}{2}I\ \middle| \ \rho\in\Psd{\cH}\right\},
\end{align}
and we consider the model $(\Her{\cH},\Tr,\cC_p,I)$.
Now, we show that the model contains a non-simulable measurement as follows.

First, we see that for any orthonormal basis $\vec{P}=(P_1,P_2)$ on $\Her{\cH}$, a measurement $\bm{M}(\vec{P}):=\{M(\vec{P}),I-M(\vec{P})\}$
defined as
\begin{align}
	M(\vec{P}):=-\frac{1-p}{2p}P_1+\frac{1+p}{2p}P_2
\end{align}
belongs to $\cM(\cC_p,I)$.
The relation $\bm{M}(\vec{P})\in\cM(\cC_p,I)$ is shown by the inequalities $\Tr \sigma M(\vec{P})\ge0$ and $\Tr \sigma_p \left(I-M(\vec{P})\right)\ge0$ for any $\sigma_p\in\cS(\cC_p,I)$ written as $\sigma_p=p\sigma+\frac{1-p}{2}I$ with $\sigma\in\Psd{\cH}$,
which are shown as follows:
\begin{align}
	&\Tr \sigma M(\vec{P})=\Tr \left(p\sigma+\frac{1-p}{2}I\right) \left(-\frac{1-p}{2p}P_1+\frac{1+p}{2p}P_2\right)\\
	=&-\frac{(1-p)}{2}\Tr \sigma P_1+\frac{(1-p)}{2}\Tr \sigma P_2-\frac{(1-p)^2}{4p}+\frac{(1+p)(1-p)}{4p}\\
	\stackrel{(a)}{\ge} &-\frac{(1-p)}{2}-\frac{(1-p)^2}{4p}+\frac{(1+p)(1-p)}{4p}
	=0.
\end{align}
The inequality $(a)$ is shown by the inequalities $\Tr \sigma P_1\le1$ and $\Tr \sigma P_2\ge0$.
Here, we remark that the equality of $(a)$ is attained with $\sigma_p=\rho_1:=pP_1+\frac{1-p}{2}I$,
which implies the equations $\Tr \rho_1M(\vec{P})=0$ and $\Tr \rho_1\left(I-M(\vec{P})\right)=1$.
The other inequality $\Tr \sigma_p \left(I-M(\vec{P})\right)\ge0$ is shown similarly,
and the state $\sigma_p=\rho_2:=pP_2+\frac{1-p}{2}I$ satisfies the equations $\Tr \rho_2M(\vec{P})=1$ and $\Tr \rho_2\left(I-M(\vec{P})\right)=0$.
In this way,
the measurement $\bm{M}(\vec{P})$ belongs to $\cM(\cC_p,I)$,
and moreover, the measurement $\bm{M}(\vec{P})$ discriminates two states $\rho_1$ and $\rho_2$.
Besides,
the two states $\rho_1$ and $\rho_2$ satisfy the inequality $\Tr \rho_1\rho_2>0$ by their constructions.
In other words,
the measurement $\bm{M}(\vec{P})$ discriminates two non-orthogonal states.

Similarly to the discussion in Section~\ref{sect:3-2-3},
because the measurement $\bm{M}(\vec{P})$ discriminates two non-orthogonal states
any POVM $\bm{N}=\{N_1,N_2\}$ with $N_1,N_2\in\Psd{\cH^{\otimes n}}$ never satisfies the relation
\begin{align}
	M(\vec{P}) \rho=N_1\rho^{\otimes n}\quad \forall \rho\in\cD(\bm{M}(\vec{P}))\cap\cD(\bm{N})
\end{align}
for any natural number $n$.
In other words,
the measurement $\bm{M}(\vec{P})$ is not $n$-simulable for any natural number $n$.

In this way,
the model $(\Her{\cH},\Tr,\cC_p,I)$ contains a non-simulable measurement
even though the model is not an ES.

\end{document}